\documentclass[showpacs,twocolumn,pra,superscriptaddress,notitlepage]{revtex4-2}
\pdfoutput=1
\usepackage{qcircuit}
\usepackage[dvips]{graphicx}
\usepackage{amsmath,amssymb,amsthm,mathrsfs,amsfonts,dsfont,mathtools}
\usepackage{subfigure, epsfig}
\usepackage{braket}
\usepackage{bm}
\usepackage{enumerate}
\usepackage[dvipsnames,svgnames]{xcolor}
\usepackage[colorlinks=true, allcolors=MidnightBlue, linktoc = all]{hyperref}
\usepackage{newtxtext,newtxmath}
\usepackage{stackrel}

\newtheorem{theorem}{Theorem}
\newtheorem{lemma}{Lemma}
\newtheorem{proposition}{Proposition}

\newtheorem{corollary}{Corollary}
\newtheorem{definition}{Definition}

\newcommand{\ketbra}[2]{\vert #1 \rangle \langle #2 \vert}
\newcommand{\dm}[1]{\ketbra{#1}{#1}}
\newcommand{\mc}[1]{\mathcal{#1}}

\DeclareMathOperator{\tr}{Tr}

\newcommand{\yx}[1]{{\color{black} #1}}

\newcommand{\id}{I}

\newcommand{\bal}{\begin{equation}\begin{aligned}}
\newcommand{\eal}{\end{aligned}\end{equation}}
\newcommand{\sbar}{\;\rule{0pt}{9.5pt}\right|\;}
\newcommand{\lset}{\left\{\left.}
\newcommand{\rset}{\right\}}
\DeclareMathOperator{\Tr}{Tr}

\newcommand{\mA}{\mathcal{A}}

\newcommand{\mE}{\mathcal{E}}
\newcommand{\mN}{\mathcal{N}}

\newcommand{\mT}{\mathcal{T}}
\newcommand{\mD}{\mathcal{D}}
\newcommand{\mI}{\mathcal{I}}
\newcommand{\mF}{\mathcal{F}}
\newcommand{\mL}{\mathcal{L}}

\newcommand{\mM}{\mathcal{M}}
\newcommand{\mO}{\mathcal{O}}

\newcommand{\mS}{\mathcal{S}}

\newcommand{\mZ}{\mathcal{Z}}

\newcommand{\mbR}{\mathbb{R}}

\DeclareMathOperator*{\argmin}{arg\,min}

\usepackage{xr}
\makeatletter
\newcommand*{\addFileDependency}[1]{
  \typeout{(#1)}
  \@addtofilelist{#1}
  \IfFileExists{#1}{}{\typeout{No file #1.}}
}
\makeatother

\begin{document}

\title{Virtual quantum resource distillation: General framework and applications}

\author{Ryuji Takagi}
\email{ryuji.takagi@phys.c.u-tokyo.ac.jp}
\affiliation{Department of Basic Science, The University of Tokyo, Tokyo 153-8902, Japan}
\affiliation{Nanyang Quantum Hub, School of Physical and Mathematical Sciences, Nanyang Technological University, 637371, Singapore}

\author{Xiao Yuan}
\email{xiaoyuan@pku.edu.cn}
\affiliation{Center on Frontiers of Computing Studies, Peking University, Beijing 100871, China}
\affiliation{School of Computer Science, Peking University, Beijing 100871, China}
\affiliation{Stanford Institute for Theoretical Physics, Stanford University, Stanford California 94305, USA}

\author{Bartosz Regula}
\email{bartosz.regula@gmail.com}
\affiliation{Mathematical Quantum Information RIKEN Hakubi Research Team, RIKEN Cluster for Pioneering Research (CPR) and RIKEN Center for Quantum Computing (RQC), Wako, Saitama 351-0198, Japan}
\affiliation{Department of Physics, Graduate School of Science, The University of Tokyo, Bunkyo-ku, Tokyo 113-0033, Japan}

\author{Mile Gu}
\email{mgu@quantumcomplexity.org}
\affiliation{Nanyang Quantum Hub, School of Physical and Mathematical Sciences, Nanyang Technological University, 637371, Singapore}
\affiliation{Centre for Quantum Technologies, National University of Singapore, 3 Science Drive 2, 117543, Singapore}
\affiliation{CNRS-UNS-NUS-NTU International Joint Research Unit, UMI 3654, Singapore 117543, Singapore}

\begin{abstract}
  We develop the general framework of virtual resource distillation --- an alternative distillation strategy proposed in \href{https://journals.aps.org/prl/abstract/10.1103/PhysRevLett.132.050203}{Phys. Rev. Lett. \textbf{132}, 050203 (2024)}, which extends conventional quantum resource distillation by integrating the power of classical postprocessing. The framework presented here is applicable not only to quantum states, but also dynamical quantum objects such as quantum channels and higher-order processes. We provide a general characterization and benchmarks for the performance of virtual resource distillation in the form of computable semidefinite programs as well as several operationally motivated quantities. We apply our general framework to various concrete settings of interest, including standard resource theories such as entanglement, coherence, and magic, as well as settings involving dynamical resources such as quantum memory, quantum communication, and non-Markovian dynamics. The framework of probabilistic distillation is also discussed.
\end{abstract}

\maketitle


\section{Introduction}

The advantages of quantum algorithms and information processing are enabled by the efficient use of quantum resources such as quantum entanglement~\cite{Horodecki09} and superposition~\cite{RevModPhys.89.041003}.
However, it is usually difficult to prepare such quantum resources with high quality due to inevitable noise and imperfection. 
The standard approach to address this issue is through resource distillation, a class of protocols to prepare high-quality quantum resources from those of lower quality.

The feasibility and performance of resource distillation is a major topic of study in quantum information theory, often analyzed using the tools developed in quantum resource theories~\cite{chitambar2018quantum}, which are frameworks that deal with the quantification and manipulation of physical quantities that are costly to access in a given setting. Resource distillation has been studied in various resource theories, with the ultimate goal of producing output quantum objects, such as quantum states and channels, that are as close as possible to a desired target object. Although this goal is well motivated, as it allows for versatile use of the distilled resource object, it may be too restrictive depending on the objective of the overall quantum algorithm that utilizes the processed resource object after the distillation procedure.

Here, we observe that many quantum algorithms, such as variational quantum algorithms~\cite{Cerezo2021variational}, ultimately aim to obtain some classical output --- i.e., numerical values --- retrieved by measuring the expectation values of suitable observables. 
Such quantum algorithms do not strictly require distilling a desired quantum resource physically, as long as we can retrieve the expected values of all observables made on the quantum objective. 
This motivates us to propose a novel variant of resource distillation. 
Our distillation strategy does not directly distill a better quantum object, but instead fully utilizes the potential of classical postprocessing to `virtually' approximate its measurement statistics, allowing us to simulate the expectation values one would obtain if one were in physical possession of the target object.
We remark that several protocols known as virtual cooling~\cite{Cotler2019quantum} and virtual distillation in quantum error mitigation~\cite{Koczor2021exponential,Huggins2021virtual} share a similar idea that classical postprocessing enables us to simulate purer quantum states, although they differ from our framework --- these techniques extract the measurement statistics for purer quantum objects by coherently interacting multiple copies of noisy objects, while our virtual resource distillation applies a probabilistic operation to a single copy of the noisy quantum object, which is inspired by and related to the techniques for simulating unphysical objects~\cite{Buscemi2013twopoint,Jiang2021physical,Regula2021operational}, and error mitigation techniques based on quasiprobability~\cite{PhysRevLett.119.180509}.

This paper provides an extension and a rigorous theoretical foundation to our companion paper~\cite{maintext}, where the notion of virtual resource distillation is first introduced. Notably, although the discussion in~\cite{maintext} focuses on resource theories of quantum states, here we present extensive discussions of the fully general framework, which includes applications to dynamical resource theories of quantum channels and higher-order processes~\cite{dana2018resource,diaz2018using,rosset2018resource,bauml2019resource,dynamicalEntanglement,dynamicalCoherence,pirandola2017fundamental,faist_2018,theurer2018quantifying,takagi2019general,liu2019operational,liu2019resource,wang2019quantifying,PhysRevResearch.1.033169,takagi_2020,Kristjansson2020resource,Gour2019how,yuan_2020,Regula2021fundamental,Fang2020no-go,Takagi2021optimal,Berk2021resourcetheoriesof}, as well as the framework of probabilistic distillation~\cite{horodecki1999general,Jonathan1999minimal,Fang2018probabilistic,PhysRevLett.125.060405,regula_2022,regula_2021-4,Regula2021fundamental,Fang2020no-go}. 
We give detailed comparisons between the sampling overhead needed to realize virtual distillation protocols and conventional ones. 
We discuss applications of virtual resource distillation in various physical settings, providing an in-depth study of the distillation performance in each case.
Along the way, we include the proofs for Theorems~1 and 2 in the companion paper, which correspond to Theorems~\ref{prop:general_bounds_cost_parallel} and \ref{prop:cost twirling} in this article.



\section{Resource theories}

The restrictions imposed by a given physical settings can be usually represented by a limited set of quantum states and operations that one has access to.  
For instance, when two parties are physically separated apart and quantum communication is hard to establish, it is reasonable to study the scenario where they only have access to local quantum operations and classical communication (LOCC).  
In such a scenario, they can only generate separable states, and other states are costly ``resources'' that cannot be created for free, where entanglement serves as the resource quantity of interest.  
Central questions in such a scenario include (1) resource quantification: what is a good way of quantifying the underlying resource that we do not have free access to, and (2) resource manipulation: what are the resource transformations possible by only using the freely accessible operations. 
In general, the manipulated resources need not be quantum states, but can be e.g.\ quantum channels, measurements, or higher-order quantum operations~---~we will approach the problem generally by considering all such `resource objects' in a common formalism.

Resource theories are frameworks that provide a systematic approach to studying the above questions of quantum resource quantification and manipulation~\cite{chitambar2018quantum}.
The basic building blocks of the resource theory framework include a set $\mF$ of free objects ---  a subset of objects that can be prepared for free --- and a set $\mO$ of free operations --- the accessible operations that are allowed to transform the resource objects in the given setting. 
To reflect the physical constraints and the underlying quantum resource, we impose a basic condition on free operations; no free operation can create a resourceful object from a free object, i.e., if $\Lambda\in\mO$, then $\Lambda(X)\in\mF\ \forall X\in\mF$. 
The maximal such set is called \emph{resource non-generating operations}, examples of which include separability-preserving operations in entanglement theory and maximally incoherent operations in coherence theory, and any arbitrary set $\mO$ of free operations is then a subset of resource non-generating operations. 
With these concepts, we can formalize resource quantification by considering a function $R$ from objects to real numbers.
In particular, we call a function $R$ \emph{resource measure} or \emph{monotone} if (1) it always gives the smallest value for all free objects, i.e., for some constant $c$, $R(X)= c \ \forall X\in\mF$ and $R(X)\geq c\ \forall X$, and (2) it is monotonically nonincreasing under free operations, i.e., $R(X)\geq R(\Lambda(X))$ for every object $X$ and every free operation $\Lambda\in\mO$.

Depending on the specific setting of interest, one can flexibly select the set of resource objects to study. When we are interested in the manipulation of quantum states and the resources contained therein, the set of all quantum states are the relevant objects of study, which includes the designated set $\mF$ of free states as a subset.  
In this scenario, quantum channels serve as the operations that manipulate quantum states. We call the resource theories whose resource objects are quantum states \emph{resource theories of quantum states} or \emph{state theories} in short.

On the other hand, if one would like to study the resource contents belonging to quantum channels, as done e.g.\ in the theory of quantum communication, then the relevant 
objects of study become the set of quantum channels. 
Quantum channels are manipulated by quantum superchannels~\cite{Chiribella2008quantum,Chiribella2008transforming} that transform quantum channels to quantum channels. We call this framework \emph{resource theories of quantum channels}.
A quantum superchannel is constructed by a combination of two channels, between which another channel can be inserted. 
Inserting a channel $\mE$ into the slot for a superchannel $\Xi$ then results in an output channel $\Xi(\mE)$ (Fig.~\ref{fig:superchannel_comb}).

\begin{figure}
    \centering
    \includegraphics[width=\columnwidth]{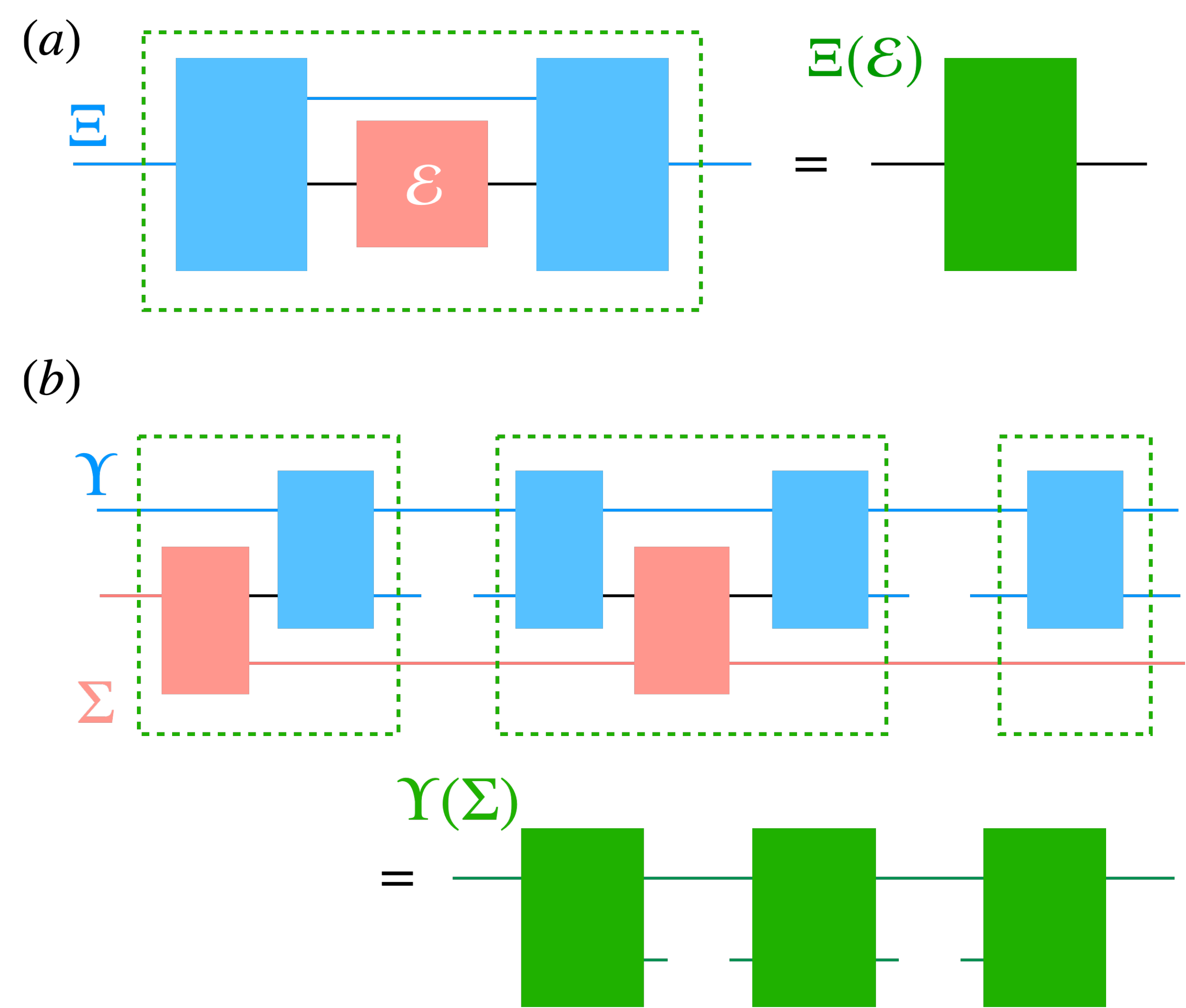}
    \caption{(a) Channel transformation by a superchannel. A superchannel $\Xi$ is constructed by two channels (denoted by blue boxes) that are connected by the identity channel (blue wire between two boxes). Inserting a channel $\mE$ into a slot between the two channels results in another channel $\Xi(\mE)$ (green). (b) Comb transformation by another comb. Combs $\Upsilon$ (blue) and $\Sigma$ (red) are constructed by a series of channels and interlocking them results in another comb $\Upsilon(\Sigma)$ (green).}
    \label{fig:superchannel_comb}
\end{figure}

One can extend superchannels to ones with multiple empty slots, known as quantum combs, and take them as the main resource object to study, which constructs \emph{resource theories of quantum combs}~\cite{Chiribella2008quantum,Berk2021resourcetheoriesof}. For instance, this framework is useful for studying noise suppression, where non-Markovian noise can be considered as a quantum comb~\cite{Berk2021extracting}. 
A quantum comb consists of a set of bipartite quantum channels and has empty slots between these channels.
This allows a quantum comb $\Upsilon$ to act on another quantum comb $\Sigma$ by interlocking them as in Fig.~\ref{fig:superchannel_comb}, which results in another comb $\Upsilon(\Sigma)$. 
Therefore, quantum combs themselves serve as the operations that manipulate quantum combs, and thus any set $\mO$ of free operations is a subset of quantum combs in this framework.

It is worth noting that there is a strict hierarchy among these three types of objects: quantum states are special forms of quantum channels, and quantum channels are special forms of quantum combs. 
Unless otherwise stated, our results hold for any type of resource object as long as it is isomorphic to a closed convex subset of operators acting on a finite-dimensional Hilbert space. 
Our results also do not assume any specification of the set $\mF$ of free objects and the set $\mO$ of free operations except the assumption that they are closed convex sets. 
This approach that does not rely on the specific structure of free objects is due to the recently developed \emph{general resource theories}~\cite{horodecki2013quantumness,chitambar2018quantum}, which has provided operational characterizations of general resources in terms of discrimination tasks~\cite{takagi2018operational,takagi2019general,Uola2019quantifying_conic,Oszmaniec2019operational}, resource erasure~\cite{anshu_2017,liu2019resource}, and resource manipulation~\cite{liu2019one,Regula2020benchmarking,Regula2021oneshot,PhysRevLett.125.060405,Regula2021fundamental,Fang2020no-go,Takagi2021oneshot,regula_2022,regula_2021-4}.


\section{Resource distillation}

Many quantum information processing protocols are designed under the assumption that we are in possession of a specific form of resource objects, e.g., maximally entangled states.  
However, this assumption is hard to meet in a realistic noisy scenario.
Therefore, preparing the desired specific object from distorted noisy ones using only the freely accessible operations is a crucial subroutine in the realization of quantum information processing tasks.
This procedure is known as \emph{resource distillation}, and its performance in relation to resource quantification has been a major topic of study. 

For instance, suppose two parties, Alice and Bob, would like to run the quantum teleportation protocol but only have access to noisy entangled states. 
Quantum teleportation can be run by first distilling a maximally entangled state from the accessible noisy entangled states by using local operations and classical communication, then using the distilled entangled state as a resource for quantum teleportation. 
In this protocol, the performance of the distillation process plays a crucial role in characterizing the efficiency of running quantum teleportation. In the following, we formalize the distillation performance in general settings, including scenarios where the resource object of interest is not only a quantum state but a quantum process --- representing, for instance, cases where one is interested in purifying a noisy quantum operation into a unitary one.

Let $T$ be a desired target object {--- again, this can be a quantum state, channel, or even a more general comb, depending on the setting of interest. Suppose now} that our goal is to obtain as many copies of $T$ as possible within a tolerable error $\varepsilon$. 
The one-shot distillation rate is defined as 
\bal
 D^\varepsilon(X)\coloneqq \sup_{\Lambda\in\mO}\lset m \sbar \Lambda(X)\sim_\varepsilon T^{\otimes m} \rset 
\eal
where $A\sim_\varepsilon B$ means that $A$ and $B$ are $\varepsilon$-close with respect to some distance measure. 
In this work, we focus on the trace-norm--based distance, which has the form below depending on the type of resource theories under study. 

For two quantum states $\rho_1$ and $\rho_2$, we consider the trace distance, i.e.,  
\bal
 \rho_1\sim_\varepsilon \rho_2 \iff \frac{1}{2}\|\rho_1-\rho_2\|_1 \leq \varepsilon
\eal
where $\|\cdot\|_1$ is the trace norm. 
In this manuscript, we use $\tau$ to denote a target state for state distillation and $\psi$ to emphasize that the target is pure.
A target state is commonly set as a pure state with some specific form, e.g., a Bell state in the entanglement distillation and the $T$ state in the magic state distillation, although we do not put any restriction about the target state unless otherwise stated.
For clarity, we will sometimes use $\tau$ instead of $T$ when the target is a quantum state.

For two quantum channels $\mE_1$ and $\mE_2$, the distance between two channels is described by the diamond distance, i.e., 
\bal\label{eq:diamond}
 \mE_1\sim_\varepsilon \mE_2 &\iff  \max_{\rho}\frac{1}{2}\|{\rm id}\otimes\mE_1(\rho)-{\rm id}\otimes\mE_2(\rho)\|_1\leq \varepsilon \\
 &\iff \frac{1}{2}\|\mE_1-\mE_2\|_\diamond\leq \varepsilon
\eal
where $\|\cdot\|_\diamond$ is the diamond norm~\cite{Kitaev1997quantum}, with $\rm id$ denoting the identity channel on an ancillary space --- \emph{a priori} unbounded, but it is in fact sufficient to consider an ancillary space of the same dimension as the input space of the channel~\cite[Theorem~3.46]{watrous_2018}.
What Eq.~\eqref{eq:diamond} means is that the distance between two channels can be measured by the maximum trace distance between two output states obtained by the partial application to the same input state.

This idea can be extended to measuring the distance between two quantum combs.
For given two combs $\Upsilon_1$ and $\Upsilon_2$ with the same input-output structure, we can consider applying them to another interlocking comb that outputs quantum states and takes the trace distance between these two states (Fig.~\ref{fig:comb_distance}).
We define the distance between $\Upsilon_1$ and $\Upsilon_2$ by taking the maximization over all such interlocking combs, i.e.,
\bal
 \frac{1}{2}\|\Upsilon_1-\Upsilon_2\|_c \coloneqq\max_\Sigma \frac{1}{2}\|\Upsilon_1(\Sigma)-\Upsilon_2(\Sigma)\|_1
\eal
where the maximization is taken over all the combs such that $\Upsilon_{1,2}(\Sigma)$ is a quantum state.  
We remark that this distance measure was introduced and discussed previously in Refs.~\cite{Chiribella2008memory,Gutoski2012on_a_measure}.
With this distance measure, we can define the $\varepsilon$-closedness of two combs $\Upsilon_1$ and $\Upsilon_2$ as 
\bal
 \Upsilon_1 \sim_\varepsilon \Upsilon_2 \iff \frac{1}{2}\|\Upsilon_1(\Sigma)-\Upsilon_2(\Sigma)\|_c\leq \varepsilon.
\eal

\begin{figure}
    \centering
    \includegraphics[width=\columnwidth]{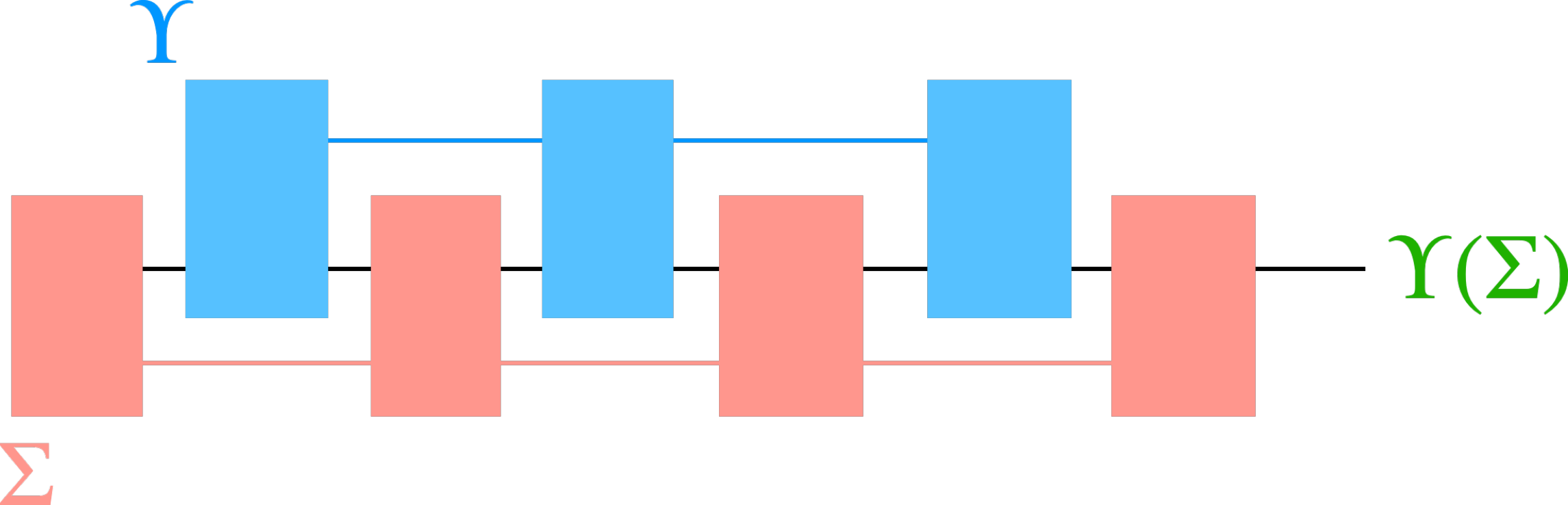}
    \caption{For a given comb $\Upsilon$, interlocking another comb $\Sigma$ that fills all the slots of $\Upsilon$ results in a quantum state as an output $\Upsilon(\Sigma)$. $\frac{1}{2}\|\Upsilon_1-\Upsilon_2\|_c$ is defined as the maximum trace distance between states $\Upsilon_1(\Sigma)$ and $\Upsilon_2(\Sigma)$.}
    \label{fig:comb_distance}
\end{figure}

\section{Virtual resource distillation}

\subsection{Setting}

As the ultimate goal of most quantum information processing tasks is to obtain classical information of interest, they typically terminate with measurements. 
We show that when the desired classical information is an expectation value of the final quantum state at the output of the protocol, we can extend the notion of resource distillation. 

For clarity, let us begin with state theories, i.e., $\mF$ is a set of quantum states and $\mO$ is a set of quantum channels. 
Consider a quantum information processing protocol that applies a quantum channel $\Lambda$ to a certain resource state $\tau$ to produce the final quantum state $\Lambda(\tau)$. 
This is followed by a measurement of an observable $M$.  
Here, we assume $-I/2\leq M \leq I/2$ because we can always normalize a bounded observable acting on a finite-dimensional system. 
This ensures that the measurement outcomes for $M$ are bounded in $[-1/2, 1/2]$.

We suppose that the classical information of interest is the expectation value of the final state for the observable $M$, i.e., $\Tr(M\Lambda(\tau))$.
The expectation value can be estimated by taking the statistical average over many measurement outcomes. 
The Hoeffding inequality~\cite{Hoeffding1963probability} ensures that $\mO(\log(1/\delta)/\beta^2)$ samples provide the estimate with accuracy $\beta$ with probability $1-\delta$. 

Suppose that we are not in possession of the state $\tau$ but instead have access to another state $\rho$ together with free operations $\mO$.  
For simplicity, we also suppose that the accessible free operations only accommodate a single copy of $\rho$ coherently, in the sense that they are one-shot protocols. 
This will allow us to provide a description particularly suitable for near-term technologies, where the size of available quantum devices is restricted; a more general characterization can be obtained by assuming that the input state is of the form $\rho^{\otimes n}$ for some number of copies of a state $\rho$.

One conventional way to estimate the desired expectation value is to prepare $D^\varepsilon(\rho)$ copies of the state $\tau$ from the available state $\rho$ by a distillation protocol. 
The optimal distillation protocol can prepare a state $\tilde \tau$ that is $\varepsilon$-close to $\tau^{\otimes D^\varepsilon(\rho)}$. 
{Recalling that we are interested in the expectation value of a single copy of $\tau$, we will assume}
that every reduced state of $\tilde \tau$ is identical --- which can be realized by symmetrization --- and let $\tau'$ be the reduced state of $\tilde \tau$ with the same size of $\tau$. 
Since the partial trace does not increase the trace distance, we have $\frac{1}{2}\|\tau' - \tau\|_1\leq \varepsilon$.
We then get 
\bal
|\Tr(M\Lambda(\tau'))-\Tr(M\Lambda(\tau))|&\leq \frac{1}{2}\|\Lambda(\tau')-\Lambda(\tau)\|_1\\
&\leq \frac{1}{2}\|\tau' - \tau\|_1\\
&\leq \varepsilon
\eal
for an arbitrary observable $M$ with $-I/2\leq M \leq I/2$.

Therefore, the distilled state admits the estimation of expectation value with the accuracy $\varepsilon$.
This means that $\mO((1/D^\varepsilon(\rho))\log(1/\delta)/\beta^2)$ copies of $\rho$ provide the estimate with accuracy $\beta+\varepsilon$ with probability $1-\delta$, where the number of copies differs by the factor $D^\varepsilon(\rho)$ compared to the case when the resource state $\tau$ is available. 

\subsection{Framework}

We now introduce virtual resource distillation. 
The basic idea is that one can apply classical postprocessing to help estimate the desired expectation value.  
Suppose that a state $\tilde \tau$ that is $\varepsilon$-close to $\tau^{\otimes m}$ can be decomposed into the form 
\bal
 \tilde \tau = \lambda_+ \Lambda_+(\rho) - \lambda_-\Lambda_-(\rho)
\eal
where $\Lambda_\pm\in\mO$ are free operations and $\lambda_\pm\geq 0$ are nonnegative numbers.
Since we assume here that $\Lambda_\pm$ are trace preserving, we have $\lambda_+-\lambda_-=1$.
Letting $\tau'_\pm$ be the reduced states of $\Lambda_\pm(\rho)$, the reduced state $\tau'$ of $\tilde \tau$ can be written as $\tau'=\lambda_+ \tau'_+ - \lambda_- \tau'_-$.
Defining $\gamma\coloneqq \lambda_+ + \lambda_-$ and $p_\pm = \lambda_\pm/\gamma$, the quantity
\begin{equation}\begin{aligned}
&\lambda_+\Tr(M\Lambda(\Lambda_+(\rho))) - \lambda_-\Tr(M\Lambda(\Lambda_-(\rho))) \\
 &= p_+\gamma\Tr(M\Lambda(\Lambda_+(\rho))) -  p_- \gamma \Tr(M\Lambda(\Lambda_-(\rho)))
\end{aligned}\end{equation}
corresponds to the desired expectation value $\Tr(M\Lambda(\tau))$ with error $\varepsilon$. 
This form ensures that the following procedure gives the estimator of the expectation value with bias $\varepsilon$.
\begin{enumerate}
    \item Flip the biased coin that lands heads with probability $p_+$ and tails with probability $p_-$.
    \item When we see heads, apply $\Lambda_+$ to $\rho$ and measure {$m$ commuting observables $M\otimes I^{\otimes m-1},\ I\otimes M \otimes I^{\otimes m-2}, \dots, I^{\otimes m-1}\otimes M$ to get outcomes $o_1,\dots,o_m$}.  
    Store the value $\gamma o_1,\dots,\gamma o_m$ to the classical register. If we see tails, apply $\Lambda_-$ to $\rho$, measure the same observables, and get $m$ measurement outcomes $o_1,\dots, o_m$. Store the value $-\gamma o_1,\dots,-\gamma o_m$ in the classical register. 
    \item Repeat the above process and take the sample average of the values stored in the classical register. 
\end{enumerate}

Due to the classical postprocessing, in which we multiply $\gamma$ or $-\gamma$ to the measurement outcome, the possible range of each random variable changes to $[-\gamma/2, \gamma/2]$.
The Hoeffding inequality ensures that this procedure allows us to estimate the desired expectation value with accuracy $\beta+\varepsilon$ with probability $1-\delta$ with $\mO((\gamma^2/m )\log (1/\delta)/\beta^2)$ samples.
Comparing to the way that the conventional distillation rate $D^\varepsilon(\rho)$ is involved in the sample number of $\rho$ motivates us to introduce the \emph{virtual resource distillation rate} as 
\begin{equation}\begin{aligned}
   V^\varepsilon(\rho)\coloneqq &\sup_m \frac{m}{C^\varepsilon(\rho,m)^2}
  \end{aligned}\end{equation}
where $C^\varepsilon(\rho,m)$ is the \emph{virtual resource distillation overhead} defined by 
\begin{equation}\begin{aligned}
  &C^\varepsilon(\rho,m)\\
  &\coloneqq \inf\lset\lambda_++\lambda_-\sbar \frac{1}{2}\|\tau^{\otimes m} - (\lambda_+ \Lambda_+(\rho) - \lambda_- \Lambda_-(\rho))\|_1\leq \varepsilon,\right.\\
   &\hphantom{\coloneqq \inf\lset\lambda_++\lambda_-\sbar\right.}\left.\lambda_\pm\geq 0,\ \lambda_+-\lambda_-=1,\ \Lambda_\pm \in \mO\rset.
\end{aligned}\end{equation}
In the above, we used $\inf$ rather than $\min$ to implicitly allow for pathological situations where a feasible decomposition of $\tilde\tau$ does not exist, and hence $C^\varepsilon(\rho,m) = \infty$. As long as the optimization is feasible --- as is the case in most of the practically encountered cases --- the optimum will always be achieved for all closed sets of operations $\mO$.

{We note that the above procedure can be easily adapted to virtually simulate the expectation value of any observable $M'$ acting on the many-copy state $\tau^{\otimes m}$; above, we only studied single-copy measurements that provide a natural justification to our definition of the distillation rate $V^\varepsilon$, but the approach itself is much more general.}

The above argument can also be extended to resource theories of quantum channels and combs. 
For channel theories, consider a quantum information processing protocol described by a superchannel $\Xi$ applying to a certain resource channel $\mA$ such that $\Xi(\mA)$ is a quantum state corresponding to the final state in the algorithm right before the terminating measurement $M$. 
Then,
\bal
 \left|\Tr(M\Xi(\mA))-\Tr(M\Xi(\mA'))\right|&\leq\frac{1}{2}\left\|\Xi(\mA)-\Xi(\mA')\right\|_1\\
 &\leq \frac{1}{2}\|\mA-\mA'\|_\diamond
\eal
where in the last inequality, we used that the diamond norm does not increase under superchannels and that the diamond distance for quantum states reduces to the trace distance.
This allows us to approximate $\mA^{\otimes m}$ from a given channel $\mE$ by $\lambda_+\Lambda_+(\mE)-\lambda_-\Lambda_-(\mE)$ where $\Lambda_\pm\in\mO$ are free superchannels. 

The case for resource theories of quantum combs goes similarly. 
A quantum information processing protocol can now be considered as a comb $\Upsilon$ applied to a certain resource comb $\Theta$. Then, the error in expectation value for using another comb $\Theta'$ instead of $\Theta$ is bounded by 
\bal
 \left|\Tr(M\Upsilon(\Theta))-\Tr(M\Upsilon(\Theta'))\right|&\leq\frac{1}{2}\left\|\Upsilon(\Theta)-\Upsilon(\Theta')\right\|_1\\
 &\leq \frac{1}{2}\|\Theta-\Theta'\|_c,
\eal
where the fact that the comb distance does not increase under the application of another comb is apparent from the definition of the comb distance.
This similarly allows us to approximate copies of the $\Theta^{\otimes m}$ from a given comb $\Upsilon$ by $\lambda_+\Lambda_+(\Upsilon)-\lambda_-\Lambda_-(\Upsilon)$ where $\Lambda_\pm\in\mO$ are free combs.

These observations can be summarized as the following definition of virtual resource distillation rate and overhead that can be applied to general types of resource objects.
\begin{definition}\label{def:rate and overhead}
For a resource theory with a set $\mO$ of free operations, the virtual resource distillation rate of a given object $X$ with respect to the target object $T$ is 
  \begin{equation}\begin{aligned}
   V^\varepsilon(X)\coloneqq &\sup_m \frac{m}{C^\varepsilon(X,m)^2}
  \end{aligned}\end{equation}
with the virtual resource distillation overhead $C^\varepsilon(X,m)$ defined by 
\bal
  C^\varepsilon(X,m)\coloneqq &\inf\lset\lambda_++\lambda_-\sbar T^{\otimes m} \sim_\varepsilon \lambda_+ \Lambda_+(X) - \lambda_- \Lambda_-(X),\right.\\
   &\left.\lambda_\pm\geq 0,\ \lambda_+-\lambda_- = 1,\ \Lambda_\pm \in \mO\rset.
\eal
\end{definition}
We remark that the overhead is equivalently written as 
\bal
  C^\varepsilon(X,m)= &\inf\lset\sum_i |\lambda_i|\sbar T^{\otimes m} \sim_\varepsilon \sum_i \lambda_i \Lambda_i(X),\right.\\
   &\left.\lambda_i\in\mbR\ \forall i,\ \sum_i\lambda_i= 1,\ \Lambda_i \in \mO\ \forall i\rset.
\eal
The form in Definition~\ref{def:rate and overhead} is recovered by letting $\lambda_+\coloneqq \sum_{i:\lambda_i\geq0}\lambda_i$ and $\lambda_-\coloneqq \sum_{i:\lambda_i<0}(-\lambda_i)$, as well as $\Lambda_+\coloneqq \lambda_+^{-1}\sum_{i:\lambda_i\geq 0}\lambda_i \Lambda_i$ and $\Lambda_-\coloneqq \lambda_-^{-1}\sum_{i:\lambda_i< 0}(-\lambda_i) \Lambda_i$, and noting that $\Lambda_\pm\in\mO$ follows due to the convexity of $\mO$. 

Since the conventional distillation can be reconstructed with a restriction $\lambda_-=0$, we always have $D^\varepsilon(X)\leq V^\varepsilon(X)$. 
We will see later that this inequality is strict in many cases. 

{We stress here that this framework is conceptually very different from many-copy distillation protocols, which are often encountered in practical applications of conventional distillation. Namely, at all stages of the virtual distillation process, only single-copy operations are used, and no joint channels acting on $\rho^{\otimes n}$ are needed.}

We note also a superficial conceptual similarity to a recent approach of~\cite{regula_2022-3}, where distillation under non-completely-positive resource manipulation protocols was considered; however, that framework does not require the operations to be implementable through a classical postprocessing of physical (completely positive) free operations, yielding a setting that may be difficult to directly compare with ours.

We further remark that several protocols known as virtual cooling~\cite{Cotler2019quantum} and virtual distillation in quantum error mitigation~\cite{Koczor2021exponential,Huggins2021virtual} share a similar idea that classical postprocessing enables us to simulate purer quantum states, although they differ from our framework --- these techniques extract the measurement statistics for purer quantum objects by coherently interacting multiple copies of noisy objects, while our virtual resource distillation applies a probabilistic operation to a single copy of the noisy quantum object, which is inspired by and related to the techniques for simulating unphysical objects~\cite{Buscemi2013twopoint,Jiang2021physical,Regula2021operational}, and error mitigation technique based on quasiprobability~\cite{PhysRevLett.119.180509}.


\subsection{Probabilistic distillation}\label{sec:prob}

A more general form of distillation protocols are those that can succeed only with some probability. Here, we describe the basic setting, compare it with virtual distillation, and discuss the possibility of extending virtual resource distillation to probabilistic protocols.

Let us begin with conventional distillation in state theories.
General probabilistic operations are represented by subchannels (completely positive trace-nonincreasing maps). 
Any such map can be thought of as being part of a quantum instrument, that is, a collection of probabilistic operations $\{\Lambda_i\}_i$ such that the overall transformation $\sum_i \Lambda_i$ is trace preserving. The outcome $i$, obtained with probability $p_i \coloneqq \Tr \Lambda_i(\rho)$, then corresponds to the final state $\Lambda_i(\rho) / p_i$. Among such probabilistic operations, we define a subset $\mO_{\leq 1}$ of subchannels and call it free subchannels.
Then, probabilistic distillation is a process that transforms a given state to a desired target state with $\mO_{\leq 1}$ with some probability.

{%
The distillation process should, upon success, output a state close to the target state, which can then be measured to estimate the expectation value of interest. 
The experimenters, knowing if the distillation did or did not succeed, can postselect only the successful outcomes of the process. 
To collect a sufficient number of samples to estimate the expectation value with the desired accuracy, one needs to use a number of samples that is inversely proportional to the success probability. 
This motivates the definition of a probabilistic one-shot distillation rate with respect to the target state $\tau$ as
\bal
D^\varepsilon_p(\rho):=\sup_{\Lambda^p\in\mO_{\leq 1}}\lset m \Tr(\Lambda^p(\rho)) \sbar \frac{\Lambda^p(\rho)}{\Tr(\Lambda^p(\rho))}\sim_\varepsilon \tau^{\otimes m} \rset.
\eal
Here the superscript $p$ denotes the probabilistic nature of subchannels. 

Such an approach is seemingly very similar to virtual distillation: multiple samples are taken by applying free operations to a single copy of $\rho$, and a protocol can only succeed by collecting a sufficient number of them.
As we showed in~\cite{maintext}, virtual distillation can offer strict improvements over probabilistic distillation, and in particular there exist cases when $D^\varepsilon_p(\rho) = 0 < V^\varepsilon(\rho)$. 
However, the fact that free \emph{sub}channels $\mO_{\le 1}$ are employed in probabilistic distillation --- which are, in general, a strictly larger class of maps than $\mO$ --- means that a direct comparison between the probabilistic one-shot rate $D^\varepsilon_p$ and the virtual rate $V^\varepsilon$ may not be possible in general.

Let us then formalize an explicit extension of virtual distillation to subchannels, which we will allow for a more direct comparison with conventional probabilistic approaches.
}%
Suppose that the target state can be written as 
\bal
\tau^{\otimes m}\sim_\varepsilon \lambda_+\frac{\Lambda^{\yx{p}}_+(\rho)}{\Tr(\Lambda^{\yx{p}}_+(\rho))}-\lambda_-\frac{\Lambda^{\yx{p}}_-(\rho)}{\Tr(\Lambda^{\yx{p}}_-(\rho))}
\label{eq:target decomposition probabilistic postprocess}
\eal
for $\lambda_\pm\geq 0$ and $\Lambda^{\yx{p}}_\pm\in\mO_{\leq 1}$.
Then, extending the deterministic virtual resource distillation introduced above, the expectation value can be estimated in the following manner. 
\begin{enumerate}
    \item Flip a biased coin that lands on heads with probability $p_+\coloneqq\frac{\lambda_+}{\lambda_++\lambda_-}$ and on tails with probability $p_-\coloneqq\frac{\lambda_-}{\lambda_++\lambda_-}$.
    \item When we see heads, apply $\Lambda^{\yx{p}}_+$ to $\rho$. 
    If failure is reported, start over from Step~1. If successful, measure {$m$ commuting observables $M\otimes I^{\otimes m-1},\ I\otimes M \otimes I^{\otimes m-2}, \dots, I^{\otimes m-1}\otimes M$ to get outcomes $o_1,\dots,o_m$}.  
    Store the value $\gamma o_1,\dots,\gamma o_m$ with $\gamma\coloneqq \lambda_++\lambda_-$ to the classical register. If we see tails, apply $\Lambda^{\yx{p}}_-$ to $\rho$ and follow the same procedure.
    \item Repeat the above process and take the sample average of the values stored in the classical register. 
\end{enumerate}
Note that postselection is involved in the second step, which makes the protocol probabilistic. 
This process provides an estimator with bias $\varepsilon$. 
The number of samples to use scales with $\gamma^2$ by the same mechanism for the deterministic case, as well as the average success probability in Step~2, $\frac{\lambda_+\Tr(\Lambda^{\yx{p}}_+(X))+\lambda_-\Tr(\Lambda^{\yx{p}}_-(X))}{\lambda_++\lambda_-}$, which linearly contributes to the sampling cost. 
This motivates us to introduce the probabilistic virtual distillation rate with respect to a target state $\tau$ defined as 
\begin{equation}\begin{aligned}
&V_p^\varepsilon(\rho)\\
&\coloneqq \sup_{\Lambda^{\yx{p}}_\pm\in\mO_{\leq 1}} \lset \frac{m\left[\lambda_+\Tr(\Lambda^{\yx{p}}_+(\rho))+\lambda_-\Tr(\Lambda^{\yx{p}}_-(\rho))\right]}{(\lambda_++\lambda_-)^3}\sbar \lambda_\pm\geq 0,\right.\\
&\left. \quad\tau^{\otimes m}\sim_\varepsilon \lambda_+\frac{\Lambda^{\yx{p}}_+(\rho)}{\Tr(\Lambda^{\yx{p}}_+(\rho))}-\lambda_-\frac{\Lambda^{\yx{p}}_-(\rho)}{\Tr(\Lambda^{\yx{p}}_-(\rho))},\ \lambda_+-\lambda_-=1\rset.
\label{eq:virtual probabilistic state}
\end{aligned}\end{equation}
\yx{Since probabilistic distillation could outperform deterministic distillation, we expect that probabilistic virtual distillation also outperforms deterministic virtual distillation. We leave a detailed study of this advantage in a future work.}

Alternatively, the expectation value can be estimated without postselection.
Let us rewrite \eqref{eq:target decomposition probabilistic postprocess} with $\lambda_\pm\rightarrow\lambda_{\pm}\Tr(\Lambda^{\yx{p}}_\pm(\rho))$, which gives $\tau^{\otimes m}=\lambda_+\Lambda^{\yx{p}}_+(\rho)-\lambda_-\Lambda^{\yx{p}}_-(\rho)$. 
With this $\lambda_\pm$, we follow the same procedure as above, except that instead of postselecting on the successful events in Step~2, we store the value 0 upon failure. 
This gives the virtual distillation rate without postselection as
\begin{equation}\begin{aligned}
&\tilde V_p^\varepsilon(\rho)\\
&\coloneqq \sup_{\Lambda^{\yx{p}}_\pm\in\mO_{\leq 1}} \lset \frac{m}{(\lambda_++\lambda_-)^2}\sbar\tau^{\otimes m}\sim_\varepsilon \lambda_+\Lambda^{\yx{p}}_+(\rho)-\lambda_-\Lambda^{\yx{p}}_-(\rho),\right.\\
&\left.\quad\lambda_\pm\geq 0,\ \lambda_+\Tr(\Lambda^{\yx{p}}_+(\rho))-\lambda_-\Tr(\Lambda^{\yx{p}}_-(\rho))=1\rset
\label{eq:virtual probabilistic without postselection}
\end{aligned}\end{equation}

We remark that the values of $\lambda_\pm$ in \eqref{eq:virtual probabilistic state} do not explicitly depend on $\Tr(\Lambda^{\yx{p}}_\pm(\rho))$ as $\lambda_+-\lambda_-= 1$. 
On the other hand, the values for $\lambda_\pm$ in \eqref{eq:virtual probabilistic without postselection} are larger than those in \eqref{eq:virtual probabilistic state} by a factor of $1/\Tr(\Lambda^{\yx{p}}_\pm(\rho))$. 
This makes $V_p^\varepsilon$ scale with $\Tr(\Lambda^{\yx{p}}_\pm(\rho))$ while
$\tilde V_p^\varepsilon$ scale with $\left[\Tr(\Lambda^{\yx{p}}_\pm(\rho))\right]^{2}$, reflecting the absence of postselection. 
Therefore, $\tilde V_p^\varepsilon$ becomes significantly smaller than $V_p^\varepsilon$ when the success probability of the free subchanenls is small. 

The discussion becomes more involved for channel theories. 
The probabilistic channel transformation can be formalized by subsuperchannels, which transform channels to subchannels even when acting only on a part of a larger system~\cite{burniston2019necessary}. 
The difference from the case of state theories is that the success probability of the protocol depends not only on the description of the subchannel but also input states.  
Therefore, to ensure that the resultant channel is close to the target channel upon success, we need to make sure that all output states are close to the desired final states upon success~\cite{Regula2021fundamental}. 

The number of samples to ensure the estimation of expectation values with the desired accuracy also depends on the success probability, which depends on input states. 
It is therefore reasonable to take the worst-case scenario and define the probabilistic distillation rate of a channel $\mE$ with respect to the target channel $\mA$ as 
\bal
D^\varepsilon_p(\mE):=&\sup_{\Lambda\in\mO_{\leq 1}}\min_\rho\lset m \Tr({\rm id}\otimes\Lambda^{\yx{p}}(\mE)(\rho)) \sbar\right.\\
&\left.\frac{{\rm id}\otimes\Lambda^{\yx{p}}(\mE)(\sigma)}{\Tr({\rm id}\otimes\Lambda^{\yx{p}}(\mE)(\sigma))}\sim_\varepsilon {\rm id}\otimes \mA^{\otimes m}(\sigma),\ \forall\sigma \rset,
\eal
with the minimization being over all possible input states $\rho$.

Analogously, the probabilistic virtual resource distillation rate with postselection can be written as
\begin{widetext}
\bal
V_p^\varepsilon(\mE)\coloneqq &\sup_{\Lambda^{\yx{p}}_\pm\in\mO_{\leq 1}}\min_\rho \lset \frac{m\left[\lambda_+\Tr({\rm id}\otimes\Lambda^{\yx{p}}_+(\mE)(\rho))+\lambda_-\Tr({\rm id}\otimes\Lambda^{\yx{p}}_-(\mE)(\rho))\right]}{(\lambda_++\lambda_-)^3}\sbar\right.\\
&\left. \quad{\rm id}\otimes \mA^{\otimes m}\sim_\varepsilon \lambda_+\frac{{\rm id}\otimes\Lambda^{\yx{p}}_+(\mE)(\sigma)}{\Tr({\rm id}\otimes\Lambda^{\yx{p}}_+(\mE)(\sigma))}-\lambda_-\frac{{\rm id}\otimes\Lambda^{\yx{p}}_-(\mE)(\sigma)}{\Tr({\rm id}\otimes\Lambda^{\yx{p}}_-(\mE)(\sigma))}\ \forall \sigma,\ \lambda_+-\lambda_-=1,\ \lambda_\pm\geq 0\rset.
\eal
\end{widetext}

Unlike the state case, the probabilistic virtual resource distillation rate without postselection does not work in general, as 
\bal
{\rm id}\otimes \mA^{\otimes m}(\sigma)\sim_\varepsilon \lambda_+{\rm id}\otimes\Lambda^{\yx{p}}_+(\mE)(\sigma)-\lambda_-{\rm id}\otimes\Lambda^{\yx{p}}_-(\mE)(\sigma)\ \forall\sigma
\eal
is not generally satisfied when $\Tr[{\rm id}\otimes\Lambda^{\yx{p}}_\pm(\mE)(\sigma)]$ differs depending on $\sigma$. 
Indeed, if we take the trace on both sides, the left-hand side always gives $\Tr[{\rm id}\otimes \mA^{\otimes m}(\sigma)]=1,\, \forall \sigma$ while the right-hand side can vary for different states $\sigma$.

An analogous extension can be made to the resource theories of combs, where a similar subtlety about the success probability depending on the input channels and states remains.  

{In the rest of this paper, we focus on the deterministic virtual rates defined through $V^\varepsilon$ and $C^\varepsilon$, which are more easily characterizable than the probabilistic virtual rates, while at the same time already being sufficiently general to allow for considerable advantages over conventional distillation.}


\subsection{Estimation of probability distribution}
The above discussion shows that when there exist $\Lambda_\pm\in\mO$ and $\lambda_\pm\geq 0$ such that $T\sim_\varepsilon \lambda_+\Lambda_+(X)-\lambda_-\Lambda_-(X)$, any expectation value of a target object $T$ can be obtained by measuring  $\Lambda_{\pm}(X)$ instead.
Here, we apply this observation to the estimation of probability distributions of $T$.

Let $\eta$ be an arbitrary output state resulting from a target object $T$. For state theories, $\eta$ coincides with $T$, while for channel and comb theories $\eta$ is an output from $T$ with an arbitrary input state. 
Suppose that we measure $\eta$ in the computational basis. For each measurement, we get a one-shot measurement outcome $o_j$ with probability $p(j)=\tr[\eta\ketbra{j}{j}]$. After $N$ independent measurements, we will get $n(j)$ counts for outcome $o_j$ with $\sum_j n(j) = N$\yx{, and according to the Hoeffding inequality~\cite{Hoeffding1963probability}, with failure probability $\delta\in (0,1)$, we have
\begin{equation}
 \bigg|p(j) - \tilde p(j)\bigg|=  \mO\left(\sqrt{\frac{\log\delta}{N}}\right),
\end{equation}
where $\tilde p(j) = n(j)/N$.
}

Next we consider how to simulate this measurement process with $\Lambda_{\pm}(X)$.
Since the projector $\ketbra{j}{j}$ is also an observable, we can directly apply the virtual resource distillation framework developed above. 
Let $\eta_\pm$ be output states from $\Lambda_\pm(X)$.
To ensure the same accuracy, we use $(\lambda_++\lambda_-)^2 N$ copies of $X$.
With probability $\lambda_\pm/(\lambda_++\lambda_-)$, we measure $\eta_\pm$ in the computational basis, multiplying the outcomes --- +1 (click) or 0 (no click) --- by $\pm(\lambda_++\lambda_-)$ for each $j$ and take the sample average. 
This is equivalent to measuring $\eta_\pm$ for $N_{\pm}\coloneqq(\lambda_++\lambda_-)^2\frac{\lambda_\pm}{\lambda_++\lambda_-}N=\lambda_{\pm}(\lambda_++\lambda_-)N$ times, where our estimate for the probability is 
\bal
p'(j) \coloneqq (\lambda_++\lambda_-)\left[\frac{n_+(j)}{N_+}- \frac{n_-(j)}{N_-}\right]
\eal
with $n_\pm(j)$ standing for the number of times the outcome $o_j$ was observed among $N_\pm$ measurements. 
The general framework of virtual resource distillation developed above ensures that 
\bal
 \left|p(j) - p'(j)\right|=\mO\left(\yx{\sqrt{\frac{\log\delta}{N}}}\right) + \varepsilon.
\eal

This also implies that the measurement counts $n(j)$ can be approximated as
\begin{equation}
    n(j) \approx Np'(j) = \frac{n_+}{\lambda_+} - \frac{n_-}{\lambda_-}.
\end{equation}
\yx{with negligible error and failure probability.}


\subsection{Virtual resource monotones}

Virtual distillation rate is an operationally motivated quantity and may be hard to evaluate exactly for some settings. 
Therefore, it will be useful to establish other quantities that can help evaluate the virtual distillation rate. 
Here, we introduce a notion of a monotone that can always be used to bound the virtual distillation overhead with zero error. 

\begin{proposition}\label{prop:virtual_monotones}
Let $M$ be a function that obeys the following properties:
\begin{enumerate}
	\item $M(X) \geq M\left( \mu_+ \Lambda_+ (X) - \mu_- \Lambda_+ (X) \right)\; \forall \Lambda_+,\Lambda_-\in \mO, \; \mu_++\mu_- = 1$,
	\item $M(\mu X) = \mu M(X)\; \forall \mu > 0$. 
\end{enumerate}
Then,
\begin{equation}\begin{aligned}\label{Eq:costmonotonebound}
	C^{0}(X,m) \geq \frac{ M( T^{\otimes m} )}{ M( X )}.
\end{aligned}\end{equation}
\end{proposition}
\begin{proof}
For any operations $\Lambda_\pm \in {\mO}$ and any $\lambda_\pm > 0$, we can write
\begin{equation}\begin{aligned}
	& M\left(\lambda_+ \Lambda_+ (X) - \lambda_- \Lambda_- (X)\right)\\
        \quad &= M\left( [ \lambda_+ + \lambda_-] \frac{\lambda_+ \Lambda_+ (X) - \lambda_- \Lambda_- (X)}{\lambda_+ + \lambda_-} \right)\\
        \quad &= [ \lambda_+ + \lambda_-] \,M\left( \frac{\lambda_+ \Lambda_+ (X) - \lambda_- \Lambda_- (X)}{\lambda_+ + \lambda_-} \right)\\
        \quad &\leq [ \lambda_+ + \lambda_-] \, M(X)
\end{aligned}\end{equation}
where the last two lines are by definition of a virtual monotone. Optimizing over all virtual operations yields the stated result.
\end{proof}

Examples of virtual resource monotones are:
\begin{enumerate}
\item In the resource theory of entanglement, the base norm $\|\rho\|_{\mc S} = \min \lset \mu_+ + \mu_- \sbar \rho = \mu_+ \sigma_+ - \mu_- \sigma_-,\; \sigma_\pm \in {\mc S} \rset$, where $\mS$ is the set of separable states, is a virtual monotone under all separability-preserving operations. This quantity is directly related to the (standard) robustness of entanglement~\cite{vidal1999robustness}.
\item Also in the resource theory of entanglement, the negativity $\|\rho^{\Gamma}\|_1$~\cite{Vidal2002computable} is a virtual monotone under all PPT operations.
\item In the resource theory of coherence, the $\ell_1$ norm of coherence $\|\rho\|_{\ell_1}$~\cite{Baumgratz14} is a virtual monotone under all incoherent operations.
\end{enumerate}

One special virtual resource monotone is the inverse virtual distillation overhead,
\begin{equation}
    \tilde{M}(X, m) \coloneqq \frac{1}{C^{0}(X,m)}.
\end{equation}
To prove the first property, we denote $X' = \lambda_+ \Lambda_+ (X) - \lambda_- \Lambda_+ (X)$ and have 
\begin{equation}\begin{aligned}
    &C^{0}\left(X' ,m \right) \\
    &= \inf \left\{\tilde \lambda_++\tilde \lambda_-\,|\,\tilde\lambda_+\tilde \Lambda_+ (X') - \tilde \lambda_- \tilde \Lambda_+ (X')= T^{\otimes m} \right\},\\
    &=\inf \left\{\tilde \lambda_++\tilde \lambda_-\,|\,\left(\tilde\lambda_+\lambda_+\tilde \Lambda_+\circ \Lambda_+ + \tilde\lambda_-\lambda_-\tilde \Lambda_-\circ \Lambda_- \right) ( X)\right.\\
    &\quad\left.- \left(\tilde\lambda_+\lambda_-\tilde \Lambda_+\circ \Lambda_- + \tilde\lambda_-\lambda_+\tilde \Lambda_-\circ \Lambda_+ \right) ( X)= T^{\otimes m} \right\}.
\end{aligned}\end{equation}
Note that 
\begin{equation}
    \tilde \lambda_++\tilde \lambda_- = \tilde\lambda_+\lambda_+ + \tilde\lambda_-\lambda_- + \tilde\lambda_+\lambda_- + \tilde\lambda_-\lambda_+, 
\end{equation}
thus we have 
\begin{equation}
    C^{0}\left(X' ,m \right) \ge C^{0}\left(X ,m \right)
\end{equation}
and hence
\begin{equation}
    \tilde{M}(X, m) \ge \tilde{M}(X', m).
\end{equation}
It is obvious that $\tilde{M}(\lambda X, m) = \lambda \tilde{M}(X, m)$.

According to Eq.~\eqref{Eq:costmonotonebound}, we have
\begin{equation}
    M(X) \ge M(T^{\otimes m}) \cdot \tilde{M}(X, m).
\end{equation}
Consider normalized virtual monotones with  $M(T^{\otimes m}) = 1$, we thus have
\begin{equation}
   M(X) \ge  \tilde{M}(X, m). 
\end{equation}
That is, the specific virtual monotone $\tilde{M}(X, m)$ lower bounds all normalized virtual monotones.


\section{Evaluation of virtual resource distillation overhead}

Evaluating the virtual resource distillation rate is generally a formidable task, mainly due to the optimization over the number $m$ of copies of the target object. 
Here, we show that a closely related quantity, namely the virtual resource distillation overhead $C^\varepsilon(\rho,m)$ for fixed $m$, can be efficiently characterized in general state theories.

 We first present a useful alternative form of the distillation overhead.
The assumption that $\mO$ is convex ensures that the overhead $C^\varepsilon(\rho,m)$ is a solution of a convex optimization program.  
Therefore, taking the convex dual (see Appendix~\ref{app:dual}), we obtain an alternative expression for state theories as 
\begin{widetext}
\bal
C^\varepsilon(\rho,m)= \inf_{\substack{\tilde \tau\sim_\varepsilon \tau^{\otimes m}\\\Tr(\tilde \tau)=1}}\sup_{W\in{\rm Herm}}\lset 2\Tr(W\tilde \tau)-1\sbar 0\leq \Tr(W\Lambda(\rho))\leq 1,\ \forall\Lambda\in\mO\rset
\label{eq:overhead dual}
\eal
\end{widetext}
where Herm is the set of Hermitian operators. 
This form can be extended to channel and combs theories by considering Choi operators in place of quantum states.

The overhead with $\varepsilon=0$ can be computed by linear or semidefinite programming if the structure of the free operations is sufficiently simple. 
For example, if the set of free objects is the convex hull of a finite number of objects $\{f_i\}_i$, an operation $\Lambda$ is resource non-generating if and only if 
\begin{equation}
 \forall i,\, \Lambda(f_i) = \sum_j p_j f_j,\, p_j\geq0,\,\sum_j p_j = 1,\, f_j \in \mF.
\end{equation}
Then the requirement of a resource non-generating operation is characterized by a linear constraint, making the overhead computable by linear programming. 
This includes resource theories such as purity, thermodynamics, coherence, and magic.

There also exist resource theories with the infinite number of extreme free resource objects, such as the theory of entanglement or its dynamical counterpart, the theory of quantum memory. We then need to relax the requirement of free operations by considering free objects either induced from a finite number of extreme objects or characterizable via a semidefinite constraint, making the overhead computable by semidefinite programming (SDP). Both ways will give an upper bound to the overhead and the bound can be tight with better approximations.

As an example, we write down the SDP for entanglement under positive partial transpose (PPT) operations, i.e., the set of bipartite channels $AB \to A'B'$ whose Choi operators are PPT across the bipartition $AA' : BB'$. 
Suppose the unnormalized Choi operator of a PPT channel $\mc N$ is $J^{A'ABB'}_{\mc N}$, it should satisfy 
\begin{equation}
    J^{A'ABB'}_{\mc N}\ge 0,~\tr_{AB}[J^{A'ABB'}_{\mc N}] = I_{A'B'},~(J^{A'ABB'}_{\mc N})^{T_{BB'}}\ge 0.
\end{equation}
For any input state $\rho_{AB}$, the output state is
\begin{equation}
    \mc N(\rho_{AB}) = \tr_{A'B'}[\rho_{A'B'}^T\cdot J^{A'ABB'}_\mN].
\end{equation}
Then the virtual distillation overhead of $\rho_{AB}$ with respect to the Bell state $\Phi\coloneqq\dm{\Phi}$ with $\ket{\Phi}=\frac{1}{\sqrt{2}}(\ket{00}+\ket{11})$ is
\begin{equation}
\begin{aligned}
&C^\varepsilon(\rho_{AB}) =\min \big\{ \lambda_++\lambda_- \;|\\
 &\qquad\tr_{A'B'}[\rho_{A'B'}^T\cdot J^{A'ABB'}_{\Lambda_+}] - \tr_{A'B'}[\rho_{A'B'}^T\cdot J^{A'ABB'}_{\Lambda_-}] = \Phi,\\
&\qquad J^{A'ABB'}_{ \Lambda_+},~ J^{A'ABB'}_{ \Lambda_-}\ge 0,\\
&\qquad \tr_{AB}[J^{A'ABB'}_{ \Lambda_+}] = \lambda_+I_{A'B'},~\tr_{AB}[J^{A'ABB'}_{ \Lambda_-}]=\lambda_-I_{A'B'} \\
&\qquad (J^{A'ABB'}_{\Lambda_+})^{T_{BB'}},~(J^{A'ABB'}_{\Lambda_-})^{T_{BB'}}\ge 0 \, \big\}.
\end{aligned}
\end{equation}

\subsection{Tight bounds for general resources of quantum states}

We will now show that much simpler bounds based on convex and semidefinite programming can be obtained for state theories. This will remove the need to optimize over all free operations, and apply also to resource theories such as quantum entanglement.
We introduce general upper and lower bounds on the distillation overhead in general quantum resource theories. They depend on several resource measures, whose definitions we now recall.

For a set $\mc F$ of free states, define the generalized robustness $R^g_{\mc F}$~\cite{vidal1999robustness,PhysRevA.67.054305,Harrow2003robustness}, the standard robustness $R^s_{\mc F}$~\cite{vidal1999robustness}, and the resource fidelity $F_{\mc F}$ as
\begin{equation}\begin{aligned}
	R^g_{\mc F} (\rho) \coloneqq& \inf \lset \lambda \sbar \frac{\rho + \lambda \omega}{1+\lambda} \in \mc F,\; \omega \in \mc D \rset\\
	R^s_{\mc F} (\rho) \coloneqq& \inf \lset \lambda \sbar \frac{\rho + \lambda \sigma}{1+\lambda} \in \mc F,\; \sigma \in \mc F \rset\\
 F_{\mc F} (\rho) \coloneqq& \max_{\sigma \in \mc F} F(\rho , \sigma)
\end{aligned}\end{equation}
where $\mD$ is the set of quantum states and $F$ is the fidelity. We remark that when $\rho$ is a pure state, it holds that $F_{\mc F} (\rho) = \max_{\sigma \in \mc F} \tr (\rho \sigma)$. 
Define now the following optimization problem:
\begin{equation}\begin{aligned}\label{eq:general_cost_exact}
	\zeta_\varepsilon^s(\rho, k)
	\coloneqq \min \big\{ \mu_+ + \mu_- \;\big|\;& 0 \leq Q_+ \leq \mu_+ \id,\; 0 \leq Q_- \leq \mu_- \id,\\
	& \tr Q_+ \sigma \leq \frac{\mu_+}{k}\;\; \forall \sigma \in {\mc F},\\
	&\tr Q_- \sigma \leq \frac{\mu_-}{k} \; \forall \sigma \in {\mc F},\\
	& \mu_+ - \mu_- = 1,\\
	&\tr \rho (Q_+ - Q_-) \geq 1- \varepsilon \big\},
\end{aligned}\end{equation}
where $k$ is some parameter to be fixed.
We also define $\zeta_\varepsilon^g(\rho,k)$ to be the same optimization except that the inequality constraints $\Tr(Q_+\sigma)\leq \mu_+/k$, $\Tr(Q_-\sigma)\leq\mu_-/k$, $\forall\sigma\in\mF$ become equality constraints.  
Then, we obtain the following general lower and upper bounds.

\begin{theorem}[Theorem~1 in the companion paper~\cite{maintext}]\label{prop:general_bounds_cost_parallel}
Consider a convex resource theory and a target pure resource state $\psi$. 
Let $\mO$ be the class of resource non-generating operations. If $R^s_{\mc F}(\psi) < \infty$, then
\begin{equation}
   \zeta^s_\varepsilon\left(\rho, F_{\mc F}(\psi^{\otimes m})^{-1}\right)\le C^\varepsilon(\rho,m) \leq  \zeta^s_\varepsilon\left(\rho, R^s_{\mc F} (\psi^{\otimes m}) + 1 \right).
   \label{eq:overhead standard SDP}
\end{equation}
Furthermore, if it holds that $\braket{\psi|\sigma|\psi}$ is constant for all $\sigma \in \mc F$, then 
\begin{equation}\begin{aligned}
 \zeta^g_\varepsilon\left(\rho, F_{\mc F}(\psi^{\otimes m})^{-1}\right)\le C^\varepsilon(\rho,m) \leq \zeta^g_\varepsilon\left(\rho, R^g_{\mc F}(\psi^{\otimes m}) + 1 \right).
 \label{eq:overhead generalized SDP}
\end{aligned}\end{equation}

\end{theorem}

The crucial property of the bounds is that whenever $R^s_{\mc F} (\psi^{\otimes m}) + 1 = F_{\mc F}(\psi^{\otimes m})^{-1}$ --- which is true in resource theories such as bi- and multi-partite entanglement~\cite{vidal1999robustness,Regula2020benchmarking} or multi-level quantum coherence~\cite{Johnston2018evaluating} --- or if $R^g_{\mc F}(\psi^{\otimes m}) + 1 = F_{\mc F}(\psi^{\otimes m})^{-1}$ and the overlap $\braket{\psi|\sigma|\psi}$ is constant --- which is true in resource theories such as coherence or athermality --- then the upper and lower bounds coincide, yielding an exact expression for the overhead $C^\varepsilon(\rho,m)$.

We will later consider specific examples of resource theories, showing how the result can be applied in different contexts, and in some cases improving on and extending the statement of Theorem~\ref{prop:general_bounds_cost_parallel}.

\begin{proof}
We first prove \eqref{eq:overhead standard SDP}.
Consider any feasible distillation protocol such that $\Lambda_\pm \in \mO$ and $\frac12 \| \lambda_+ \Lambda_+ (\rho) - \lambda_- \Lambda_-(\rho) - \psi^{\otimes m} \|_1 \leq \varepsilon$. Define $Q_\pm = \lambda_\pm \Lambda_\pm^\dagger(\psi^{\otimes m})$ and $\mu_\pm = \lambda_\pm$. 
Since $\Lambda_\pm$ are free operations, it holds that $\Lambda_\pm(\sigma) \in \mc F$ for any $\sigma \in \mc F$, and hence
\begin{equation}\begin{aligned}
	\max_{\sigma \in \mc F} \tr Q_\pm \sigma &= \mu_\pm \max_{\sigma \in \mc F} \tr \psi^{\otimes m} \Lambda_\pm(\sigma)\\
	&\leq \mu_\pm \max_{\sigma' \in \mc F} \tr \psi^{\otimes m} \sigma'\\
	&\leq \mu_\pm F_{\mc F}(\psi^{\otimes m}).
 \label{eq:overhead standard lower bound}
\end{aligned}\end{equation}
Due to the fact that $\Lambda_\pm$ are CPTP maps, we also get $0 \leq Q_\pm \leq \mu_\pm \id$, and the condition $\tr \rho (Q_+ - Q_-) \geq 1-\varepsilon$ is ensured by the fact that
\begin{equation}\begin{aligned}\label{eq:trace_distance_ineq}
	\varepsilon &\geq \frac12 \left\| \psi^{\otimes m} - \lambda_+ \Lambda_+ (\rho) + \lambda_- \Lambda_-(\rho) \right\|_1\\
	&= \max \lset \tr \left[\left(\psi^{\otimes m} - \lambda_+ \Lambda_+ (\rho) + \lambda_- \Lambda_-(\rho)\right) X \right]\sbar 0 \leq X \leq \id \rset\\
	&\geq \tr \left[\left(\psi^{\otimes m} - \lambda_+ \Lambda_+ (\rho) + \lambda_- \Lambda_-(\rho)\right) \psi^{\otimes m}\right]\\
	&= 1 - \tr \rho (Q_+ - Q_-).
\end{aligned}\end{equation}
Therefore, $Q_\pm$ give a valid feasible solution to $\zeta^s_\varepsilon (\rho, F_{\mc F}(\psi^{\otimes m})$ with optimal value $\mu_+ + \mu_- = \lambda_+ + \lambda_-$, which concludes the first part of the proof.

Conversely, let $Q_\pm$ be feasible solutions to $\zeta^s_\varepsilon(\rho, R^s_{\mc F}(\psi^{\otimes m})+1)$. Note that we can always take $Q_\pm$ such that $\tr ([Q_+ - Q_-] \rho) = 1-\varepsilon$, since for any feasible $Q_\pm$ with $\tr ([Q_+ - Q_-] \rho) = t(1-\varepsilon)$ for some $t>1$, $\frac{1}{t} Q_\pm$ are also feasible with the same optimal value. Now, define the quantum channels
\begin{equation}\begin{aligned}
	\Lambda_\pm (\omega) = \tr \left(\frac{Q_\pm}{\mu_+} \omega\right) \psi^{\otimes m} + \tr \left(\left[\id - \frac{Q_\pm}{\mu_\pm}\right] \omega\right) \sigma_\psi,
 \label{eq:overhead upper bound channel}
\end{aligned}\end{equation}
where $\sigma_\psi \in \mc{F}$ is a state such that
\begin{equation}\begin{aligned}
	\frac{\psi^{\otimes m} + R^s_{\mc F}(\psi^{\otimes m}) \sigma_\psi}{1+R^s_{\mc F}(\psi^{\otimes m})} \in \mc F.
\end{aligned}\end{equation}
Notice that
\begin{equation}\begin{aligned}
	\Lambda_\pm (\sigma) &\propto \psi^{\otimes m} + \frac{\tr \left(\left[\id - \frac{Q_\pm}{\mu_\pm}\right] \sigma\right)}{\tr \left(\frac{Q_\pm}{\mu_\pm} \sigma\right)} \sigma_\psi\\
	&= \psi^{\otimes m} + \left( \frac{\mu_+}{\tr Q_\pm \sigma} - 1 \right) \sigma_\psi
\end{aligned}\end{equation}
which entails that, since $\tr Q_\pm \sigma \leq \frac{\mu_+}{R^s_{\mc F}(\psi^{\otimes m})+1}$ for any $\sigma\in\mc{F}$, we necessarily have $\Lambda_\pm(\sigma) \in \mc{F}$ and thus $\Lambda_\pm$ are resource non-generating operations.
Now, since
\begin{equation}\begin{aligned}
	&\mu_+ \Lambda_+ (\rho) - \mu_- \Lambda_- (\rho)\\
	&= \tr ([Q_+ - Q_-] \rho) \psi^{\otimes m} + \big( \mu_+ - \mu_- - \tr ([Q_+ - Q_-] \rho) \big) \sigma_\psi\\
	&=  \tr ([Q_+ - Q_-] \rho) \psi^{\otimes m} + \big( 1 - \tr ([Q_+ - Q_-] \rho) \big) \sigma_\psi,
\end{aligned}\end{equation}
we get
\begin{equation}\begin{aligned}
	\left\| \mu_+ \Lambda_+ (\rho) - \mu_- \Lambda_- (\rho) - \psi^{\otimes m} \right\|_1 &\leq 2 \left| 1 - \tr ([Q_+ - Q_-] \rho) \right|\\
	&= 2\varepsilon,
\end{aligned}\end{equation}
and thus we see that the maps realize the desired transformation with error $\frac12 \left\| \psi^{\otimes m} - \lambda_+ \Lambda_+ (\rho) + \lambda_- \Lambda_-(\rho) \right\|_1 \leq \varepsilon$,
yielding $C^\varepsilon(\rho,m) \leq \zeta_\varepsilon^s(\rho, R^s_{\mc F} (\psi^{\otimes m})+1)$.

The proof for \eqref{eq:overhead generalized SDP} proceeds analogously. 
To show the lower bounds, notice that the inequalities in \eqref{eq:overhead standard lower bound} become equalities due to the assumption that $\braket{\psi|\sigma|\psi}$ is constant for all $\sigma\in\mF$.  
To show the upper bound, we choose in \eqref{eq:overhead upper bound channel} a state $\omega_\psi$ which satisfies
\begin{equation}\begin{aligned}
	\frac{\psi^{\otimes m} + R^g_{\mc F}(\psi^{\otimes m}) \omega_\psi}{1+R^g_{\mc F}(\psi^{\otimes m})} \in \mc F.
\end{aligned}\end{equation}
instead of $\sigma_\psi$.
\end{proof}


\subsection{Bounds in terms of the maximum overlap}

In the case of state theories, we can give alternative expressions for general lower bounds for the virtual distillation overhead.
They can be formulated in relation to how close the given resource state can be brought to the target state via free operations. 
To formalize this, we define the \emph{maximum overlap} with a target pure state $\psi$ as
\bal
 f_{\mO}(\rho,m)\coloneqq \max_{\Lambda\in\mO}\tr[\Lambda(\rho)\psi^{\otimes m}].
\eal

 Then, we obtain the following general bound in terms of the maximum overlap.
 
\begin{proposition}\label{prop:cost lower bound fraction}
Let $\psi$ denote a pure target resource state, and let $\mO$ be an arbitrary convex and closed set of free operations.
Then, for every state $\rho$, positive integer $m$, and $\varepsilon\in[0,1]$,

\bal
C^\varepsilon(\rho,m)\geq \max\left\{\frac{2(1-\varepsilon)}{f_{\mO}(\rho,m)}-1,1\right\}
\label{eq:cost lower bound overlap general}
\eal
holds.
\end{proposition}
\begin{proof}

Using the max--min inequality, the dual form of the overhead \eqref{eq:overhead dual} can be lower bounded as 
\begin{widetext}
\bal
C^\varepsilon(\rho,m) 
 &\geq \sup_{W\in {\rm Herm}}\inf_{\substack{\tr[\eta]=1\\\frac{1}{2}\|\eta-\psi^{\otimes m}\|_1\leq \varepsilon}}\lset 2\tr[\eta W]-1\sbar 0\leq \tr[\Lambda(\rho)W]\leq 1,\ \forall \Lambda\in\mO\rset
 \label{eq:cost epsilon general dual}
\eal    
\end{widetext}
We note that $W=f_{\mO}(\rho,m)^{-1}\psi^{\otimes m}$ is a feasible solution for the last optimization problem because $0\leq f_{\mO}(\rho,m)^{-1}\tr[\Lambda(\rho)\psi^{\otimes m}]\leq 1,\ \forall \Lambda\in\mO$.
This ensures
\bal
C^\varepsilon(\rho,m)&\geq \inf_{\substack{\tr[\eta]=1\\\frac{1}{2}\|\eta-\psi^{\otimes m}\|_1\leq \varepsilon}}\left[2f_\mO(\rho,m)^{-1}\tr[\eta \psi^{\otimes m}]-1\right]
\label{eq:cost fraction proof intermediate}
\eal
Since every $\eta$ such that $\tr[\eta]=1$ satisfies $\frac{1}{2}\|\eta-\psi^{\otimes m}\|_1 = \max_{0\leq E \leq I} \tr[(\eta-\psi^{\otimes m})E]$, choosing $E=\psi^{\otimes m}$ specifically results in
\bal
 \frac{1}{2}\|\eta-\psi^{\otimes m}\|_1 &\geq \tr[(\psi^{\otimes m}-\eta)\psi^{\otimes m}]\\
 &= 1-\tr[\eta \psi^{\otimes m}].
\eal
Therefore, $\frac{1}{2}\|\eta-\psi^{\otimes m}\|\leq\varepsilon$ implies $\tr[\eta \psi^{\otimes m}]\geq 1-\varepsilon$. 
This allows us to lower bound the right-hand side of \eqref{eq:cost fraction proof intermediate} to get
\bal
C^\varepsilon(\rho,m)\geq\frac{2(1-\varepsilon)}{f_\mO(\rho,m)}-1.
\eal

The lower bound of 1 can be obtained by choosing $W=I$ in \eqref{eq:cost epsilon general dual}.

\end{proof}

Although the maximum overlap is operationally intuitive, it could be hard to compute, or it might obscure the relation with the resourcefulness contained in $\rho$. 
To address this, we can employ recent results that connect the operational distillation performance and fundamental resource measures that are equipped with geometric viewpoints. 
To this end, we recall the weight resource measure defined for an arbitrary convex resource theory 
\bal
 W_{\mc F}(\rho) \coloneqq \max\lset w \sbar \rho= w\sigma + (1-w)\tau,\ \sigma\in\mc F,\ \tau\in\mc D\rset.
\eal
Then, the maximum overlap can be upper bounded using the generalized robustness and the weight measure as follows. 

\begin{lemma}[\cite{Regula2021fundamental,Fang2020no-go}]\label{lem:distillation fidelity bound}
For every convex and closed set $\mc F$ of free states and every set $\mO$ of free operations,
\bal
 f_{\mO}(\rho,m)\leq F_{\mc F}(\psi^{\otimes m})\left[R_{\mc F}^g(\rho)+1\right]
\eal
and 
\bal
 f_{\mO}(\rho,m)\leq 1-\left[1-F_{\mc F}(\psi^{\otimes m})\right]W_{\mc F}(\rho).
\eal
\end{lemma}

Combining Proposition~\ref{prop:cost lower bound fraction} and Lemma~\ref{lem:distillation fidelity bound} immediately gives the following alternative lower bound for $C^\varepsilon$.

\begin{corollary}
Let $\psi$ denote the pure target state, and let $\mO$ be an arbitrary convex and closed set of free operations.
Then, for every state $\rho$, positive integer $m$, and $\varepsilon\in[0,1]$,

\begin{equation}\begin{aligned}
  C^\varepsilon(\rho,m) \geq \frac{2(1-\varepsilon)}{F_{\mc F}(\psi^{\otimes m})(R_{\mc F}^g(\rho)+1)}-1.
\end{aligned}\end{equation}
and
\begin{equation}\begin{aligned}
  C^\varepsilon(\rho,m) \geq \frac{2(1-\varepsilon)}{1-(1-F_{\mc F}(\psi^{\otimes m}))W_{\mc F}(\rho)}-1.
\end{aligned}\end{equation}

\end{corollary}

The following result provides a sufficient condition when the bound in Proposition~\ref{prop:cost lower bound fraction} is saturated.

\begin{theorem}[Theorem~2 in the companion paper~\cite{maintext}]\label{prop:cost twirling}
For a pure target state $\psi$, suppose that there exists a free ``generalized twirling'' operation~\cite{Takagi2021oneshot} $\mT\in\mO$ of the form 
\bal
 \mT(\cdot)=\tr[\psi^{\otimes m}\cdot]\psi^{\otimes m} + \tr[(I-\psi^{\otimes m})\cdot]\sigma^\star
 \label{eq:generalized twirling}
\eal
for some $\sigma^\star\in\mathcal{F}$. Then, 
\bal
 C^\varepsilon(\rho,m)=\max\left\{\frac{2(1-\varepsilon)}{f_{\mO}(\rho,m)}-1,1\right\}
\eal
for every $\varepsilon \in [0,1]$.
\end{theorem}

\begin{proof}
We already showed $C^\varepsilon(\rho,m)\geq\max\left\{\frac{2(1-\varepsilon)}{f_{\mO}(\rho,m)}-1,1\right\}$ in \eqref{eq:cost lower bound overlap general} for general case. 
To show the opposite inequality, let $\Lambda^\star\in \mO$ be the one that realizes $\tr(\Lambda^\star(\rho)\psi^{\otimes m})=f_{\mO}(\rho,m)$.
Then, the free twirling operation $\mT$ maps $\Lambda^\star(\rho)$ to the ``generalized isotropic'' state with the same overlap with the target state as
\bal
 \mT\circ \Lambda^\star(\rho) = f_{\mO}(\rho,m)\,\psi^{\otimes m} + (1-f_{\mO}(\rho, m))\,\sigma^\star.
\eal

When $\frac{2(1-\varepsilon)}{f_\mO(\rho,m)}-1<1\iff 1-f_\mO(\rho,m)<\varepsilon$, we have $\frac{1}{2}\|\psi^{\otimes m}-\mT\circ\Lambda^\star(\rho)\|_1\leq 1-f_\mO(\rho,m)<\varepsilon$, implying $C^\varepsilon(\rho,m)\leq 1$.

On the other hand, when $\frac{2(1-\varepsilon)}{f_\mO(\rho,m)}-1\geq 1$, we define
\bal
 \Lambda_+={\rm id},\quad \Lambda_-(\cdot) = \sigma^\star,
 \label{eq:map isotropic general}
\eal
both of which are free operations. 
Then, 
\begin{equation}\begin{aligned}
 &\frac{1-\varepsilon}{f_{\mO}(\rho,m)} \Lambda_+\circ\mT\circ\Lambda^\star(\rho)-\left(\frac{1-\varepsilon}{f_{\mO}(\rho,m)}-1\right)\Lambda_-\circ\mT\circ\Lambda^\star(\rho)\\
 &\quad=(1-\varepsilon)\psi^{\otimes m}+\varepsilon\sigma^\star.
\end{aligned}\end{equation}
Since 
\bal
\frac{1}{2}\|(1-\varepsilon)\psi^{\otimes m}+\varepsilon\sigma^\star-\psi^{\otimes m}\|_1 &= \frac{1}{2}\|-\varepsilon \psi^{\otimes m}+\varepsilon\sigma^\star\|_1 \\
&\leq \frac{\varepsilon}{2}\left(\|\psi^{\otimes m}\|_1+\|\sigma^\star\|_1\right)\\
&=\varepsilon
\eal
where we used the triangle inequality. 
Thus, we get 
\bal
C^\varepsilon(\rho,m)&\leq \frac{1-\varepsilon}{f_{\mO}(\rho,m)}+\frac{1-\varepsilon}{f_{\mO}(\rho,m)}-1\\
&=\frac{2(1-\varepsilon)}{f_{\mO}(\rho,m)}-1,
\eal
concluding the proof. 

\end{proof}

Although Theorem~\ref{prop:cost twirling} provides a tight characterization of distillation overhead, its assumption may appear somewhat contrived, as it could be difficult to determine whether such a free generalized twirling operation exists for a given setting.  First, we stress that the conditions are satisfied in a number of the most practically relevant resource theories --- we will see this explicitly in Section~\ref{sec:examples}. 

We can also give a useful sufficient condition for a free twirling operation of Eq.~\eqref{eq:generalized twirling} to exist: this is true whenever the fidelity-based measure and the standard robustness of the target state coincide~\cite[Lemma~5]{Takagi2021oneshot}.
This gives the following characterization, which is usually easier to verify directly.

\begin{corollary}
For a given set $\mathcal{F}$ of free states, consider the class $\mO$ of resource non-generating operations. If $F_{\mathcal{F}}(\psi^{\otimes m})^{-1}=1+R_{\mathcal{F}}^s(\psi^{\otimes m})$ holds, we have $C^\varepsilon(\rho,m)= \max\left\{\frac{2(1-\varepsilon)}{f_{\mO}(\rho,m)}-1,1\right\}$.
\end{corollary}


\section{Examples}\label{sec:examples}

Here, we apply the general framework developed above to specific physical settings of interest. 
We primarily focus on evaluating virtual resource distillation overhead $C^\varepsilon(X,m)$ for a fixed $m$, which provides a lower bound for the virtual distillation rate $V^\varepsilon(X)$ and may allow for computing $V^\varepsilon(X)$ in the case when an analytical expression of $C^\varepsilon(X,m)$ is available.

\subsection{Entanglement}

Let us first consider the resource theory of entangled states, where separable states construct the set $\mS$ of free states~\cite{Horodecki09}.
In entanglement theory, resource non-generating operations are conventionally called separability-preserving (or non-entangling) operations, and any physically motivated set of operations such as the set of LOCC is a subset of separability-preserving operations.

The set of separable states does not allow an efficient characterization via semidefinite constraint and thus relevant resource measures are generally hard to compute. 
To obtain bounds for these quantities, it is often useful to consider a set $\mS_{\rm PPT}$ of positive-partial transpose (PPT) states. 
Since every separable state is PPT, $\mS$ is a subset of $\mS_{\rm PPT}$, and inclusion is known to be strict. 
The resource non-generating operations for the set $\mS_{\rm PPT}$ are called PPT-preserving maps.
A relevant useful class which forms a subset of PPT-preserving maps is known as PPT operations, which are the channels whose Choi operators are PPT.

As a target state $\psi$, we take the bipartite qubit Bell state $\Phi=\dm{\Phi}$ with $\ket{\Phi}=(\ket{00}+\ket{11})/\sqrt{2}$.

\subsubsection{General bounds}

Let us first study the property for general states. 
We first recall the hypothesis-testing relative entropy of entanglement~\cite{Datta2009} as
\begin{equation}\label{Eq:defEmin}
	E_H^{\varepsilon}(\rho) = \max_{\substack{0 \leq A \leq I\\\operatorname{Tr}\left(A{\rho}\right) \geq 1-\varepsilon}} \min_{\sigma \in \mathcal{S}}(-\log_2\operatorname{Tr}(A \sigma)).
\end{equation}
Then, we get the following upper bound for the overhead. 
\begin{proposition}
For an arbitrary state $\rho$, the overhead for virtual resource distillation with separability-preserving operations is upper bounded as 
\begin{equation}
	C^{\varepsilon}(\rho,m)  \le 2^{{m-E_H^{\varepsilon}(\rho)}+1}-1
\end{equation}
for $m\ge E_H^{\varepsilon}(\rho)$.     
\end{proposition}

\begin{proof}
	Following a similar argument in Theorem~\ref{prop:general_bounds_cost_parallel}, we explicitly construct the distillation protocol saturating the upper bound. We consider two separability-preserving channels as
	\begin{equation}\label{Eq:}
\begin{aligned}
	  \Lambda_1(\rho)&=\frac{1}{2^{m-E_H^{\varepsilon}(\rho)}}\operatorname{Tr}(A \rho) \Phi^{\otimes m}\\
	  &+\left(1-\frac{1}{2^{m-E_H^{\varepsilon}(\rho)}}\operatorname{Tr}(A \rho) \right) \frac{I-\Phi^{\otimes m}}{2^{2k}-1},\\
	  \Lambda_2(\rho)&=\frac{I-\Phi^{\otimes m}}{2^{2k}-1},
\end{aligned}
\end{equation}
where $A$ is the operator that achieves the optimal solution of Eq.~\eqref{Eq:defEmin}. For any separable state $\sigma\in\mc S$, we have 
\begin{equation}
	\frac{1}{2^{m-E_H^{\varepsilon}(\rho)}}\operatorname{Tr}(A \sigma) \le \frac{1}{2^m},
\end{equation}
indicating that the output state $\Lambda_1(\sigma)$ is also separable and hence $\Lambda_1$ is separability preserving. Since the output state of $\Lambda_2$ is separable, $\Lambda_2$ is also separability preserving. 

Letting $\lambda_1=2^{m-E_H^{\varepsilon}(\rho)}$ and $\lambda_2=\lambda_1-1$, we have
\begin{equation}\begin{aligned}
	\Phi' &= \lambda_1\Lambda_1(\rho)-\lambda_2\Lambda_2(\rho) \\
	&= \operatorname{Tr}(A \rho) \Phi^{\otimes m}+\left(1-\operatorname{Tr}(A \rho) \right) \frac{I-\Phi^{\otimes m}}{2^{2m}-1}.
\end{aligned}\end{equation}
Since $\operatorname{Tr}(A \rho)\ge 1-\varepsilon$, we have
\begin{equation}
	\frac{1}{2}\|\Phi'-\Phi^{\otimes m}\|\le \varepsilon. 
\end{equation}
Therefore, this constructs a virtual resource distillation protocol with an overhead of 
\begin{equation}    
	\lambda_1+\lambda_2 = 2^{{m-E_H^{\varepsilon}(\rho)}+1}-1.
\end{equation}
\end{proof}

We thus obtain a lower bound to the distillation rate as
\begin{equation}
	V^\varepsilon(\rho)\ge \max_m \frac{m}{(2^{{m-E_H^{\varepsilon}(\rho)}+1}-1)^2}.
 \label{eq:virtual rate entanglement bound}
\end{equation}
We can also evaluate the right-hand side as follows. 
Let $c\coloneqq 2^{-E_H^\varepsilon(\rho)+1}$ and $g(m)\coloneqq m/(2^m c-1)^2$. 
We then get 
\bal
 \frac{d}{dm}g(m)=\frac{c 2^m(1-2m\ln 2)-1}{(2^m c-1)^3}.
 \label{eq:derivative}
\eal
The fact that $0\leq E_H^\varepsilon(\rho)\leq 1$ for every $\rho$ gives $1\leq c\leq 2$.
Since $c 2^m(1-2m\ln 2 )-1< 0$ for $m\geq 1$ and $c\geq 0$, we get $\frac{d}{dm}g(m)<0$ for all $m\geq 1$. 
Therefore, the maximum happens when $m=1$, which gives the tightest form of \eqref{eq:virtual rate entanglement bound} as
\begin{equation}
	V^\varepsilon(\rho)\ge \frac{1}{\left(2^{{2-E_H^{\varepsilon}(\rho)}}-1\right)^2}.
\end{equation}

\subsubsection{Exact expressions via semidefinite programming}

Here, we focus on the zero-error distillation, i.e., $\varepsilon=0$, and provide the exact expressions for the virtual resource distillation overhead.

\begin{proposition}\label{prop:SEPP}
Let $\mO$ be the class of separability-preserving operations. Then, for any state $\rho$, the distillation overhead is given by the convex optimization program
\begin{equation}\begin{aligned}\label{eq:sepp_cost_exact}
	C^0(\rho,m) = \min \big\{ \mu_+ &+ \mu_- \;\big|\; 0 \leq Q_+ \leq \mu_+ \id,\; 0 \leq Q_- \leq \mu_- \id,\\
	& \tr Q_+ \sigma \leq \frac{\mu_+}{2^m},\; \tr Q_- \sigma \leq \frac{\mu_-}{2^m} \; \forall \sigma \in {\mc S},\\
	&\tr \rho (Q_+ - Q_-) = 1 = \mu_+ - \mu_-  \big\}.
\end{aligned}\end{equation}
\end{proposition}
\begin{proof}
Direct application of Theorem~\ref{prop:general_bounds_cost_parallel}.
\end{proof}

By focusing on PPT operations, the above optimization can take a simpler form and become an SDP.

\begin{proposition}\label{prop:PPTP}
Let $\mO$ be the class of positive partial transpose (PPT) operations. Then, for any state $\rho$, the virtual distillation overhead is given by the semidefinite program
\begin{equation}\begin{aligned}\label{eq:ppt_cost_exact}
	C^0(\rho,m) = \min \big\{ \mu_+ + \mu_- \;\big|\;& 0 \leq Q_+ \leq \mu_+ \id,\; 0 \leq Q_- \leq \mu_- \id,\\
	& \left\| Q_+^\Gamma \right\|_{\infty} \leq \frac{\mu_+}{2^m},\; \left\| Q_-^\Gamma \right\|_{\infty}\leq \frac{\mu_-}{2^m},\\
	&\tr \rho (Q_+ - Q_-) = 1 = \mu_+ - \mu_-  \big\}.
\end{aligned}\end{equation}
\end{proposition}
\begin{proof}
For any PPT channels $\Lambda_\pm$ such that $\lambda_+ \Lambda_+ (\rho) - \lambda_- \Lambda_- (\rho) = \Phi^{\otimes m}$, we define $Q_\pm = \lambda_\pm \Lambda_\pm^\dagger(\Phi^{\otimes m})$ and $\mu_\pm = \lambda_\pm$, where $\Lambda^\dagger$ is the dual (adjoint) map. Then
\begin{equation}\begin{aligned}
	 \left\| Q_+^\Gamma \right\|_{\infty} &= \max_{\omega \in \mc{D}} \left| \tr Q_+^\Gamma \omega \right|\\
	 &= \lambda_\pm\max_{\omega \in \mc{D}} \left| \tr \Phi^{\otimes m} \Lambda_\pm(\omega^\Gamma) \right|\\
	 &= \lambda_\pm\max_{\omega \in \mc{D}} \left| \tr \left(\Phi^{\otimes m}\right)^\Gamma \Lambda_\pm(\omega^\Gamma)^\Gamma \right|\\
	 &\leq \lambda_\pm\max_{\omega' \in \mc{D}} \left| \tr \left(\Phi^{\otimes m}\right)^\Gamma \omega'\right|\\
	 &= \frac{\lambda_\pm}{2^m},
\end{aligned}\end{equation}
where in the first line we used $\mc{D}$ to denote all density matrices, in the fourth line we used that  $\Lambda_\pm(\omega^\Gamma)^\Gamma$ must be a valid state because $\Lambda_\pm$ is a PPT operation, and in the fourth line we used the fact that the partial transpose of the maximally entangled state $\Phi$ is the swap operator with eigenvalues $\pm \frac{1}{2}$. Since $\Lambda_\pm$ are CPTP maps, the other conditions on $Q_\pm$ are also satisfied, meaning that $Q_\pm$ constitute feasible solutions to \eqref{eq:ppt_cost_exact}. Optimizing over all feasible PPT protocols yields $C^0(\rho,m) \geq \mu_+ + \mu_-$ for optimal $\mu_\pm$.

For the other direction, take any feasible $Q_\pm$ and $\mu_\pm$. Define
\begin{equation}\begin{aligned}
	\Lambda_\pm (\omega) = \tr \left(\frac{Q_\pm}{\mu_+} \omega\right) \Phi^{\otimes m} + \tr \left(\left[\id - \frac{Q_\pm}{\mu_+}\right] \omega\right) \frac{\id - \Phi^{\otimes m}}{2^{2m} - 1}.
\end{aligned}\end{equation}
Writing $(\Phi^{\otimes m})^\Gamma = \frac{1}{2^m} \left( \Pi_S - \Pi_A \right)$ where $\Pi_S$ and $\Pi_A$ denote the projectors onto the symmetric and antisymmetric spaces, respectively, we have that
\begin{equation}\begin{aligned}
	J_{\Lambda_\pm}^\Gamma &= \frac{Q_\pm^\Gamma}{\mu_\pm} \otimes \frac{\Pi_S - \Pi_A}{2^m} + \frac{\id - \frac{Q_-}{\mu_+}}{2^{2m}-1} \otimes \frac{(2^m - 1) \Pi_S + (2^m+1) \Pi_A}{2^m}\\
	&= \left( \frac{Q_\pm^\Gamma}{\mu_+} + \frac{1}{2^m+1} \left[ \id - \frac{Q_\pm^\Gamma}{\mu_+} \right] \right) \otimes \frac{\Pi_S}{2^m} \\
	&\quad+ \left( - \frac{Q_\pm^\Gamma}{\mu_+} + \frac{1}{2^m-1} \left[ \id - \frac{Q_\pm^\Gamma}{\mu_+} \right] \right) \otimes \frac{\Pi_A}{2^m}\\
	&\geq \left( -\frac{\id}{2^m} + \frac{1}{2^m+1} \left[ \id - \frac{\id}{2^m} \right] \right) \otimes \frac{\Pi_S}{2^m} \\
	&\quad+ \left( -\frac{\id}{2^m} + \frac{1}{2^m-1} \left[ \id - \frac{\id}{2^m} \right] \right) \otimes \frac{\Pi_A}{2^m}\\
	&= 0,
\end{aligned}\end{equation}
where in the third line we used that $-\frac{\mu_\pm}{2^m} \leq Q_\pm^\Gamma \leq \frac{\mu_\pm}{2^m}$ and that $\Pi_S$ and $\Pi_A$ are orthogonal to each other. The operations $\Lambda_\pm$ are therefore PPT, and since
\begin{equation}\begin{aligned}
	\mu_+ \Lambda_+ (\rho) - \mu_- \Lambda_- (\rho) = \Phi^{\otimes m},
\end{aligned}\end{equation}
we obtain $C^0(\rho,m) \leq \mu_+ + \mu_-$ as desired.
\end{proof}

\subsubsection{Exact expressions via singlet fraction}

We can obtain an alternative exact expression with the singlet fraction. The singlet fraction of a quantum state $\rho$ is the maximum overlap with the maximally entangled state realized by applying an arbitrary LOCC operation to $\rho$.
Here, we slightly generalize this concept and call
\bal
f_\mO(\rho,m)\coloneqq \max_{\Lambda\in\mO}\tr[\Lambda(\rho)\Phi^{\otimes m}]
\eal
the singlet fraction with respect to a set of free operations $\mO$. 
Then, we get the following characterization. 

\begin{proposition}
\label{prop:singlet fraction}
Let $\mO$ be an arbitrary subset of separability-preserving operations that contains local unitary operations assisted by shared classical randomness, e.g.\ (one-way) LOCC.  
 Then, for every state $\rho$,
\bal
 C^\varepsilon(\rho,m)=
 \max\left\{\frac{2(1-\varepsilon)}{f_\mO(\rho,m)}-1,1\right\}. 
\eal
\end{proposition}

\begin{proof}
This is a direct application of Theorem~\ref{prop:cost twirling}, as the local twirling $\mT(\cdot)\coloneqq \int dU U\otimes U^*\cdot U^\dagger\otimes U^{*\dagger}$ serves as a free generalized twirling operation.

\end{proof}

Proposition~\ref{prop:singlet fraction} implies that the performance of virtual resource distillation does not change over different choices of free operations for pure states.

\begin{corollary}
For any bipartite pure state $\phi$, the overhead $C^\varepsilon(\phi,m)$ is the same for any choice of operations ranging from one-way LOCC to separability- or PPT-preserving operations. 
It admits the analytical expression
\begin{equation}\begin{aligned}
	C^\varepsilon(\phi,m) =  \max\left\{ \frac{2^{m+1}(1-\varepsilon)}{\|\ket{\phi}\|_{[2^m]}^2}  - 1,\, 1 \right\},
\end{aligned}\end{equation}
where
\begin{equation}\begin{aligned}
	\|\ket{\phi}\|_{[m]} \coloneqq \left\|\zeta^\downarrow_{1:m-k^\star}\right\|_{\ell_1} + \sqrt{k^\star} \left\|\zeta^\downarrow_{m-k^\star+1:d}\right\|_{\ell_2}
\end{aligned}\end{equation}
denotes the so-called $m$-distillation norm~\cite{regula_2019-2}. Here, $d$ is the local dimension, $\zeta^\downarrow_{1:k}$ stands for the vector consisting of the $k$ largest (by magnitude) Schmidt coefficients of $\ket{\phi}$, analogously $\ket{\zeta^\downarrow_{k+1:d}}$ denotes the vector of the $d-k$ smallest Schmidt coefficients of $\ket{\phi}$ with $\zeta^\downarrow_{1:0}$ being the zero vector, and 
\begin{equation}\begin{aligned}
	 k^\star \coloneqq \argmin_{1 \leq k \leq m} \frac{1}{k} \left\|\zeta^\downarrow_{m-k+1:d}\right\|_{\ell_2}^2.
\end{aligned}\end{equation}
\end{corollary}
\begin{proof}
Follows since the singlet fraction of any pure state is the same under these operations, and is given by $f(\phi,m) = 2^{-m} \|\ket{\phi}\|_{[2^m]}^2$~\cite[Theorem~15]{regula_2019-2}.
\end{proof}

Proposition~\ref{prop:singlet fraction} can also be employed to provide alternative exact expressions for the virtual resource distillation overhead in terms of entanglement measures previously studied. 

\begin{proposition}
Consider separability-preserving operations as the set of free operations. Then, for every state $\rho$,
\bal
 C^\varepsilon(\rho,m)=
 \max\left\{\frac{2(1-\varepsilon)}{G_{\mS}(\rho;2^m)}-1,1\right\}
\eal
where 
\bal
 G_{\mathcal{S}}(\rho;k)\!\coloneqq\! \sup\lset \tr[\rho W]\sbar \!0\leq W\leq I,\ \tr[W\sigma]\!\leq\! \frac{1}{k}, \forall \sigma\in\mathcal{S}\rset\!.
\eal

When $\rho$ and $\Phi^{\otimes m}$ are defined on the same space (i.e.\ they have the same dimension), 
\bal
 C^\varepsilon(\rho,m)=\max\left\{
 \frac{2^{m+1}(1-\varepsilon)}{1+R_{\mS}^g(\rho)}-1,1\right\}
\eal
where $R_{\mS}^g(\rho)$ is the generalized robustness of entanglement.

The same result holds for PPT-preserving operations by replacing $\mS$ with $\mS_{\rm PPT}$. 
\end{proposition}

\begin{proof}
It is known that, when $\mO$ is the set of separability-preserving operations, the maximum overlap $f_\mO(\rho,m)=\max_{\Lambda\in \mO}\tr(\Lambda(\rho)\Phi^{\otimes m})$ can be characterized by $f_\mO(\rho,m)=G_{\mS}(\rho;2^m)$ for a general state $\rho$ and $f_\mO(\rho,m)=(1+R_{\mS}^g(\rho))\,2^{-m}$ if $\rho$ and $\Phi^{\otimes m}$ act on the same space~\cite{Regula2020benchmarking}. 
The result then follows by applying Proposition~\ref{prop:singlet fraction}.

\end{proof}

\subsubsection{Isotropic states}

The exact characterization of virtual distillation overhead with respect to the singlet fraction allows us to derive analytical expressions for the class of isotropic states. 
Let $\rho_\alpha^I$ be an isotropic state defined as 
\bal
 \rho_\alpha^I(k)\coloneqq (1-\alpha)\, \Phi^{\otimes k} + \alpha\,\frac{I-\Phi^{\otimes k}}{2^{2k}-1}.
\eal

\begin{proposition}
Let $\mO$ be an arbitrary subset of separability-preserving or PPT-preserving operations that includes all separable operations. Then, for any $1 \leq m \leq k$,
	\begin{equation}\begin{aligned}
		C^\varepsilon(\rho_\alpha^I(k),m) = \begin{cases}
 \max\{ 2^{m+1}(1-\varepsilon) -1, \,1\} & \alpha\geq 1-2^{-k} \\
 \max\!\left\{ \frac{2(1-\varepsilon)}{1-\alpha c}-1 , \,1 \right\} & \alpha\leq 1-2^{-k},
 \end{cases}
  \end{aligned}\end{equation}
  where $c \coloneqq \left(2^k-2^{k-m}\right)/\left(2^k-1\right)$. 
In particular,
	\begin{equation}\begin{aligned}
		C^\varepsilon(\rho_\alpha^I(k),k) = \begin{cases}
 \max\{ 2^{k+1}(1-\varepsilon) -1,\, 1\} & \alpha\geq 1-2^{-k} \\
 \max\!\left\{ \frac{1+\alpha-2\varepsilon}{1-\alpha} , \, 1 \right\} & \alpha\leq 1-2^{-k}
 \end{cases}
 \end{aligned}\end{equation}
 and this latter result holds also for any set of free operations that contains local unitary operations assisted by shared classical randomness, in particular for LOCC.
\end{proposition}
\begin{proof}
By Proposition~\ref{prop:singlet fraction}, we have that
\begin{equation}\begin{aligned}\label{eq:isotropic_cost1}
	C^\varepsilon(\rho_\alpha^I(k),m) = \max\left\{\frac{2(1-\varepsilon)}{f_\mO(\rho^I_\alpha(k),m)}-1,1\right\}. 
\end{aligned}\end{equation}

For $\alpha\geq 1-2^{-k}$, the state $\rho_\alpha^I(k)$ is separable~\cite{Horodecki1999reduction}. Then $\Lambda(\rho_\alpha^I(k))$ is separable (respectively, PPT) for any separability-preserving (PPT-preserving) map $\Lambda$; using the fact that the overlap of $\psi^{\otimes m}$ with any PPT state is at most $2^{-m}$ and it is achieved by a separable state, it follows that $f_\mO(\rho_\alpha^I(k),m) = 2^{-m}$. For $\alpha < 1-2^{-k}$, we use the result of~\cite[Theorem~18]{regula_2019-2}, which states that
\begin{equation}\begin{aligned}
	f_\mO(\rho_\alpha^I(k), m ) = 1 - \alpha \frac{2^k-2^{k-m}}{2^k-1}.
\end{aligned}\end{equation}
Plugging these values into Eq.~\eqref{eq:isotropic_cost1} concludes the first part of the proof.

For the second part, we notice that $\Tr\left[ \rho_\alpha^I(k) \psi^{\otimes k} \right] = 1-\alpha$. Since the value of this overlap cannot be increased by any free operation~\cite[Corollary~15]{regula_2021-4}, we get $f_\mO(\rho_\alpha^I(k),k) = 1-\alpha$ as desired.
\end{proof}

\subsubsection{Bound entanglement does not help virtual distillation}

The virtual distillation overhead is governed by the size of the coefficients in a linear combination of accessible states that form a decomposition of a target state.
Intuitively, a smaller overhead could be realized if one were given a larger set of accessible states. 
In the context of entanglement distillation, if we are given some entangled state, the set of accessible states obtained by applying LOCC operations to the given entangled state is strictly greater than the set of separable states. 
This leads to a natural question: ``Is every entangled state useful for virtual distillation?''
The following result answers this question in the negative. 
A similar restriction holds also beyond LOCC operations, applying to all PPT-preserving maps.

\begin{proposition}

Consider LOCC as the set of free operations. Then, for every bound-entangled state $\rho$ and every $\varepsilon \in [0,1)$,
\bal
 C^\varepsilon(\rho,m) = C_{\mS}^\varepsilon(m)=\max\left\{2^{m+1}(1-\varepsilon)-1,1\right\}
 \label{eq:bound entanglement overhead}
\eal
where $C_{\mS}^\varepsilon(m)$ is the virtual distillation overhead for separable states, which takes the same value for all separable states.
If $\rho$ is PPT, then \eqref{eq:bound entanglement overhead} holds for an arbitrary subset of PPT-preserving operations that can prepare all separable states.

\end{proposition}

\begin{proof}

Recall that LOCC singlet fraction for states with zero distillable entanglement satisfies $f_\mO(\rho,m)\leq 2^{-m}$~\cite{Horodecki1999reduction}.
Noting that LOCC can prepare every separable state shows $f_\mO(\rho,m)=2^{-m}$.
Taking $\mO={\rm LOCC}$ in Proposition~\ref{prop:singlet fraction} then proves the first statement. 

To show the latter statement, note that for every PPT-preserving map $\Lambda$ and every PPT state $\rho$, 
\begin{equation}\begin{aligned}
    \Tr [ \Lambda(\rho) \Phi^{\otimes m} ] &\leq \max_{\sigma \in \mathrm{PPT}} \Tr [ \sigma \Phi^{\otimes m} ]\\
    &= \max_{\sigma \in \mathrm{PPT}} \Tr [ \sigma^\Gamma (\Phi^{\otimes m})^\Gamma ]\\
    &\leq \max_{\omega \in \mD} \Tr [ \omega\, (\Phi^\Gamma)^{\otimes m} ]\\
    &= 2^{-m},
\end{aligned}\end{equation}
where we used that the eigenvalues of $\Phi^\Gamma$ are $\pm \frac12$. Noting that the same overlap can be achieved by optimizing over separable states~\cite{shimony_1995} and invoking Proposition~\ref{prop:singlet fraction}, the result follows.
\end{proof}


\subsection{Coherence}\label{subsec:coherence}
We next consider the resource theory of coherence (superposition)~\cite{RevModPhys.89.041003}, where the set $\mI$ of free states consists of diagonal states with respect to a given preferred basis $\{\ket{i}\}_{i}$, i.e., $\mI\coloneqq\lset \sum_i p_i \dm{i}\sbar \sum_ip_i=1,\ p_i\geq 0,\forall i\rset$.

We first show that the virtual resource distillation overhead admits an analytical expression for an arbitrary single-qubit state. 

\begin{proposition}\label{prop:coherence_onequbit}
For an arbitrary qubit state $\rho$ and an arbitrary set $\mO$ of free operations that contains probabilistic applications of Pauli $X$ and $Z$, 
\bal
 C^\varepsilon(\rho,1) = \max\left\{\frac{1-2\varepsilon}{M_{l_1}(\rho)},1\right\}
\eal
where $M_{l_1}(\rho)=\sum_{i\neq j}|\braket{i|\rho|j}|$ is the $l_1$ norm of coherence.
\end{proposition}
\begin{proof}
Let $\rho=\begin{pmatrix} \alpha & \beta \\ \beta & 1-\alpha \end{pmatrix}$. 
We take $\beta\geq 0$ because any state can be brought to this form by the Pauli $Z$ operation, and $C^\varepsilon(\rho, 1)$ is invariant under such an operation. 
To see $C^\varepsilon(\rho,1)\leq \max\left\{\frac{1-2\varepsilon}{M_{l_1}(\rho)},1\right\}$, 
let $\mT(\cdot)\coloneqq \frac{1}{2}\cdot + \frac{1}{2}X\cdot X$ and $\mc Z\circ\mT(\cdot)\coloneqq \frac{1}{2}Z\cdot Z + \frac{1}{2}ZX\cdot XZ$.
Also, let $s(\varepsilon)=\frac{1-2\varepsilon}{4\beta}+\frac{1}{2}$. 
Then, consider the unit trace operator $\eta$ defined as
\bal
 \eta\coloneqq s(\varepsilon) \mT(\rho) - (s(\varepsilon)-1) \mc Z\circ\mT(\rho).
\eal
A direct computation reveals that $\eta=(1-\varepsilon)\dm{+}+\varepsilon\dm{-}$ and 
$\frac{1}{2}\|\eta-\dm{+}\|_1\leq \varepsilon$. 
When $\frac{1-2\varepsilon}{M_{l_1}(\rho)}\leq 1$, $s(\varepsilon)-1\leq 0$ and $\eta$ is a convex combination of $\mT(\rho)$ and $\mZ\circ\mT(\rho)$, giving $C^\varepsilon(\rho,1)\leq 1$. When $\frac{1-2\varepsilon}{M_{l_1}(\rho)}\geq 1$ 
, we obtain 
\bal
C^\varepsilon(\rho,1)&\leq 2s(\varepsilon)-1=\frac{1-2\varepsilon}{2\beta}=\frac{1-2\varepsilon}{M_{l_1}(\rho)}.
\eal

On the other hand, $C^\varepsilon(\rho,1)\geq \max\left\{\frac{1-2\varepsilon}{M_{l_1}(\rho)},1\right\}$ can be obtained from Theorem~\ref{prop:coherence single qubit analytical} below.
\end{proof}

We now characterize the virtual resource distillation overhead for general states via semidefinite programming.
Recall that a channel $\mE$ is called a maximally-incoherent operation (MIO) if it maps every incoherent state to an incoherent state, i.e., $\mE(\sigma)\in\mI\,\forall \sigma \in \mI$, and $\mE$ is called a dephasing-covariant operation (DIO) if it commutes with the completely dephasing map $\Delta(\cdot)=\sum_i \dm{i}\cdot\dm{i}$, i.e., $\mE\circ\Delta = \Delta\circ\mE$. 

\begin{theorem}
Let $\mO$ be the class of maximally incoherent operations (MIO) or dephasing-covariant incoherent operations (DIO), and let $\Delta(\cdot) = \sum_i \ket{i}\!\bra{i} \cdot \ket{i}\!\bra{i}$ be the diagonal map. Then, for any state $\rho$, the virtual distillation overhead is given by the semidefinite program
\begin{equation}\begin{aligned}\label{eq:coh_cost_exact}
	C^\varepsilon(\rho,m) = \min \big\{ \mu_+ + \mu_- \;\big|\;& 0 \leq Q_+ \leq \mu_+ \id,\; 0 \leq Q_- \leq \mu_- \id,\\
	& \Delta(Q_+) = \frac{\mu_+}{2^m} \id,\; \Delta(Q_-) = \frac{\mu_-}{2^m} \id,\\
	&\mu_+ - \mu_- =1,\\
	&\tr \rho (Q_+ - Q_-) \geq 1 -\varepsilon \big\}.
\end{aligned}\end{equation}
\end{theorem}
\begin{proof}
The lower bound on $C^\varepsilon(\rho,m)$ is essentially an application of Theorem~\ref{prop:general_bounds_cost_parallel}, but let us consider it explicitly for completeness. 
Consider then any feasible MIO protocol such that $\Lambda_\pm \in \mO$ and $\frac12 \| \lambda_+ \Lambda_+ (\rho) - \lambda_- \Lambda_-(\rho) - \ket{+}\!\bra{+}^{\otimes m}\|_1 \leq \varepsilon$. Define $Q_\pm = \lambda_\pm \Lambda_\pm^\dagger(\ket{+}\!\bra{+}^{\otimes m})$ and $\mu_\pm = \lambda_\pm$. 
For each $i$, we then have that
\begin{equation}\begin{aligned}
	\braket{i | Q_\pm | i} &= \lambda_\pm \tr \left[ \Lambda_{\pm}(\ket{i}\!\bra{i}) \ket{+}\!\bra{+}^{\otimes m} \right]\\
	&= \lambda_\pm \frac{1}{2^m}
\end{aligned}\end{equation}
where the last line follows since $\ket{+}\!\bra{+}^{\otimes m}$ has a constant overlap $2^{-m}$ with any incoherent state. This means in particular that $\Delta(Q_\pm) = \frac{\mu_\pm}{2^m} \id$. Verifying that other conditions are also satisfied due to the fact that each $\Lambda_\pm$ is a CPTP map, we have that $Q_\pm$ and $\mu_\pm$ are feasible solutions to \eqref{eq:coh_cost_exact}, leading to $C^\varepsilon(\rho,m) \geq \mu_+ - \mu_-$ for all MIO maps.

Conversely, let $Q_\pm$ be feasible solutions to \eqref{eq:coh_cost_exact}. Define the quantum channels
\begin{equation}\begin{aligned}
	\Lambda_\pm (\omega) = \tr \left(\frac{Q_\pm}{\mu_\pm} \omega\right) \ket{+}\!\bra{+}^{\otimes m} + \tr \left(\left[\id - \frac{Q_\pm}{\mu_\pm}\right] \omega\right) \frac{\id - \ket{+}\!\bra{+}^{\otimes m}}{2^m - 1}.
\end{aligned}\end{equation}
For any state $\omega$, we have that
\begin{equation}\begin{aligned}
	\Lambda_\pm \circ \Delta (\omega) &= \tr \left( \Delta\left[\frac{Q_\pm}{\mu_\pm}\right] \omega\right) \ket{+}\!\bra{+}^{\otimes m} \\ 
	&\quad+ \tr \left(\left[\id - \Delta\left(\frac{Q_\pm}{\mu_\pm}\right)\right] \omega\right) \frac{\id - \ket{+}\!\bra{+}^{\otimes m}}{2^m - 1}\\
	&= \frac{1}{2^m} \ket{+}\!\bra{+}^{\otimes m} + \frac{1}{2^m} \left(\id - \ket{+}\!\bra{+}^{\otimes m} \right)\\
	&= \frac{1}{2^m} I\\
	&= \Delta \circ \Lambda_\pm (\omega),
\end{aligned}\end{equation}
and so the constructed maps are both DIO. As the maps realize the virtual distillation of $\ket{+}\!\bra{+}^{\otimes m}$ from $\rho$ up to error $\varepsilon$, we have that $C^\varepsilon(\rho,m) \leq \mu_+ + \mu_-$ under DIO. Since ${\rm DIO} \subseteq {\rm MIO}$, the cost under MIO lower bounds the cost under DIO, and the result follows.
\end{proof}

For the case of a single qubit, the above program can be analytically solved for every number $m$ of target states.

\begin{theorem}\label{prop:coherence single qubit analytical}
Let $\mO$ be MIO or DIO. Then, for every single-qubit state $\rho$, every integer $m\geq 1$, and all $\varepsilon\in[0,1]$, the virtual distillation overhead is given by
\begin{equation}\begin{aligned}
	C^\varepsilon(\rho,m) = \max\left\{\frac{2^m(1-\varepsilon)-1}{M_{l_1}(\rho)},1\right\}
\end{aligned}\end{equation}
where $M_{l_1}(\rho)$ is the $l_1$-norm of coherence.
\end{theorem}
\begin{proof}

Without loss of generality, we assume that $\rho$ is on the $XZ$ plane in the Bloch coordinate, i.e., $\tr(\rho Y)=0$, and that $\braket{+|\rho|+}\geq \braket{-|\rho|-}$, as one can always bring any state onto the $XZ$ plane with $\braket{+|\rho|+}\geq \braket{-|\rho|-}$ by applying an incoherent unitary, and $C^\varepsilon(\rho,m)$ is invariant under any incoherent unitary. 

To constrain the form of $Q_{\pm}$, we consider a map $\Lambda_{XZ}$ defined by 
\bal
\Lambda_{XZ}(\cdot)\coloneqq \frac{1}{2} \tr(\cdot)\,I + \frac{1}{2} \tr(X\,\cdot)\,X + \frac{1}{2} \tr(Z\,\cdot)\,Z.
\eal
When applied to a quantum state, $\Lambda_{XZ}$ projects it to the $XZ$ plane. 
It is straightforward to check that it is unital, i.e., $\Lambda_{XZ}(I)=I$, and self dual, i.e., $\Lambda_{XZ}^\dagger = \Lambda_{XZ}$.
$\Lambda_{XZ}$ is also positive. This is because for an arbitrary single-qubit state $\sigma$, $\Lambda_{XZ}(\sigma)$ is also a valid state and hence positive. 
Since every positive operator acting on the single-qubit system is proportional to a quantum state, their positivity remains under $\Lambda_{XZ}$. 

By assumption, we have $\Lambda_{XZ}(\rho)=\rho$.
Let $Q_\pm^\star$ and $\mu_{\pm}^\star$ be the operators and real numbers that give the optimal solution of \eqref{eq:coh_cost_exact}.
Then, one can see that $\Lambda_{XZ}(Q_\pm^\star)$ also give the optimal solution $\mu_+^\star + \mu_-^\star$ as follows. 
First, $0\leq \Lambda_{XZ}(Q_\pm^\star)\leq \mu_+^\star I$ follows from that $\Lambda_{XZ}$ is positive and unital. 
$\Delta\circ \Lambda_{XZ}(Q_\pm^\star)=\frac{\mu_\pm}{2^m}I$ follows from the fact that $\Delta\circ \Lambda_{XZ} = \Lambda_{XZ} \circ \Delta$ and $\Lambda_{XZ}(I)=I$. 
Finally, 
\bal
 \tr\rho[\Lambda_{XZ}(Q_+^\star)-\Lambda_{XZ}(Q_-^\star)] &= \tr\Lambda_{XZ}^\dagger(\rho)(Q_+^\star-Q_-^\star)\\
 &= \tr\Lambda_{XZ}(\rho)(Q_+^\star-Q_-^\star)\\
 &=\tr\rho(Q_+^\star-Q_-^\star)\\
 &\geq 1-\varepsilon
\eal
where in the second equality we used that $\Lambda_{XZ}$ is self dual, and the third equality we used $\Lambda_{XZ}(\rho)=\rho$ by assumption. 
Thus, it suffices to restrict our attention to operators $Q_\pm$ such that $Q_\pm = c_\pm^I I + c_\pm^X X + c_\pm^Z Z$ where $c_\pm^P$ are some real numbers.  
In addition, the condition $\Delta(Q_\pm)\propto I$ further imposes $c_\pm^Z=0$. 
This allows us to write $Q_\pm$ of the form 
\bal
 Q_\pm = c_\pm^+\dm{+} + c_\pm^-\dm{-},\quad c_\pm^+,c_\pm^-\in\mathbb{R}.
\eal

In terms of this expression, \eqref{eq:coh_cost_exact} can be rewritten as 

\begin{equation}\begin{aligned}
	C^\varepsilon(\rho,m) = \min \big\{ \mu_+ + \mu_- \;\big|\;& 0 \leq c_\pm^+,c_\pm^- \leq \mu_\pm,\\
	&c_\pm^+ + c_\pm^- = \frac{\mu_\pm}{2^{m-1}},\\
	&\mu_+ - \mu_- =1,\\
	&(c_+^+-c_-^+)\braket{+|\rho|+}\\
	&+(c_+^--c_-^-)\braket{-|\rho|-}\geq 1 -\varepsilon \big\}.
\end{aligned}\end{equation}

Since $c_\pm^++c_\pm^-=\frac{\mu_\pm}{2^{m-1}}$ ensures $c_\pm^+, c_\pm^-\leq \mu_\pm$, we can further simplify it to

\begin{equation}\begin{aligned}
	C^\varepsilon(\rho,m) &= \min \big\{ 2^{m-1}(c_+^+ + c_+^-+ c_-^++c_-^-)\;\big|\; c_\pm^+, c_\pm^- \geq 0,\\
	&c_+^+ - c_-^+ + c_+^-- c_-^-=1/2^{m-1},\\
	&(c_+^+-c_-^+)\braket{+|\rho|+}+(c_+^--c_-^-)\braket{-|\rho|-}\geq 1 -\varepsilon \big\}.
\end{aligned}\end{equation}

Since the second and third constraints only involve $c_+^+-c_-^+$ and $c_+^--c_-^-$, the minimum occurs when $c_+^+c_-^+=c_+^-c_-^-=0$, because if $c_+^+, c_-^+\neq 0$ or $c_+^-, c_-^-\neq 0$, one can always make the objective function smaller while keeping the values of $c_+^+-c_-^+$ and $c_+^--c_-^-$.
Therefore, letting $\alpha\coloneqq c_+^+-c_-^+$ and $\beta\coloneqq c_+^--c_-^-$, we get

\begin{equation}\begin{aligned}
	C^\varepsilon(\rho,m) &= \min \big\{ 2^{m-1}(|\alpha|+|\beta|)\;\big|\; \alpha,\beta\in\mbR,\\
	&\alpha + \beta =1/2^{m-1},\\
	&\alpha\braket{+|\rho|+}+\beta\braket{-|\rho|-}\geq 1 -\varepsilon \big\}\\
 &= \min \big\{ 2^{m-1}(|\alpha|+|\beta|)\;\big|\; \alpha,\beta\in\mbR,\\
	&\alpha + \beta =1/2^{m-1},\\
	&(1+M_{l_1}(\rho))/2^{m}-\beta M_{l_1}(\rho)\geq 1 -\varepsilon \big\}
\end{aligned}\end{equation}
where in the second equality, we rewrote the left-hand side of the third constraint as $(\alpha+\beta)\braket{+|\rho|+}+\beta(-\braket{+|\rho|+}+\braket{-|\rho|-})$ and used the second constraint as well as the definition of $l_1$-norm of coherence $M_{l_1}(\rho)=\braket{+|\rho|+}-\braket{-|\rho|-}$ and the normalization $\braket{+|\rho|+}+\braket{-|\rho|-}=1$.

Suppose that $\alpha<0$. Then, the second constraint enforces $\beta=1/2^{m-1}-\alpha>0$. 
However, this is ensured to be suboptimal because by decreasing $\beta>0$ so that $\alpha\to 0$, the objective function can be monotonically reduced while the third constraint is ensured to be satisfied, as the left-hand side of the third constraint is monotonically decreasing with $\beta>0$.

Thus, we can set $\alpha\geq 0$ and write
\begin{equation}\begin{aligned}
	C^\varepsilon(\rho,m) &= \min \big\{ 1+2^{m-1}(-\beta+|\beta|)\;\big|\; \beta\leq 1/2^{m-1},\\
	&(1+M_{l_1}(\rho))/2^{m}-\beta M_{l_1}(\rho)\geq 1 -\varepsilon \big\}.
\end{aligned}\end{equation}
The form of the objective function implies that the minimum occurs at the largest $\beta$ that satisfies both constraints and particularly takes 1 if $\beta$ can become nonnegative.

When $1-\varepsilon<(1-M_{l_1}(\rho))/2^m$, we have
\bal
(1+M_{l_1}(\rho))/2^{m}-\beta M_{l_1}(\rho)\geq (1-M_{l_1}(\rho))/2^m
\eal
due to the first constraint. 
Therefore, the second constraint is always satisfied in this case, and the minimization occurs for any nonnegative $\beta$, which gives $C^\varepsilon(\rho,m)=1$.

Suppose now that $1-\varepsilon\geq (1-M_{l_1}(\rho))/2^m$.
In this case, the minimum occurs when the second constraint becomes equality, which gives 
\bal
 \beta = \frac{1}{M_{l_1}(\rho)}\left[\frac{1+M_{l_1}(\rho)}{2^m } - (1-\varepsilon)\right].
\eal
When $1-\varepsilon\leq (1+M_{l_1}(\rho))/2^m$, we have $\beta \geq 0$, which makes $C^\varepsilon(\rho,m)=1$. On the other hand, when $1-\varepsilon\geq (1+M_{l_1}(\rho))/2^m$, we have $\beta \leq 0$ and 
\bal
 C^\varepsilon(\rho,m)&= 1 - \frac{2^{m}}{M_{l_1}(\rho)}\left[\frac{1+M_{l_1}(\rho)}{2^m}-(1-\varepsilon)\right]\\
 &=\frac{2^m(1-\varepsilon)-1}{M_{l_1}(\rho)},
\eal
which is greater than or equal to 1. 

These cases are summarized as 
\bal
 C^\varepsilon(\rho,m) &=\max\left\{\frac{2^m(1-\varepsilon)-1}{M_{l_1}(\rho)},1\right\},
\eal
concluding the proof.

\end{proof}

The analytical expression for the distillation overhead in Theorem~\ref{prop:coherence single qubit analytical} allows for an exact characterization of the virtual resource distillation rate. 

\begin{corollary}
    Let $\mO$ be MIO or DIO. Then, for every single-qubit state $\rho$ and $\varepsilon\in[0,1]$, the virtual distillation rate is given by
\begin{equation}\begin{aligned}
	V^\varepsilon(\rho) = \max\left\{\frac{(\tilde m +1)M_{l_1}(\rho)^2}{[2^{\tilde m+1}(1-\varepsilon)-1]^2},\tilde m\right\}
\end{aligned}\end{equation}
where $M_{l_1}(\rho)$ is the $l_1$-norm of coherence, and 
\bal
 \tilde m\coloneqq \left\lfloor \log \left(\frac{M_{l_1}(\rho)+1}{1-\varepsilon}\right) \right\rfloor.
 \label{eq:coherence optimal number}
 \eal 
\end{corollary}
\begin{proof}
    We first remark that the function $m/C^\varepsilon(\rho,m)^2$ is monotonically decreasing with $m$, as we can write $m/C^\varepsilon(\rho,m)^2=M_{l_1}(\rho)^2 g(m)$ where $g(m)=m/(2^m(1-\varepsilon)-1)^2$ is the same function that appears in \eqref{eq:derivative} with $c=1-\varepsilon$, which is shown to be monotonically decreasing for all $m\geq 1$ and $c\geq 0$.
    This implies that the optimal $m$ for the maximization $\sup_m m/C^\varepsilon(\rho,m)^2$ is not greater than the smallest $m$ such that 
    \bal
     \frac{2^m(1-\varepsilon)-1}{M_{l_1}(\rho)}\geq 1.
     \label{eq:overhead larger}
    \eal
    On the other hand, for $m$ such that 
     \bal
     \frac{2^m(1-\varepsilon)-1}{M_{l_1}(\rho)}\leq 1,
     \label{eq:overhead smaller}
    \eal
    we have $m/C^\varepsilon(\rho,m)=m$, which is an increasing function with $m$.
    Therefore, letting $\tilde m$ be the maximum integer $m$ satisfying \eqref{eq:overhead smaller}, the maximum for $\sup_m m/C^\varepsilon(\rho,m)^2$ happens at either $m=\tilde m$, which gives $C^\varepsilon(\rho,m)=1$, or $m=\tilde m +1$, which gives $C^\varepsilon(\rho,m)=[2^{\tilde m+1}(1-\varepsilon)-1]/M_{l_1}(\rho)$.
    The result then follows by observing that $\tilde m$ can be explicitly obtained in the form of \eqref{eq:coherence optimal number}.
\end{proof}


\subsection{Magic}

Here, we discuss the resource theory of magic states, which is motivated by the scenario of fault-tolerant quantum computation~\cite{veitch2014resource,PhysRevLett.118.090501}.
The stabilizer states are the states that can be created by Clifford gates and classical randomness, and the resource theory of magic quantifies how much a given state deviates from the set of stabilizer states. 
Among them, the $T$ state defined by 
\bal
 T\coloneqq \dm{T} = \frac{I+(X+Y)/\sqrt{2}}{2},\quad \ket{T}\coloneqq\frac{1}{\sqrt{2}}(\ket{0}+e^{i\pi/4}\ket{1})
\eal
plays the major role, as access to the $T$ state together with stabilizer operations is sufficient to realize universal quantum computation. 
Therefore, magic state distillation protocols usually take $T$ state as the target state to synthesize. 
The important class of noisy input states for magic state distillation is the dephased $T$ state, 
\bal
 \rho_p^T\coloneqq (1-p)T + p I/2.
\eal
An arbitrary state can be brought into this form by the Clifford operation which applies $SX$ with $S={\rm diag}(1,i)$ with probability $1/2$. Therefore, the design of a magic state distillation protocol can focus on these specific noisy states as its input. 

We can apply our virtual resource distillation framework to this class of states, potentially providing better computational accuracy for algorithms run on fault-tolerant quantum computers that aim to obtain expectation values.
We note that a closely related setting was discussed in terms of a combination of quantum error mitigation and error correction methods~\cite{Lostaglio2021error,Piveteau2021error,Suzuki2021quantum}.
Our framework encompasses this strategy for magic state distillation as an application of the general approach of virtual resource distillation. 
In particular, the following analytical expressions for the virtual resource distillation extend the result in Ref.~\cite{Lostaglio2021error} to the regime with nonzero error.
Here, we let $\mO_{\rm STAB}$ and $\mF_{\rm STAB}$ denote the sets of stabilizer operations and states, respectively. 

\begin{proposition}\label{pro:magic T cost}

Let $p_{\rm th}\coloneqq 1/\sqrt{2}$ be the maximum value such that $\rho_{p_{\rm th}}^T\in \mF_{\rm STAB}$. Then, the virtual resource distillation overhead with respect to the target state $\ket{T}=(\ket{0}+e^{i\pi/4}\ket{1})/\sqrt{2}$ under stabilizer operations is characterized by
\bal
 C^\varepsilon(\rho^T_p,1)=\begin{cases}
\frac{1-2\varepsilon}{p_{\rm th}} & |p|\leq p_{\rm th}\\
\frac{1-2\varepsilon}{|p|} & |p|> p_{\rm th}
 \end{cases}.
 \eal
\end{proposition}

\begin{proof}
Recall that
\begin{equation}\begin{aligned}
C^\varepsilon(\rho^T_p,1)&=\min\lset\lambda_+ + \lambda_-\sbar \eta=\lambda_+\Lambda_+(\rho^T_p)-\lambda_-\Lambda_-(\rho^T_p),\right.\\ 
&\left.\, \Lambda_\pm\in \mO_{\rm STAB},\, \frac{1}{2}\|\eta-T\|_1\leq \varepsilon\rset.
\label{eq:cost noisy T definition}
\end{aligned}\end{equation}

Let $\mT(\cdot)=I\cdot I+SX\cdot (SX)^\dagger$ where $S={\rm diag}(1,i)$ is the phase gate and $X$ is the Pauli $X$ operator. 
Then, defining $\overline{T}=\dm{\overline{T}}$ with $\ket{\overline{T}}\coloneqq Z\ket{T}= \frac{1}{\sqrt{2}}(\ket{0}-e^{i\pi/4}\ket{1})$, we get for any state $\rho$ that 
\bal
\mT(\rho)=\tr(T\rho)\,T + \tr(\overline{T}\rho)\,\overline{T}.
\label{eq:T twirling}
\eal
Using $\mT(T) = T$ and the monotonicity of the trace norm under quantum channels, we have 
\bal
 \frac{1}{2}\|\mT(\eta)-T\|_1=\frac{1}{2}\|\mT(\eta)-\mT(T)\|_1 \leq \frac{1}{2}\|\eta-T\|_1. 
\eal
This ensures that if $\eta$ is a feasible solution of \eqref{eq:cost noisy T definition}, so is $\mT(\eta)$. 
Thus, the optimal solution $\eta$ for \eqref{eq:cost noisy T definition} can be restricted to the form $\eta=\lambda_+ \mT\circ\Lambda_+(\rho_p^T) - \lambda_- \mT\circ\Lambda_-(\rho_p^T)$. 

Note that $\rho^T_p = \frac{1+p}{2}T + \frac{1-p}{2}\overline{T}$. 
We assume $p\geq 0$ without the loss of generality because we can always apply $Z$ to flip the sign.
If $p\leq p_{\rm th}$, i.e., $\rho_p^T\in\mF_{\rm STAB}$, $\mT\circ\Lambda_\pm(\rho_p^T)$ is also a stabilizer state as $\mT\circ\Lambda_\pm\in\mO_{\rm STAB}$. Therefore, the optimal solutions should take 
\bal
\mT\circ\Lambda_+(\rho_p^T)&=\rho_{ p_{\rm th}}^T=\frac{1+p_{\rm th}}{2}T + \frac{1-p_{\rm th}}{2}\overline{T}\\
\mT\circ\Lambda_-(\rho_p^T)&=\rho_{ -p_{\rm th}}^T=\frac{1-p_{\rm th}}{2}T + \frac{1+p_{\rm th}}{2}\overline{T}.
\eal
This form specifies the optimal $\eta$ as
\bal
 \eta = \frac{1+(\lambda_++\lambda_-)p_{\rm th}}{2}\, T + \frac{1-(\lambda_++\lambda_-)p_{\rm th}}{2}\,\overline{T}.
\eal
Therefore, the condition $\frac{1}{2}\|\eta-T\|_1\leq \varepsilon$ is equivalent to $\frac{1-(\lambda_++\lambda_-)p_{\rm th}}{2}\leq \varepsilon$.
The optimal $\lambda_\pm$ under this condition gives $C^\varepsilon(\rho^T_p,1)=\frac{1-2\varepsilon}{p_{\rm th}}$.

When $p>p_{\rm th}$, there always exists $\Lambda\in\mO_{\rm STAB}$ such that $\rho_{p'}^T=\mT\circ \Lambda(\rho_p^T)$ for every $p'\in[-p,p]$---such a $\Lambda$ is either realized by mixing the maximally mixed state with $\rho^T_p$, or applying $Z$ to $\rho^T_p$ to make $\rho^T_{-p}$ and mixing the maximally mixed state to it. 
On the other hand, no $p'\not\in[-p,p]$ can be realized because otherwise, the free operation $\mT\circ\Lambda\in\mO_{\rm STAB}$ would increase a resource monotone (e.g., trace-distance measure $R_{\rm tr}(\rho)\coloneqq \min_{\sigma\in\mF_{\rm STAB}}\frac{1}{2}\|\rho-\sigma\|_1$).

Therefore, following a similar argument for the case of $p\leq p_{\rm th}$, the optimal solutions should take 
\bal
\mT\circ\Lambda_+(\rho_p^T)&=\rho_{ p}^T=\frac{1+p}{2}T + \frac{1-p}{2}\overline{T}\\
\mT\circ\Lambda_-(\rho_p^T)&=\rho_{ -p}^T=\frac{1-p}{2}T + \frac{1+p}{2}\overline{T}.
\eal
This form specifies the optimal $\eta$ as
\bal
 \eta = \frac{1+(\lambda_++\lambda_-)p}{2}\, T + \frac{1-(\lambda_++\lambda_-)p}{2}\,\overline{T}.
\eal
Therefore, the condition $\frac{1}{2}\|\eta-T\|_1\leq \varepsilon$ is equivalent to $\frac{1-(\lambda_++\lambda_-)p}{2}\leq \varepsilon$.
The optimal $\lambda_\pm$ under this condition gives $C^\varepsilon(\rho^T_p,1)=\frac{1-2\varepsilon}{p}$.

\end{proof}

For qutrit states, one of the magic states that maximize the negativity of the discrete Wigner function~\cite{veitch2012negative} is known as the Strange state:
\bal
 S= \dm{S},\quad \ket{S}\coloneqq\frac{1}{\sqrt{2}}(\ket{1}-\ket{2}).
\eal
We can characterize overhead for the strange state as follows.
\begin{proposition}\label{prop:magic overlap}
Let $\mc O_{\rm STAB}$ be the set of stabilizer operations, and consider $S$ as the target state for virtual resource distillation. 
Then, we have
\bal
 C^\varepsilon(\rho,m) = \max\left\{\frac{2(1-\varepsilon)}{f_{\mc O_{\rm STAB}}(\rho,m)} -1,1\right\}.
\eal
where $f_{\mc O_{\rm STAB}}(\rho,m)\coloneqq\max_{\Lambda\in\mO_{\rm STAB}}\Tr[\Lambda(\rho)S^{\otimes m}]$ be the maximum overlap with the copies of $S$ state. 
\end{proposition}
\begin{proof}
The Strange state admits a free twirling operation of the form \eqref{eq:generalized twirling}~\cite{veitch2014resource, Takagi2021oneshot} and hence the result follows as a consequence of Theorem~\ref{prop:cost twirling}. 
In particular, the twirling operation for the $S$ state is the random application of Clifford unitaries that correspond to elements in the special linear group ${\rm SL}(2,\mathbb{Z}_3)$.

\end{proof}


\subsection{Quantum memory and error mitigation}

Next, we consider an example of virtual distillation in channel theories. We consider a case that distills a $d$-dimensional noisy memory $\mc M$ into an ideal $d$-dimensional memory $\mI_d$ --- that is, the identity channel, which perfectly preserves any quantum system --- by adding extra gates after the application of the memory channel.

The distillation overhead with respect to the target channel $\mI_d$ is
\begin{equation}
	\begin{aligned}
			&C^{\varepsilon}(\mc M,m) \\
   &= \min_{\{\mN_i\}_i, \{\lambda_i\}_i}\lset\sum_i|\lambda_i|\sbar\frac{1}{2}\left\|\sum_i \lambda_i \mc N_i\circ\mc M-\mc I_d^{\otimes m}\right\|_{\diamond}\le \varepsilon\rset.
	\end{aligned}
\label{eq:overhead memory}
\end{equation}
We can imagine this as the case where we have a quantum gate $\mc U$ followed by a noise $\mc M$. Then we aim to apply extra gates so that the noise is canceled.  
The strategy we consider here thus contains several error mitigation methods~\cite{PhysRevLett.119.180509,Li2017}, and similar performance analysis was also studied~\cite{Takagi2021optimal,Regula2021operational,Jiang2021physical}.  

The diamond norm $\left\|\sum_i \lambda_i \mc N_i\circ\mc M-\mI_d^{\otimes m}\right\|_{\diamond}$ can be written as a semidefinite program as follows~\cite{watrous2009semidefinite}.
\begin{equation}
	\begin{aligned}
			&\left\|\sum_i \lambda_i \mc N_i\circ\mc M-\mI_d^{\otimes m}\right\|_{\diamond} \\
   &= \min\Big\{2\lambda-\sum_i{\lambda_i}+1\,\Big\vert\,\lambda J_{\mc E}\geq \sum_i \lambda_i J_{\mc N_i\circ\mc M}-J_{\mI_d^{\otimes m}},\\
   &\hphantom{= \min\Big\{2\lambda-\sum_i{\lambda_i}+1\,\Big\vert\,}\mc E\in {\rm CPTP}\Big\} 
	\end{aligned}
\end{equation}
where $J_{\Lambda}$ denotes the Choi state for a channel $\Lambda$.
Therefore, we can write the virtual distillation overhead as a semidefinite program as
\begin{equation}
	\begin{aligned}
			C^{\varepsilon}(\mc M,m) &= \min_{\mc N_{\pm}, \lambda_{\pm}}\lset{\lambda_++\lambda_-}\sbar 2\lambda-\sum_i{\lambda_i}+1\le \varepsilon,\right.\\ 
			&\left.\lambda J_{\mc E}\geq \sum_i \lambda_i J_{\mc N_i\circ\mc M}-J_{\mI_d^{\otimes m}}, \mc E\in {\rm CPTP}\rset.
	\end{aligned}
\end{equation}

In Fig.~\ref{fig:memory}, we compute the overhead for depolarising channels, dephasing channels, and stochastic replacement channels for $m=1$.

\begin{figure}[h]\centering
  \includegraphics[width=\columnwidth]{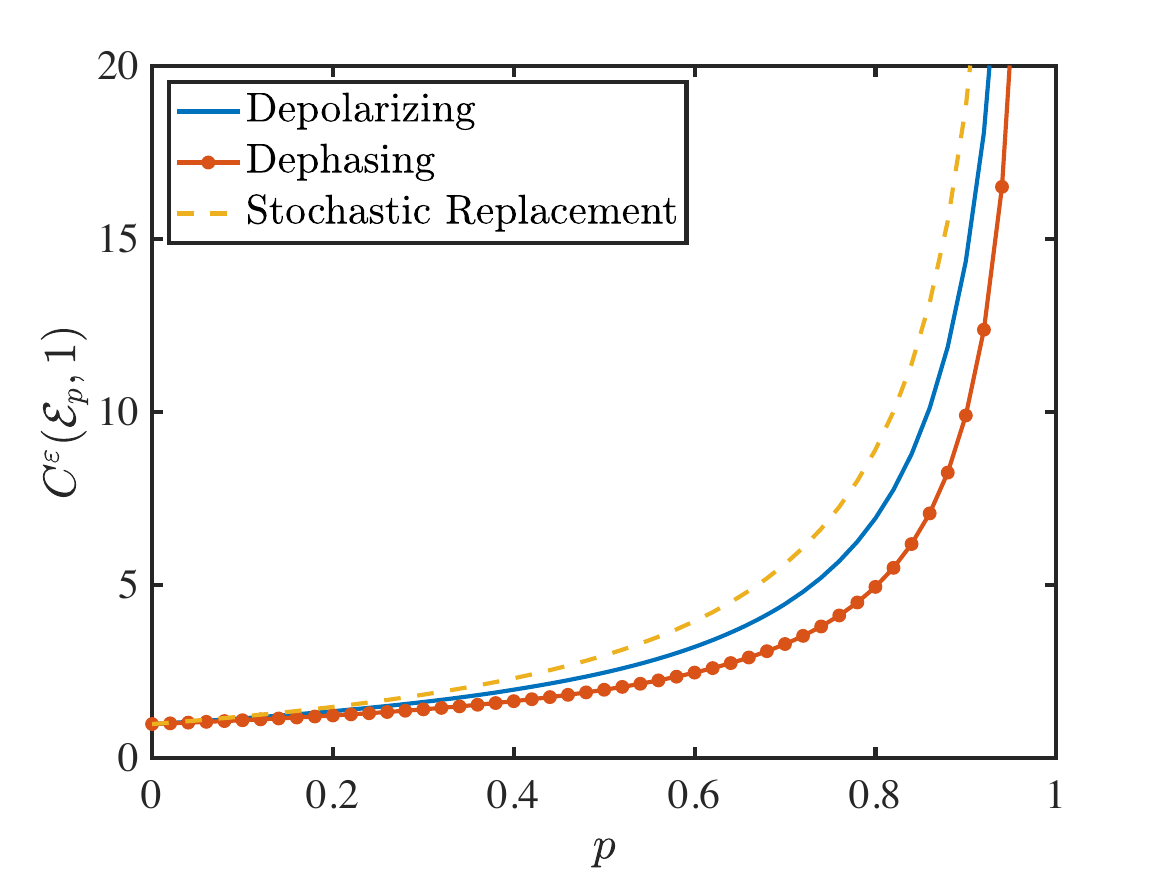}
  \caption{
  Distillation overhead for depolarising channels $\mE_p(\rho)=p\rho+(1-p)I/2$, dephasing channels $\mE_p(\rho)=p\rho+(1-p)Z\rho Z$, and stochastic replacement channels $ \mE_p(\rho) = p\rho+(1-p)|{0}\rangle\langle{0}|$. Here we consider $\varepsilon=0.01$.}\label{fig:memory}
\end{figure}


\subsection{Quantum communication}

A similar argument can be applied to the setting of quantum communication, in which Alice aims to send a quantum state to Bob via a noisy channel $\mE$.
In the usual setting of quantum communication, Alice and Bob apply additional quantum operations available to them so that Bob can recover the quantum state that was initially in Alice's hands. 
The rate at which noiseless qubits can be successfully sent is known as the quantum capacity of a noisy channel. 

We remark that this process can be considered as a channel distillation --- the operations Alice and Bob can apply are considered as free operations applied to a noisy channel, with the overall process constructing a superchannel that transforms a channel $\mE$ to (approximately) the identity channel.
From this perspective, quantum capacity coincides with the distillation rate with respect to available encoding and decoding operations constructing free operations, for which resource-theoretic tools can be employed to study the properties. 
This resource-theoretic view of quantum communication has recently been actively studied and provided new insights into the theory of quantum communication~\cite{takagi_2020,Regula2021fundamental,Fang2020no-go,Regula2021oneshot}.

The framework of virtual resource distillation allows us to extend the conventional setting of quantum communication.
The specific restrictions on communication scenarios, reflected in the choice of free operations applied to noisy quantum channels, largely depend on the assisting resource available --- such as non-signaling/entanglement assistance~\cite{bennett_2002,leung_2015} or classical communication assistance~\cite{Bennett96}.  
Since our framework is applicable to general resource theories, i.e., any choice of free operations, the technique of virtual resource distillation can be applied to communication settings with very general types of physical restrictions.
Here, we consider the setting of the most physical relevance, which is an unassisted setting in which Alice and Bob can only make local operations on their side. 

A subtlety that arises in virtual resource distillation for unassisted communication is that postprocessing after the measurement made by Bob requires one bit of preshared randomness (or classical communication from Alice to Bob) to agree on the operations they apply on each side. 
To avoid this, here we consider a significantly weaker setting, in which only Bob applies a probabilistic operation followed by measurement and classical postprocessing. 
In this setting, Bob can generate a random bit on his own and choose his operation and corresponding classical postprocessing. 

This reduces the estimation of virtual resource distillation rate and overhead to a framework almost identical to the one discussed in the last subsection. 
Here, our target channel is the identity channel $\mI$. 
For a given noisy channel $\mE$, the virtual resource distillation overhead with respect to the target channel $\mI$ is then characterized by $C^\varepsilon(\mE,m)$ defined in \eqref{eq:overhead memory} with $d$ being the dimension of the space that $\mE$ acts on.

Let us now focus on the evaluation of distillation overhead for the case of $\varepsilon=0$ and $m=1$. 
In this case, the overhead can be written as 
\bal\label{eq:virt_capacity}
 C^0(\mE,1)=\min_{\{\mN_i\}_i,\{\lambda_i\}_i}\lset\sum_i|\lambda_i|\sbar \mE^{-1}= \sum_i \lambda_i \mN_i\rset
\eal
assuming that the inverse map for $\mE$ exists. 
This overhead can then give a lower bound for the virtual resource distillation rate as $V^\varepsilon(\mE)\geq 1/C^0(\mE,1)^2$ for every $\varepsilon\in[0,1]$.

This quantity was studied in Ref.~\cite{Jiang2021physical,Regula2021operational} and shown to coincide with the diamond norm of the inverse map $\|\mE^{-1}\|_\diamond$~\cite{Regula2021operational}.
The analytical expressions of $C^0(\mE,1)$ for some noisy channels of interest were then obtained. 
For instance, for the $d$-dimensional depolarizing channel $\mD_p(\rho)\coloneqq (1-p)\rho + pI/d$, we have
\bal
 C^0(\mD_p,1) = \frac{1+(1-2/d^2)p}{1-p}.
\eal
This in particular provides the lower bound for the virtual distillation rate for a qubit depolarizing channel as  
\bal
 V^\varepsilon(\mD_p)\geq \left(\frac{1-p}{1+p/2}\right)^2. 
\eal

The exact variant of virtual distillation overhead in \eqref{eq:virt_capacity} is conceptually similar to zero-error quantum communication, where no error is allowed in the protocol.
We note, however, that a direct comparison of the lower bound we obtained for the virtual distillation rate and the quantum capacity may not be fair, because the computation of quantum capacity assumes that: (1) Alice also applies her operation, and (2) asymptotically many channel uses are allowed. More generally, the error in communication can be nonzero as long as it vanishes in the limit of infinitely many channel uses.
We can see that our lower bound $1/C^0(\mD_p,1)^2$ can already be significantly greater than the quantum capacity. 
To see this, recall that the quantum capacity for qubit depolarizing channel $Q(\mD_p)$ has a simple upper bound $Q(\mD_p)\leq 1-4p$ for $p\leq 1/4$ and $Q(\mD_p)=0$ for $p\leq 1/4$~\cite{Smith2008additive}. 
It is straightforward to check that $1-4p< 1/C^0(\mD_p,1)^2$ for $p\in (0,1]$, which results in $Q(\mD_p)<V^0(\mD_p)$ for the whole range of $p$. 
In particular, $V^0(\mD_p)>0$ for all $p\in[0,1]$, which has a stark contrast to the quantum capacity, which becomes $Q(\mD_p)=0$ for $p>1/4$.

As another example, take the qubit amplitude damping channel $\mA_\gamma(\rho)\coloneqq A_0\cdot A_0^\dagger + A_1\cdot A_1^\dagger$ with $A_0\coloneqq\dm{0}+\sqrt{1-\gamma}\dm{1}$ and $A_1\coloneqq \sqrt{\gamma}\ketbra{0}{1}$. 
Using the results in Refs.~\cite{Regula2021operational,Jiang2021physical}, the overhead can be computed as 
\bal
 C^0(\mA_\gamma,1) = \frac{1+\gamma}{1-\gamma},
\eal
which gives a lower bound for the virtual distillation rate as $V^\varepsilon(\mA_\gamma)\geq \left(\frac{1-\gamma}{1+\gamma}\right)^2$.
On the other hand, the quantum capacity of amplitude damping is known to be \cite{Giovannetti2005information}
\bal
 Q(\mA_\gamma)=\max_t [h_2((1-\gamma)t) - h_2(\gamma t)]
\eal
where $h_2(p)\coloneqq -p\log p - (1-p)\log (1-p)$ is the binary entropy.
Fig.~\ref{fig:amplitude_dampling} plots $Q(\mA_\gamma)$ and $1/C^0(\mA_\gamma)^2$ for $\gamma\in[0,1]$.
This shows that the virtual distillation rate is ensured to be greater than quantum capacity for $\gamma\geq 0.4$ and it remains nonzero while $Q(\mA_\gamma)=0$ for $\gamma\geq 1/2$.

\begin{figure}
    \centering
    \includegraphics[width=0.8\columnwidth]{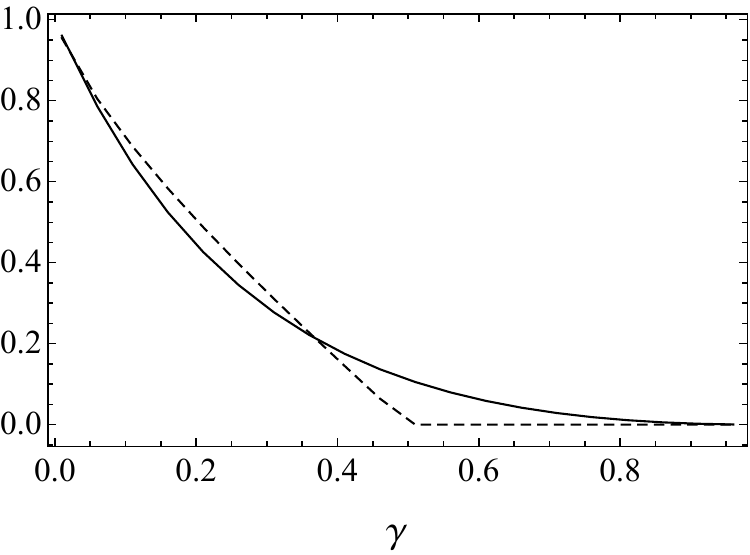}
    \caption{The lower bound $1/C^0(\mA_\gamma,1)^2$ for the virtual distillation rate (solid) and quantum capacity $Q(\mA_\gamma)$ (dashed) for the qubit amplitude damping channel $\mA_\gamma$.}
    \label{fig:amplitude_dampling}
\end{figure}


\subsection{Dephased non-Markovian processes}

Let us now discuss an application of virtual resource distillation in comb theories. 
Quantum combs become most relevant when the system and environment interact with each other, where one does not have control over the environment.
Non-Markovian dynamics particularly appear when the environment sustains quantum memory over multiple time steps. 
In such a scenario, physically accessible operations, which we take as free operations, should be quantum combs that only act on the accessible system. 
Here, we discuss an example where virtual resource distillation enables us to distill the environment comb that has the perfect quantum memory from the one with inferior memory.

Consider an $L$-step non-Markovian process that involves qubit system and environment that interact with each other via two CNOT gates at every time step. 
We let $S_j$ and $E_j$ refer to the system and environment at $j$\,th time step.
Every set of CNOT gates is followed by a partial dephasing channel $\mZ_p(\rho)\coloneqq (1-p)\rho + Z\rho Z$ in the environment, where $Z$ is the Pauli $Z$ operator (Fig.~\ref{fig:comb}).  

A quantum comb can be described by its Choi operator, which corresponds to the state obtained by inputting one end of a maximally entangled state into every input port of the comb (up to normalization)~\cite{Chiribella2008quantum}. The Choi operator of our given object is written as 
\bal
J_\Upsilon = \star_{j=1}^L\left(J_{\mZ_p}^{E_j} \star J_{\rm CNOT}^{S_{j} E_{j}}\right)
\eal
where $J_{\rm CNOT}^{S_{j}E_{j}}$ is the Choi operator for two CNOT gates at $j$\,th time step, and $J_{\mZ_p}^{E_j}$ is the Choi operator for $\mZ_p$. 
For channels $\mM:A\to B$ and $\mE:B\to C$, $J_\mE\star J_\mM\coloneqq\Tr_B\left[J_\mE^{T_B}J_\mM\right]$ refers to the link product, which gives the Choi operator for the concatenated channel $\mE\circ\mM$~\cite{Chiribella2008quantum}.

The dephasing channel degrades the quantum memory and decoheres quantum states over time.  
Our goal is to remove the effect of this dephasing noise in the environment by applying operations in the system.
Therefore, we set our target comb $\Theta$ as 
\bal
J_\Theta = \star_{j=1}^L\left(J_{\mI}^{E_j} \star J_{\rm CNOT}^{S_{j} E_{j}}\right).
\eal 

We would like to find a set $\{\Lambda_i\}_i$ of free operations, i.e., quantum combs that only act on the system side, so that $\Theta=\sum_i \lambda_i \Lambda_i(\Upsilon)$ for some real numbers $\{\lambda_i\}_i$.
To this end, it is useful to notice that the inverse map $\mZ_p^{-1}$ for the partial dephasing is decomposed as~\cite{PhysRevLett.119.180509}
\bal
 \mZ_p^{-1} = \frac{1-p}{1-2p}\mI - \frac{p}{1-2p}\mZ
\eal
and the optimal overhead is realized by this decomposition~\cite{Takagi2021optimal,Regula2021operational,Jiang2021physical}, which gives $C^0(\mZ_p,1)=1/(1-2p)$ with respect to the target channel $\mI$ in light of the discussion in the previous two subsections. 
The implementation of this inverse map is realized by applying $\mZ$ at probability $p$ followed by postprocessing, i.e., multiplying $1/(1-2p)$ to the measurement outcome with a possible sign factor if Pauli $Z$ is applied.

We now observe that the same action can be made on the environment by applying a $Z$ operator on the system side. 
Namely, we apply $Z$ operators both before and after the CNOTs at probability $p$, and we do not apply anything at probability $1-p$ (Fig.~\ref{fig:comb}). 
When a $Z$ operator is applied, the action of $Z$ propagates to the environment through the second CNOT gate, while the effect on the system side cancels out by the second $Z$ operator after the CNOTs.  

\begin{figure}
    \centering
    \includegraphics[width=\columnwidth]{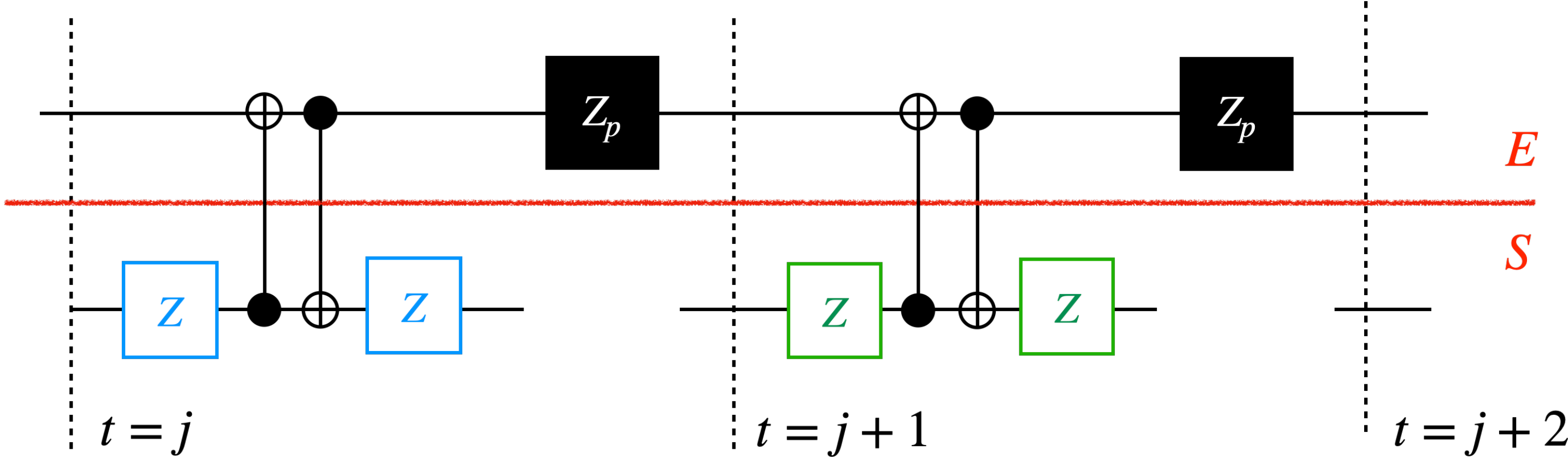}
    \caption{A diagram showing $j$\,th and $j+1$\,th time steps. The $Z$ operations in blue and green are stochastic operations applied at probability $p$ in the virtual distillation process. Two $Z$'s in the same color are either simultaneously applied or are not applied at all, while the application of blue and green $Z$'s are independent. }
    \label{fig:comb}
\end{figure}

We can independently apply the same procedure at every time step, which constructs $2^L$ free operations $\{\Lambda_{\vec{i}}\}_{\vec{i}\in\{0,1\}^L}$ where the location of 1's in $\vec{i}$ specifies the time steps at which $Z$ operators are applied. 
Letting $|\vec{i}|$ denote the number of 1's in $\vec{i}$, the desired linear decomposition of the target comb is written as
\bal
 \Theta = \sum_{\vec{i}\in\{0,1\}^L} (-1)^{|\vec{i}|} (1-p)^{L-|\vec{i}|}p^{|\vec{i}|}\Lambda_{\vec{i}}(\Upsilon),
\eal
which gives 
\bal
 C^{0}(\Upsilon,1)\leq \frac{1}{(1-2p)^L}.
\eal
We conjecture that equality holds because this is essentially the most efficient way to counteract the dephasing, although we leave the full investigation to future work. 

The implementation of the virtual resource distillation then is realized as follows. 
\begin{enumerate}
    \item At each step, apply $Z$ operators before and after the CNOTs at probability $p$ and do nothing at probability $1-p$. Record which operation was applied. 
    \item Multiply $(-1)^{\rm sgn}(1-2p)^L$ to the measurement outcome, where sgn takes 0 if the total number of time steps at which $Z$ operators were applied was even, and 1 if odd. 
    \item Repeat the same procedure many times and take a sample average of the postprocessed measurement outcomes.
\end{enumerate}

This example ensures that our virtual resource distillation framework in comb theories can be applied to memory preservation in the environment by manipulating the accessible system only. 
This may be seen as an error mitigation protocol applied to non-Markovian dynamics, which was previously studied in several other settings~\cite{Hokoshima2021relationship,Ho2021enhancing}.


\section{Conclusions}

We presented the framework of virtual resource distillation applicable to general resource theories with an arbitrary set of convex free objects and free operations, including general types of quantum objects such as quantum states, channels, and higher-order processes represented by quantum combs. 
We derived various expressions and bounds for virtual resource distillation rate and overhead, both in general settings and concrete theories of practical interest, demonstrating its versatility and broad applicability.  

Promising future directions include obtaining explicit evaluations of the performance of probabilistic virtual resource distillation, which may find use in practical settings that are not possible to characterize using only deterministic protocols. 
It will also be interesting to consider experimental implementations of virtual resource distillation that are possible on today's  quantum devices.

A further open question is the relation between virtual distillation rates and asymptotic rates of conventional distillation. Although, as remarked earlier and in~\cite{maintext}, the latter are very different from our approach --- requiring in particular coherent manipulation of many-copy input states $\rho^{\otimes n}$ with an unbounded number of copies --- it would nevertheless be interesting to understand whether virtual distillation can already improve on such rates. Furthermore, the many-copy extension of our virtual framework, i.e.\ the behaviour of $V^\varepsilon(\rho^{\otimes n})$ when more copies of the input state $\rho$ can be manipulated coherently, is an interesting question of its own. Although this sacrifices the `experimentally friendly' character of virtual distillation protocols,  it could be useful to understand to what extent virtual distillation capabilities can be improved through such many-copy protocols, leading to the delineation of the ultimate limits of virtual quantum resource distillation.

\begin{acknowledgments}
We thank Suguru Endo, Patrick Hayden, Jayne Thompson, and Mark M.\ Wilde for insightful discussions.
This work is supported by the National Natural Science Foundation of China Grant No.~12175003, the Singapore Ministry of Education Tier 1 Grant RG77/22, the National Research Foundation, Singapore, and Agency for Science, Technology and Research (A*STAR) under its QEP2.0 programme (NRF2021-QEP2-02-P06) and the Singapore Ministry of Education Tier 2 Grant MOE-T2EP50221-0005,  B.R.\ is partially supported by the Japan Society for the Promotion of Science (JSPS) KAKENHI Grant No.\ 22KF0067. R.T.\ is supported by the Lee Kuan Yew Postdoctoral Fellowship at Nanyang Technological University Singapore.
\end{acknowledgments}

\appendix

\section{Dual formulation of $C^\varepsilon$}
\label{app:dual}

Here we prove Eq.~\eqref{eq:overhead dual}, that is, establish a dual form of the virtual distillation cost $C^\varepsilon$.

To begin, we write
\begin{align}
  C^\varepsilon(\rho,m) &= \inf_{\substack{\widetilde\tau \sim_\varepsilon \tau^{\otimes m}}} \inf \big\{\lambda_++\lambda_- \;\big|\; \widetilde\tau = \lambda_+ \Lambda_+(\rho) - \lambda_- \Lambda_-(\rho),\nonumber\\
   &\qquad\qquad\qquad\lambda_\pm\geq 0,\ \lambda_+-\lambda_- = 1,\ \Lambda_\pm \in \mO\big\}\nonumber\\
   &= \inf_{\substack{\widetilde\tau \sim_\varepsilon \tau^{\otimes m}\\\Tr\widetilde\tau = 1}} \inf \big\{\lambda_++\lambda_- \;\big|\; \widetilde\tau = \lambda_+ \Lambda_+(\rho) - \lambda_- \Lambda_-(\rho),\nonumber\\
   &\qquad\qquad\qquad\lambda_\pm\geq 0,\ \Lambda_\pm \in \mO\big\}\\
   &= \inf_{\substack{\widetilde\tau \sim_\varepsilon \tau^{\otimes m}\nonumber\\\Tr\widetilde\tau = 1}} \inf \big\{\Tr\widetilde\Lambda_+(\rho) + \Tr\widetilde\Lambda_-(\rho) \;\big|\; \nonumber\\
   &\qquad\qquad\qquad\widetilde\tau = \widetilde \Lambda_+(\rho) - \widetilde\Lambda_-(\rho),\nonumber\\
   &\qquad\qquad\qquad \widetilde\Lambda_\pm \in \operatorname{cone}(\mO)\big\},\nonumber
\end{align}
where $\operatorname{cone}(\mO) = \lset \lambda \Lambda \sbar \lambda \geq 0,\; \Lambda \in \mO \rset$. We will now take the Lagrange dual of the inner minimisation. The Lagrangian of this problem is
\begin{align}
\mL\left(\widetilde\Lambda_\pm; H, X, Y\right) &=  \Tr\widetilde\Lambda_+(\rho) + \Tr\widetilde\Lambda_-(\rho)\nonumber \\
&\quad -\Tr\left( H \left[ \widetilde \Lambda_+(\rho) - \widetilde\Lambda_-(\rho)-\widetilde\tau\right]\right)\\
&\quad - \Tr X J_{\widetilde\Lambda_+} - \Tr Y J_{\widetilde\Lambda_-},\nonumber
\end{align}
where $J_{\Lambda}$ denotes the Choi operator of the corresponding map, and $H, X, Y$ are Lagrange multipliers satisfying $\Tr X J_{\Lambda} \geq 0 \; \forall \Lambda \in \mO$ and analogously for $Y$. 
Using the Choi--Jamiolkowski isomorphism, we can rewrite this as
\bal
\mL\left(\widetilde\Lambda_\pm; H, X, Y\right) &=  \Tr (\id \otimes \rho^T) J_{\widetilde\Lambda_+} + \Tr (\id \otimes \rho^T) J_{\widetilde\Lambda_-} \\
&\hspace{-10pt} + \Tr H \widetilde\tau - \Tr (H\otimes\rho^T) J_{\widetilde \Lambda_+} + \Tr (H\otimes \rho^T)J_{\widetilde\Lambda_-(\rho)}\\
&\hspace{-10pt} - \Tr X J_{\widetilde\Lambda_+} - \Tr Y J_{\widetilde\Lambda_-}.
\eal

By definition, the dual problem is then~\cite{Boyd2004convex}
\begin{align}
 &\sup_{\substack{H \in \mathrm{Herm}\\X \,:\, \Tr X J_{\Lambda} \geq 0 \; \forall \Lambda \in \mO\\Y \,:\, \Tr Y J_{\Lambda} \geq 0 \; \forall \Lambda \in \mO}} 
\;\inf_{\widetilde\Lambda_\pm \in \mathrm{Herm}} \;\mL\left(\widetilde\Lambda_\pm; H, X, Y\right)\nonumber\\
&= \sup \big\{ \Tr H \widetilde \tau \;\big|\; \Tr ([\id - H]\otimes\rho^T) J_{\Lambda} \geq 0 \; \forall \Lambda \in \mO,\\
&\qquad\qquad\qquad\; \Tr ([\id + H]\otimes\rho^T) J_{ \Lambda} \geq 0 \; \forall \Lambda \in \mO \big\}\nonumber\\
&= \sup \big\{ \Tr H \widetilde \tau \;\big|\; -1 \leq \Tr H \Lambda(\rho) \leq 1 \; \forall \Lambda \in \mO \big\}.\nonumber
\end{align}
Since $H=0$ is strictly feasible to the above, by Slater's theorem we have that the optimal values of the primal and dual optimisation problems are equal.
A change of variables $W \coloneqq 2 H - I$ gives the form stated in Eq.~\eqref{eq:overhead dual}. 

\bibliographystyle{apsrmp4-2}
\bibliography{myref}

\begin{thebibliography}{85}%
\makeatletter
\providecommand \@ifxundefined [1]{%
 \@ifx{#1\undefined}
}%
\providecommand \@ifnum [1]{%
 \ifnum #1\expandafter \@firstoftwo
 \else \expandafter \@secondoftwo
 \fi
}%
\providecommand \@ifx [1]{%
 \ifx #1\expandafter \@firstoftwo
 \else \expandafter \@secondoftwo
 \fi
}%
\providecommand \natexlab [1]{#1}%
\providecommand \emph  [1]{``#1''}%
\providecommand \bibnamefont  [1]{#1}%
\providecommand \bibfnamefont [1]{#1}%
\providecommand \citenamefont [1]{#1}%
\providecommand \href@noop [0]{\@secondoftwo}%
\providecommand \href [0]{\begingroup \@sanitize@url \@href}%
\providecommand \@href[1]{\@@startlink{#1}\@@href}%
\providecommand \@@href[1]{\endgroup#1\@@endlink}%
\providecommand \@sanitize@url [0]{\catcode `\\12\catcode `\$12\catcode
  `\&12\catcode `\#12\catcode `\^12\catcode `\_12\catcode `\%12\relax}%
\providecommand \@@startlink[1]{}%
\providecommand \@@endlink[0]{}%
\providecommand \url  [0]{\begingroup\@sanitize@url \@url }%
\providecommand \@url [1]{\endgroup\@href {#1}{\urlprefix }}%
\providecommand \urlprefix  [0]{URL }%
\providecommand \Eprint [0]{\href }%
\providecommand \doibase [0]{http://dx.doi.org/}%
\providecommand \selectlanguage [0]{\@gobble}%
\providecommand \bibinfo  [0]{\@secondoftwo}%
\providecommand \bibfield  [0]{\@secondoftwo}%
\providecommand \translation [1]{[#1]}%
\providecommand \BibitemOpen [0]{}%
\providecommand \bibitemStop [0]{}%
\providecommand \bibitemNoStop [0]{.\EOS\space}%
\providecommand \EOS [0]{\spacefactor3000\relax}%
\providecommand \BibitemShut  [1]{\csname bibitem#1\endcsname}%
\let\auto@bib@innerbib\@empty
\bibitem [{\citenamefont {Horodecki}\ \emph {et~al.}(2009)\citenamefont
  {Horodecki}, \citenamefont {Horodecki}, \citenamefont {Horodecki},\ and\
  \citenamefont {Horodecki}}]{Horodecki09}%
  \BibitemOpen
  \bibfield  {author} {\bibinfo {author} {\bibfnamefont {R.}~\bibnamefont
  {Horodecki}}, \bibinfo {author} {\bibfnamefont {P.}~\bibnamefont
  {Horodecki}}, \bibinfo {author} {\bibfnamefont {M.}~\bibnamefont
  {Horodecki}}, \ and\ \bibinfo {author} {\bibfnamefont {K.}~\bibnamefont
  {Horodecki}},\ }\bibfield  {title} {\emph {\bibinfo {title} {Quantum
  entanglement},}\ }\href {http://dx.doi.org/10.1103/RevModPhys.81.865}
  {\bibfield  {journal} {\bibinfo  {journal} {Rev. Mod. Phys.}\ }\textbf
  {\bibinfo {volume} {81}},\ \bibinfo {pages} {865} (\bibinfo {year}
  {2009})}\BibitemShut {NoStop}%
\bibitem [{\citenamefont {Streltsov}\ \emph {et~al.}(2017)\citenamefont
  {Streltsov}, \citenamefont {Adesso},\ and\ \citenamefont
  {Plenio}}]{RevModPhys.89.041003}%
  \BibitemOpen
  \bibfield  {author} {\bibinfo {author} {\bibfnamefont {A.}~\bibnamefont
  {Streltsov}}, \bibinfo {author} {\bibfnamefont {G.}~\bibnamefont {Adesso}}, \
  and\ \bibinfo {author} {\bibfnamefont {M.~B.}\ \bibnamefont {Plenio}},\
  }\bibfield  {title} {\emph {\bibinfo {title} {Colloquium: Quantum coherence
  as a resource},}\ }\href {http://dx.doi.org/10.1103/RevModPhys.89.041003}
  {\bibfield  {journal} {\bibinfo  {journal} {Rev. Mod. Phys.}\ }\textbf
  {\bibinfo {volume} {89}},\ \bibinfo {pages} {041003} (\bibinfo {year}
  {2017})}\BibitemShut {NoStop}%
\bibitem [{\citenamefont {Chitambar}\ and\ \citenamefont
  {Gour}(2019)}]{chitambar2018quantum}%
  \BibitemOpen
  \bibfield  {author} {\bibinfo {author} {\bibfnamefont {E.}~\bibnamefont
  {Chitambar}}\ and\ \bibinfo {author} {\bibfnamefont {G.}~\bibnamefont
  {Gour}},\ }\bibfield  {title} {\emph {\bibinfo {title} {Quantum resource
  theories},}\ }\href {http://dx.doi.org/10.1103/RevModPhys.91.025001}
  {\bibfield  {journal} {\bibinfo  {journal} {Rev. Mod. Phys.}\ }\textbf
  {\bibinfo {volume} {91}},\ \bibinfo {pages} {025001} (\bibinfo {year}
  {2019})}\BibitemShut {NoStop}%
\bibitem [{\citenamefont {Cerezo}\ \emph {et~al.}(2021)\citenamefont {Cerezo},
  \citenamefont {Arrasmith}, \citenamefont {Babbush}, \citenamefont {Benjamin},
  \citenamefont {Endo}, \citenamefont {Fujii}, \citenamefont {McClean},
  \citenamefont {Mitarai}, \citenamefont {Yuan}, \citenamefont {Cincio},\ and\
  \citenamefont {Coles}}]{Cerezo2021variational}%
  \BibitemOpen
  \bibfield  {author} {\bibinfo {author} {\bibfnamefont {M.}~\bibnamefont
  {Cerezo}}, \bibinfo {author} {\bibfnamefont {A.}~\bibnamefont {Arrasmith}},
  \bibinfo {author} {\bibfnamefont {R.}~\bibnamefont {Babbush}}, \bibinfo
  {author} {\bibfnamefont {S.~C.}\ \bibnamefont {Benjamin}}, \bibinfo {author}
  {\bibfnamefont {S.}~\bibnamefont {Endo}}, \bibinfo {author} {\bibfnamefont
  {K.}~\bibnamefont {Fujii}}, \bibinfo {author} {\bibfnamefont {J.~R.}\
  \bibnamefont {McClean}}, \bibinfo {author} {\bibfnamefont {K.}~\bibnamefont
  {Mitarai}}, \bibinfo {author} {\bibfnamefont {X.}~\bibnamefont {Yuan}},
  \bibinfo {author} {\bibfnamefont {L.}~\bibnamefont {Cincio}}, \ and\ \bibinfo
  {author} {\bibfnamefont {P.~J.}\ \bibnamefont {Coles}},\ }\bibfield  {title}
  {\emph {\bibinfo {title} {Variational quantum algorithms},}\ }\href
  {http://dx.doi.org/10.1038/s42254-021-00348-9} {\bibfield  {journal}
  {\bibinfo  {journal} {Nat. Rev. Phys.}\ }\textbf {\bibinfo {volume} {3}},\
  \bibinfo {pages} {625} (\bibinfo {year} {2021})}\BibitemShut {NoStop}%
\bibitem [{\citenamefont {Cotler}\ \emph {et~al.}(2019)\citenamefont {Cotler},
  \citenamefont {Choi}, \citenamefont {Lukin}, \citenamefont {Gharibyan},
  \citenamefont {Grover}, \citenamefont {Tai}, \citenamefont {Rispoli},
  \citenamefont {Schittko}, \citenamefont {Preiss}, \citenamefont {Kaufman},
  \citenamefont {Greiner}, \citenamefont {Pichler},\ and\ \citenamefont
  {Hayden}}]{Cotler2019quantum}%
  \BibitemOpen
  \bibfield  {author} {\bibinfo {author} {\bibfnamefont {J.}~\bibnamefont
  {Cotler}}, \bibinfo {author} {\bibfnamefont {S.}~\bibnamefont {Choi}},
  \bibinfo {author} {\bibfnamefont {A.}~\bibnamefont {Lukin}}, \bibinfo
  {author} {\bibfnamefont {H.}~\bibnamefont {Gharibyan}}, \bibinfo {author}
  {\bibfnamefont {T.}~\bibnamefont {Grover}}, \bibinfo {author} {\bibfnamefont
  {M.~E.}\ \bibnamefont {Tai}}, \bibinfo {author} {\bibfnamefont
  {M.}~\bibnamefont {Rispoli}}, \bibinfo {author} {\bibfnamefont
  {R.}~\bibnamefont {Schittko}}, \bibinfo {author} {\bibfnamefont {P.~M.}\
  \bibnamefont {Preiss}}, \bibinfo {author} {\bibfnamefont {A.~M.}\
  \bibnamefont {Kaufman}}, \bibinfo {author} {\bibfnamefont {M.}~\bibnamefont
  {Greiner}}, \bibinfo {author} {\bibfnamefont {H.}~\bibnamefont {Pichler}}, \
  and\ \bibinfo {author} {\bibfnamefont {P.}~\bibnamefont {Hayden}},\
  }\bibfield  {title} {\emph {\bibinfo {title} {Quantum virtual cooling},}\
  }\href {http://dx.doi.org/10.1103/PhysRevX.9.031013} {\bibfield  {journal}
  {\bibinfo  {journal} {Phys. Rev. X}\ }\textbf {\bibinfo {volume} {9}},\
  \bibinfo {pages} {031013} (\bibinfo {year} {2019})}\BibitemShut {NoStop}%
\bibitem [{\citenamefont {Koczor}(2021)}]{Koczor2021exponential}%
  \BibitemOpen
  \bibfield  {author} {\bibinfo {author} {\bibfnamefont {B.}~\bibnamefont
  {Koczor}},\ }\bibfield  {title} {\emph {\bibinfo {title} {Exponential error
  suppression for near-term quantum devices},}\ }\href
  {http://dx.doi.org/10.1103/PhysRevX.11.031057} {\bibfield  {journal}
  {\bibinfo  {journal} {Phys. Rev. X}\ }\textbf {\bibinfo {volume} {11}},\
  \bibinfo {pages} {031057} (\bibinfo {year} {2021})}\BibitemShut {NoStop}%
\bibitem [{\citenamefont {Huggins}\ \emph {et~al.}(2021)\citenamefont
  {Huggins}, \citenamefont {McArdle}, \citenamefont {O'Brien}, \citenamefont
  {Lee}, \citenamefont {Rubin}, \citenamefont {Boixo}, \citenamefont {Whaley},
  \citenamefont {Babbush},\ and\ \citenamefont {McClean}}]{Huggins2021virtual}%
  \BibitemOpen
  \bibfield  {author} {\bibinfo {author} {\bibfnamefont {W.~J.}\ \bibnamefont
  {Huggins}}, \bibinfo {author} {\bibfnamefont {S.}~\bibnamefont {McArdle}},
  \bibinfo {author} {\bibfnamefont {T.~E.}\ \bibnamefont {O'Brien}}, \bibinfo
  {author} {\bibfnamefont {J.}~\bibnamefont {Lee}}, \bibinfo {author}
  {\bibfnamefont {N.~C.}\ \bibnamefont {Rubin}}, \bibinfo {author}
  {\bibfnamefont {S.}~\bibnamefont {Boixo}}, \bibinfo {author} {\bibfnamefont
  {K.~B.}\ \bibnamefont {Whaley}}, \bibinfo {author} {\bibfnamefont
  {R.}~\bibnamefont {Babbush}}, \ and\ \bibinfo {author} {\bibfnamefont
  {J.~R.}\ \bibnamefont {McClean}},\ }\bibfield  {title} {\emph {\bibinfo
  {title} {Virtual distillation for quantum error mitigation},}\ }\href
  {http://dx.doi.org/10.1103/PhysRevX.11.041036} {\bibfield  {journal}
  {\bibinfo  {journal} {Phys. Rev. X}\ }\textbf {\bibinfo {volume} {11}},\
  \bibinfo {pages} {041036} (\bibinfo {year} {2021})}\BibitemShut {NoStop}%
\bibitem [{\citenamefont {{Buscemi}}\ \emph {et~al.}(2013)\citenamefont
  {{Buscemi}}, \citenamefont {{Dall'Arno}}, \citenamefont {{Ozawa}},\ and\
  \citenamefont {{Vedral}}}]{Buscemi2013twopoint}%
  \BibitemOpen
  \bibfield  {author} {\bibinfo {author} {\bibfnamefont {F.}~\bibnamefont
  {{Buscemi}}}, \bibinfo {author} {\bibfnamefont {M.}~\bibnamefont
  {{Dall'Arno}}}, \bibinfo {author} {\bibfnamefont {M.}~\bibnamefont
  {{Ozawa}}}, \ and\ \bibinfo {author} {\bibfnamefont {V.}~\bibnamefont
  {{Vedral}}},\ }\bibfield  {title} {\emph {\bibinfo {title} {{Direct
  observation of any two-point quantum correlation function}},}\ }\href@noop {}
  {\Eprint {http://arxiv.org/abs/1312.4240} {arXiv:1312.4240}  (\bibinfo {year}
  {2013})}\BibitemShut {NoStop}%
\bibitem [{\citenamefont {Jiang}\ \emph {et~al.}(2021)\citenamefont {Jiang},
  \citenamefont {Wang},\ and\ \citenamefont {Wang}}]{Jiang2021physical}%
  \BibitemOpen
  \bibfield  {author} {\bibinfo {author} {\bibfnamefont {J.}~\bibnamefont
  {Jiang}}, \bibinfo {author} {\bibfnamefont {K.}~\bibnamefont {Wang}}, \ and\
  \bibinfo {author} {\bibfnamefont {X.}~\bibnamefont {Wang}},\ }\bibfield
  {title} {\emph {\bibinfo {title} {Physical {I}mplementability of {L}inear
  {M}aps and {I}ts {A}pplication in {E}rror {M}itigation},}\ }\href
  {http://dx.doi.org/10.22331/q-2021-12-07-600} {\bibfield  {journal} {\bibinfo
   {journal} {{Quantum}}\ }\textbf {\bibinfo {volume} {5}},\ \bibinfo {pages}
  {600} (\bibinfo {year} {2021})}\BibitemShut {NoStop}%
\bibitem [{\citenamefont {Regula}\ \emph {et~al.}(2021)\citenamefont {Regula},
  \citenamefont {Takagi},\ and\ \citenamefont {Gu}}]{Regula2021operational}%
  \BibitemOpen
  \bibfield  {author} {\bibinfo {author} {\bibfnamefont {B.}~\bibnamefont
  {Regula}}, \bibinfo {author} {\bibfnamefont {R.}~\bibnamefont {Takagi}}, \
  and\ \bibinfo {author} {\bibfnamefont {M.}~\bibnamefont {Gu}},\ }\bibfield
  {title} {\emph {\bibinfo {title} {Operational applications of the diamond
  norm and related measures in quantifying the non-physicality of quantum
  maps},}\ }\href {http://dx.doi.org/10.22331/q-2021-08-09-522} {\bibfield
  {journal} {\bibinfo  {journal} {{Quantum}}\ }\textbf {\bibinfo {volume}
  {5}},\ \bibinfo {pages} {522} (\bibinfo {year} {2021})}\BibitemShut {NoStop}%
\bibitem [{\citenamefont {Temme}\ \emph {et~al.}(2017)\citenamefont {Temme},
  \citenamefont {Bravyi},\ and\ \citenamefont
  {Gambetta}}]{PhysRevLett.119.180509}%
  \BibitemOpen
  \bibfield  {author} {\bibinfo {author} {\bibfnamefont {K.}~\bibnamefont
  {Temme}}, \bibinfo {author} {\bibfnamefont {S.}~\bibnamefont {Bravyi}}, \
  and\ \bibinfo {author} {\bibfnamefont {J.~M.}\ \bibnamefont {Gambetta}},\
  }\bibfield  {title} {\emph {\bibinfo {title} {Error mitigation for
  short-depth quantum circuits},}\ }\href
  {http://dx.doi.org/10.1103/PhysRevLett.119.180509} {\bibfield  {journal}
  {\bibinfo  {journal} {Phys. Rev. Lett.}\ }\textbf {\bibinfo {volume} {119}},\
  \bibinfo {pages} {180509} (\bibinfo {year} {2017})}\BibitemShut {NoStop}%
\bibitem [{\citenamefont {Yuan}\ \emph {et~al.}(2024)\citenamefont {Yuan},
  \citenamefont {Regula}, \citenamefont {Takagi},\ and\ \citenamefont
  {Gu}}]{maintext}%
  \BibitemOpen
  \bibfield  {author} {\bibinfo {author} {\bibfnamefont {X.}~\bibnamefont
  {Yuan}}, \bibinfo {author} {\bibfnamefont {B.}~\bibnamefont {Regula}},
  \bibinfo {author} {\bibfnamefont {R.}~\bibnamefont {Takagi}}, \ and\ \bibinfo
  {author} {\bibfnamefont {M.}~\bibnamefont {Gu}},\ }\bibfield  {title} {\emph
  {\bibinfo {title} {Virtual quantum resource distillation},}\ }\href
  {http://dx.doi.org/10.1103/PhysRevLett.132.050203} {\bibfield  {journal}
  {\bibinfo  {journal} {Phys. Rev. Lett.}\ }\textbf {\bibinfo {volume} {132}},\
  \bibinfo {pages} {050203} (\bibinfo {year} {2024})},\ \bibinfo {note}
  {companion paper}\BibitemShut {NoStop}%
\bibitem [{\citenamefont {Ben~Dana}\ \emph {et~al.}(2017)\citenamefont
  {Ben~Dana}, \citenamefont {Garc\'{\i}a~D\'{\i}az}, \citenamefont {Mejatty},\
  and\ \citenamefont {Winter}}]{dana2018resource}%
  \BibitemOpen
  \bibfield  {author} {\bibinfo {author} {\bibfnamefont {K.}~\bibnamefont
  {Ben~Dana}}, \bibinfo {author} {\bibfnamefont {M.}~\bibnamefont
  {Garc\'{\i}a~D\'{\i}az}}, \bibinfo {author} {\bibfnamefont {M.}~\bibnamefont
  {Mejatty}}, \ and\ \bibinfo {author} {\bibfnamefont {A.}~\bibnamefont
  {Winter}},\ }\bibfield  {title} {\emph {\bibinfo {title} {Resource theory of
  coherence: Beyond states},}\ }\href
  {http://dx.doi.org/10.1103/PhysRevA.95.062327} {\bibfield  {journal}
  {\bibinfo  {journal} {Phys. Rev. A}\ }\textbf {\bibinfo {volume} {95}},\
  \bibinfo {pages} {062327} (\bibinfo {year} {2017})}\BibitemShut {NoStop}%
\bibitem [{\citenamefont {D{\'\i}az}\ \emph {et~al.}(2018)\citenamefont
  {D{\'\i}az}, \citenamefont {Fang}, \citenamefont {Wang}, \citenamefont
  {Rosati}, \citenamefont {Skotiniotis}, \citenamefont {Calsamiglia},\ and\
  \citenamefont {Winter}}]{diaz2018using}%
  \BibitemOpen
  \bibfield  {author} {\bibinfo {author} {\bibfnamefont {M.~G.}\ \bibnamefont
  {D{\'\i}az}}, \bibinfo {author} {\bibfnamefont {K.}~\bibnamefont {Fang}},
  \bibinfo {author} {\bibfnamefont {X.}~\bibnamefont {Wang}}, \bibinfo {author}
  {\bibfnamefont {M.}~\bibnamefont {Rosati}}, \bibinfo {author} {\bibfnamefont
  {M.}~\bibnamefont {Skotiniotis}}, \bibinfo {author} {\bibfnamefont
  {J.}~\bibnamefont {Calsamiglia}}, \ and\ \bibinfo {author} {\bibfnamefont
  {A.}~\bibnamefont {Winter}},\ }\bibfield  {title} {\emph {\bibinfo {title}
  {Using and reusing coherence to realize quantum processes},}\ }\href
  {http://dx.doi.org/10.22331/q-2018-10-19-100} {\bibfield  {journal} {\bibinfo
   {journal} {Quantum}\ }\textbf {\bibinfo {volume} {2}},\ \bibinfo {pages}
  {100} (\bibinfo {year} {2018})}\BibitemShut {NoStop}%
\bibitem [{\citenamefont {Rosset}\ \emph {et~al.}(2018)\citenamefont {Rosset},
  \citenamefont {Buscemi},\ and\ \citenamefont {Liang}}]{rosset2018resource}%
  \BibitemOpen
  \bibfield  {author} {\bibinfo {author} {\bibfnamefont {D.}~\bibnamefont
  {Rosset}}, \bibinfo {author} {\bibfnamefont {F.}~\bibnamefont {Buscemi}}, \
  and\ \bibinfo {author} {\bibfnamefont {Y.-C.}\ \bibnamefont {Liang}},\
  }\bibfield  {title} {\emph {\bibinfo {title} {Resource theory of quantum
  memories and their faithful verification with minimal assumptions},}\ }\href
  {http://dx.doi.org/10.1103/PhysRevX.8.021033} {\bibfield  {journal} {\bibinfo
   {journal} {Phys. Rev. X}\ }\textbf {\bibinfo {volume} {8}},\ \bibinfo
  {pages} {021033} (\bibinfo {year} {2018})}\BibitemShut {NoStop}%
\bibitem [{\citenamefont {B{a}uml}\ \emph {et~al.}(2019)\citenamefont
  {B{a}uml}, \citenamefont {Das}, \citenamefont {Wang},\ and\ \citenamefont
  {Wilde}}]{bauml2019resource}%
  \BibitemOpen
  \bibfield  {author} {\bibinfo {author} {\bibfnamefont {S.}~\bibnamefont
  {B{a}uml}}, \bibinfo {author} {\bibfnamefont {S.}~\bibnamefont {Das}},
  \bibinfo {author} {\bibfnamefont {X.}~\bibnamefont {Wang}}, \ and\ \bibinfo
  {author} {\bibfnamefont {M.~M.}\ \bibnamefont {Wilde}},\ }\href@noop {}
  {\emph {\bibinfo {title} {Resource theory of entanglement for bipartite
  quantum channels},}\ } (\bibinfo {year} {2019}),\ \Eprint
  {http://arxiv.org/abs/1907.04181} {arXiv:1907.04181} \BibitemShut {NoStop}%
\bibitem [{\citenamefont {Gour}\ and\ \citenamefont
  {Scandolo}(2020)}]{dynamicalEntanglement}%
  \BibitemOpen
  \bibfield  {author} {\bibinfo {author} {\bibfnamefont {G.}~\bibnamefont
  {Gour}}\ and\ \bibinfo {author} {\bibfnamefont {C.~M.}\ \bibnamefont
  {Scandolo}},\ }\bibfield  {title} {\emph {\bibinfo {title} {Dynamical
  entanglement},}\ }\href {http://dx.doi.org/10.1103/PhysRevLett.125.180505}
  {\bibfield  {journal} {\bibinfo  {journal} {Phys. Rev. Lett.}\ }\textbf
  {\bibinfo {volume} {125}},\ \bibinfo {pages} {180505} (\bibinfo {year}
  {2020})}\BibitemShut {NoStop}%
\bibitem [{\citenamefont {Saxena}\ \emph {et~al.}(2020)\citenamefont {Saxena},
  \citenamefont {Chitambar},\ and\ \citenamefont {Gour}}]{dynamicalCoherence}%
  \BibitemOpen
  \bibfield  {author} {\bibinfo {author} {\bibfnamefont {G.}~\bibnamefont
  {Saxena}}, \bibinfo {author} {\bibfnamefont {E.}~\bibnamefont {Chitambar}}, \
  and\ \bibinfo {author} {\bibfnamefont {G.}~\bibnamefont {Gour}},\ }\bibfield
  {title} {\emph {\bibinfo {title} {Dynamical resource theory of quantum
  coherence},}\ }\href {http://dx.doi.org/10.1103/PhysRevResearch.2.023298}
  {\bibfield  {journal} {\bibinfo  {journal} {Phys. Rev. Research}\ }\textbf
  {\bibinfo {volume} {2}},\ \bibinfo {pages} {023298} (\bibinfo {year}
  {2020})}\BibitemShut {NoStop}%
\bibitem [{\citenamefont {Pirandola}\ \emph {et~al.}(2017)\citenamefont
  {Pirandola}, \citenamefont {Laurenza}, \citenamefont {Ottaviani},\ and\
  \citenamefont {Banchi}}]{pirandola2017fundamental}%
  \BibitemOpen
  \bibfield  {author} {\bibinfo {author} {\bibfnamefont {S.}~\bibnamefont
  {Pirandola}}, \bibinfo {author} {\bibfnamefont {R.}~\bibnamefont {Laurenza}},
  \bibinfo {author} {\bibfnamefont {C.}~\bibnamefont {Ottaviani}}, \ and\
  \bibinfo {author} {\bibfnamefont {L.}~\bibnamefont {Banchi}},\ }\bibfield
  {title} {\emph {\bibinfo {title} {Fundamental limits of repeaterless quantum
  communications},}\ }\href {http://dx.doi.org/10.1038/ncomms15043} {\bibfield
  {journal} {\bibinfo  {journal} {Nat. Commun.}\ }\textbf {\bibinfo {volume}
  {8}},\ \bibinfo {pages} {15043} (\bibinfo {year} {2017})}\BibitemShut
  {NoStop}%
\bibitem [{\citenamefont {Faist}\ and\ \citenamefont
  {Renner}(2018)}]{faist_2018}%
  \BibitemOpen
  \bibfield  {author} {\bibinfo {author} {\bibfnamefont {P.}~\bibnamefont
  {Faist}}\ and\ \bibinfo {author} {\bibfnamefont {R.}~\bibnamefont {Renner}},\
  }\bibfield  {title} {\emph {\bibinfo {title} {Fundamental {{Work Cost}} of
  {{Quantum Processes}}},}\ }\href
  {http://dx.doi.org/10.1103/PhysRevX.8.021011} {\bibfield  {journal} {\bibinfo
   {journal} {Phys. Rev. X}\ }\textbf {\bibinfo {volume} {8}},\ \bibinfo
  {pages} {021011} (\bibinfo {year} {2018})}\BibitemShut {NoStop}%
\bibitem [{\citenamefont {Theurer}\ \emph {et~al.}(2019)\citenamefont
  {Theurer}, \citenamefont {Egloff}, \citenamefont {Zhang},\ and\ \citenamefont
  {Plenio}}]{theurer2018quantifying}%
  \BibitemOpen
  \bibfield  {author} {\bibinfo {author} {\bibfnamefont {T.}~\bibnamefont
  {Theurer}}, \bibinfo {author} {\bibfnamefont {D.}~\bibnamefont {Egloff}},
  \bibinfo {author} {\bibfnamefont {L.}~\bibnamefont {Zhang}}, \ and\ \bibinfo
  {author} {\bibfnamefont {M.~B.}\ \bibnamefont {Plenio}},\ }\bibfield  {title}
  {\emph {\bibinfo {title} {Quantifying operations with an application to
  coherence},}\ }\href {http://dx.doi.org/10.1103/PhysRevLett.122.190405}
  {\bibfield  {journal} {\bibinfo  {journal} {Phys. Rev. Lett.}\ }\textbf
  {\bibinfo {volume} {122}},\ \bibinfo {pages} {190405} (\bibinfo {year}
  {2019})}\BibitemShut {NoStop}%
\bibitem [{\citenamefont {Takagi}\ and\ \citenamefont
  {Regula}(2019)}]{takagi2019general}%
  \BibitemOpen
  \bibfield  {author} {\bibinfo {author} {\bibfnamefont {R.}~\bibnamefont
  {Takagi}}\ and\ \bibinfo {author} {\bibfnamefont {B.}~\bibnamefont
  {Regula}},\ }\bibfield  {title} {\emph {\bibinfo {title} {General {{Resource
  Theories}} in {{Quantum Mechanics}} and {{Beyond}}: {{Operational
  Characterization}} via {{Discrimination Tasks}}},}\ }\href
  {http://dx.doi.org/10.1103/PhysRevX.9.031053} {\bibfield  {journal} {\bibinfo
   {journal} {Phys. Rev. X}\ }\textbf {\bibinfo {volume} {9}},\ \bibinfo
  {pages} {031053} (\bibinfo {year} {2019})}\BibitemShut {NoStop}%
\bibitem [{\citenamefont {Liu}\ and\ \citenamefont
  {Yuan}(2020)}]{liu2019operational}%
  \BibitemOpen
  \bibfield  {author} {\bibinfo {author} {\bibfnamefont {Y.}~\bibnamefont
  {Liu}}\ and\ \bibinfo {author} {\bibfnamefont {X.}~\bibnamefont {Yuan}},\
  }\bibfield  {title} {\emph {\bibinfo {title} {Operational resource theory of
  quantum channels},}\ }\href
  {http://dx.doi.org/10.1103/PhysRevResearch.2.012035} {\bibfield  {journal}
  {\bibinfo  {journal} {Phys. Rev. Research}\ }\textbf {\bibinfo {volume}
  {2}},\ \bibinfo {pages} {012035} (\bibinfo {year} {2020})}\BibitemShut
  {NoStop}%
\bibitem [{\citenamefont {Liu}\ and\ \citenamefont
  {Winter}(2019)}]{liu2019resource}%
  \BibitemOpen
  \bibfield  {author} {\bibinfo {author} {\bibfnamefont {Z.-W.}\ \bibnamefont
  {Liu}}\ and\ \bibinfo {author} {\bibfnamefont {A.}~\bibnamefont {Winter}},\
  }\bibfield  {title} {\emph {\bibinfo {title} {Resource theories of quantum
  channels and the universal role of resource erasure},}\ }\href
  {https://arxiv.org/abs/1904.04201} {\bibfield  {journal} {\bibinfo  {journal}
  {arXiv preprint arXiv:1904.04201}} (\bibinfo {year} {2019})}\BibitemShut
  {NoStop}%
\bibitem [{\citenamefont {Wang}\ \emph {et~al.}(2019)\citenamefont {Wang},
  \citenamefont {Wilde},\ and\ \citenamefont {Su}}]{wang2019quantifying}%
  \BibitemOpen
  \bibfield  {author} {\bibinfo {author} {\bibfnamefont {X.}~\bibnamefont
  {Wang}}, \bibinfo {author} {\bibfnamefont {M.~M.}\ \bibnamefont {Wilde}}, \
  and\ \bibinfo {author} {\bibfnamefont {Y.}~\bibnamefont {Su}},\ }\bibfield
  {title} {\emph {\bibinfo {title} {Quantifying the magic of quantum
  channels},}\ }\href {http://dx.doi.org/10.1088/1367-2630/ab451d} {\bibfield
  {journal} {\bibinfo  {journal} {New J. Phys.}\ }\textbf {\bibinfo {volume}
  {21}},\ \bibinfo {pages} {103002} (\bibinfo {year} {2019})}\BibitemShut
  {NoStop}%
\bibitem [{\citenamefont {Wang}\ and\ \citenamefont
  {Wilde}(2019)}]{PhysRevResearch.1.033169}%
  \BibitemOpen
  \bibfield  {author} {\bibinfo {author} {\bibfnamefont {X.}~\bibnamefont
  {Wang}}\ and\ \bibinfo {author} {\bibfnamefont {M.~M.}\ \bibnamefont
  {Wilde}},\ }\bibfield  {title} {\emph {\bibinfo {title} {Resource theory of
  asymmetric distinguishability for quantum channels},}\ }\href
  {http://dx.doi.org/10.1103/PhysRevResearch.1.033169} {\bibfield  {journal}
  {\bibinfo  {journal} {Phys. Rev. Research}\ }\textbf {\bibinfo {volume}
  {1}},\ \bibinfo {pages} {033169} (\bibinfo {year} {2019})}\BibitemShut
  {NoStop}%
\bibitem [{\citenamefont {Takagi}\ \emph {et~al.}(2020)\citenamefont {Takagi},
  \citenamefont {Wang},\ and\ \citenamefont {Hayashi}}]{takagi_2020}%
  \BibitemOpen
  \bibfield  {author} {\bibinfo {author} {\bibfnamefont {R.}~\bibnamefont
  {Takagi}}, \bibinfo {author} {\bibfnamefont {K.}~\bibnamefont {Wang}}, \ and\
  \bibinfo {author} {\bibfnamefont {M.}~\bibnamefont {Hayashi}},\ }\bibfield
  {title} {\emph {\bibinfo {title} {Application of the {{Resource Theory}} of
  {{Channels}} to {{Communication Scenarios}}},}\ }\href
  {http://dx.doi.org/10.1103/PhysRevLett.124.120502} {\bibfield  {journal}
  {\bibinfo  {journal} {Phys. Rev. Lett.}\ }\textbf {\bibinfo {volume} {124}},\
  \bibinfo {pages} {120502} (\bibinfo {year} {2020})}\BibitemShut {NoStop}%
\bibitem [{\citenamefont {Kristjánsson}\ \emph {et~al.}(2020)\citenamefont
  {Kristjánsson}, \citenamefont {Chiribella}, \citenamefont {Salek},
  \citenamefont {Ebler},\ and\ \citenamefont
  {Wilson}}]{Kristjansson2020resource}%
  \BibitemOpen
  \bibfield  {author} {\bibinfo {author} {\bibfnamefont {H.}~\bibnamefont
  {Kristjánsson}}, \bibinfo {author} {\bibfnamefont {G.}~\bibnamefont
  {Chiribella}}, \bibinfo {author} {\bibfnamefont {S.}~\bibnamefont {Salek}},
  \bibinfo {author} {\bibfnamefont {D.}~\bibnamefont {Ebler}}, \ and\ \bibinfo
  {author} {\bibfnamefont {M.}~\bibnamefont {Wilson}},\ }\bibfield  {title}
  {\emph {\bibinfo {title} {Resource theories of communication},}\ }\href
  {http://dx.doi.org/10.1088/1367-2630/ab8ef7} {\bibfield  {journal} {\bibinfo
  {journal} {New J. Phys.}\ }\textbf {\bibinfo {volume} {22}},\ \bibinfo
  {pages} {073014} (\bibinfo {year} {2020})}\BibitemShut {NoStop}%
\bibitem [{\citenamefont {Gour}\ and\ \citenamefont
  {Winter}(2019)}]{Gour2019how}%
  \BibitemOpen
  \bibfield  {author} {\bibinfo {author} {\bibfnamefont {G.}~\bibnamefont
  {Gour}}\ and\ \bibinfo {author} {\bibfnamefont {A.}~\bibnamefont {Winter}},\
  }\bibfield  {title} {\emph {\bibinfo {title} {How to quantify a dynamical
  quantum resource},}\ }\href
  {http://dx.doi.org/10.1103/PhysRevLett.123.150401} {\bibfield  {journal}
  {\bibinfo  {journal} {Phys. Rev. Lett.}\ }\textbf {\bibinfo {volume} {123}},\
  \bibinfo {pages} {150401} (\bibinfo {year} {2019})}\BibitemShut {NoStop}%
\bibitem [{\citenamefont {Yuan}\ \emph {et~al.}(2021)\citenamefont {Yuan},
  \citenamefont {Liu}, \citenamefont {Zhao}, \citenamefont {Regula},
  \citenamefont {Thompson},\ and\ \citenamefont {Gu}}]{yuan_2020}%
  \BibitemOpen
  \bibfield  {author} {\bibinfo {author} {\bibfnamefont {X.}~\bibnamefont
  {Yuan}}, \bibinfo {author} {\bibfnamefont {Y.}~\bibnamefont {Liu}}, \bibinfo
  {author} {\bibfnamefont {Q.}~\bibnamefont {Zhao}}, \bibinfo {author}
  {\bibfnamefont {B.}~\bibnamefont {Regula}}, \bibinfo {author} {\bibfnamefont
  {J.}~\bibnamefont {Thompson}}, \ and\ \bibinfo {author} {\bibfnamefont
  {M.}~\bibnamefont {Gu}},\ }\bibfield  {title} {\emph {\bibinfo {title}
  {Universal and operational benchmarking of quantum memories},}\ }\href
  {http://dx.doi.org/10.1038/s41534-021-00444-9} {\bibfield  {journal}
  {\bibinfo  {journal} {{npj} Quantum Inf.}\ }\textbf {\bibinfo {volume} {7}},\
  \bibinfo {pages} {108} (\bibinfo {year} {2021})}\BibitemShut {NoStop}%
\bibitem [{\citenamefont {Regula}\ and\ \citenamefont
  {Takagi}(2021{\natexlab{a}})}]{Regula2021fundamental}%
  \BibitemOpen
  \bibfield  {author} {\bibinfo {author} {\bibfnamefont {B.}~\bibnamefont
  {Regula}}\ and\ \bibinfo {author} {\bibfnamefont {R.}~\bibnamefont
  {Takagi}},\ }\bibfield  {title} {\emph {\bibinfo {title} {Fundamental
  limitations on distillation of quantum channel resources},}\ }\href
  {http://dx.doi.org/https://doi.org/10.1038/s41467-021-24699-0} {\bibfield
  {journal} {\bibinfo  {journal} {Nat. Commun.}\ }\textbf {\bibinfo {volume}
  {12}},\ \bibinfo {pages} {4411} (\bibinfo {year}
  {2021}{\natexlab{a}})}\BibitemShut {NoStop}%
\bibitem [{\citenamefont {Fang}\ and\ \citenamefont
  {Liu}(2022)}]{Fang2020no-go}%
  \BibitemOpen
  \bibfield  {author} {\bibinfo {author} {\bibfnamefont {K.}~\bibnamefont
  {Fang}}\ and\ \bibinfo {author} {\bibfnamefont {Z.-W.}\ \bibnamefont {Liu}},\
  }\bibfield  {title} {\emph {\bibinfo {title} {No-go theorems for quantum
  resource purification: New approach and channel theory},}\ }\href
  {http://dx.doi.org/10.1103/PRXQuantum.3.010337} {\bibfield  {journal}
  {\bibinfo  {journal} {PRX Quantum}\ }\textbf {\bibinfo {volume} {3}},\
  \bibinfo {pages} {010337} (\bibinfo {year} {2022})}\BibitemShut {NoStop}%
\bibitem [{\citenamefont {Takagi}(2021)}]{Takagi2021optimal}%
  \BibitemOpen
  \bibfield  {author} {\bibinfo {author} {\bibfnamefont {R.}~\bibnamefont
  {Takagi}},\ }\bibfield  {title} {\emph {\bibinfo {title} {Optimal resource
  cost for error mitigation},}\ }\href
  {http://dx.doi.org/10.1103/PhysRevResearch.3.033178} {\bibfield  {journal}
  {\bibinfo  {journal} {Phys. Rev. Res.}\ }\textbf {\bibinfo {volume} {3}},\
  \bibinfo {pages} {033178} (\bibinfo {year} {2021})}\BibitemShut {NoStop}%
\bibitem [{\citenamefont {Berk}\ \emph {et~al.}(2021)\citenamefont {Berk},
  \citenamefont {Garner}, \citenamefont {Yadin}, \citenamefont {Modi},\ and\
  \citenamefont {Pollock}}]{Berk2021resourcetheoriesof}%
  \BibitemOpen
  \bibfield  {author} {\bibinfo {author} {\bibfnamefont {G.~D.}\ \bibnamefont
  {Berk}}, \bibinfo {author} {\bibfnamefont {A.~J.~P.}\ \bibnamefont {Garner}},
  \bibinfo {author} {\bibfnamefont {B.}~\bibnamefont {Yadin}}, \bibinfo
  {author} {\bibfnamefont {K.}~\bibnamefont {Modi}}, \ and\ \bibinfo {author}
  {\bibfnamefont {F.~A.}\ \bibnamefont {Pollock}},\ }\bibfield  {title} {\emph
  {\bibinfo {title} {Resource theories of multi-time processes: {A} window into
  quantum non-{M}arkovianity},}\ }\href
  {http://dx.doi.org/10.22331/q-2021-04-20-435} {\bibfield  {journal} {\bibinfo
   {journal} {{Quantum}}\ }\textbf {\bibinfo {volume} {5}},\ \bibinfo {pages}
  {435} (\bibinfo {year} {2021})}\BibitemShut {NoStop}%
\bibitem [{\citenamefont {Horodecki}\ \emph {et~al.}(1999)\citenamefont
  {Horodecki}, \citenamefont {Horodecki},\ and\ \citenamefont
  {Horodecki}}]{horodecki1999general}%
  \BibitemOpen
  \bibfield  {author} {\bibinfo {author} {\bibfnamefont {M.}~\bibnamefont
  {Horodecki}}, \bibinfo {author} {\bibfnamefont {P.}~\bibnamefont
  {Horodecki}}, \ and\ \bibinfo {author} {\bibfnamefont {R.}~\bibnamefont
  {Horodecki}},\ }\bibfield  {title} {\emph {\bibinfo {title} {General
  teleportation channel, singlet fraction, and quasidistillation},}\ }\href
  {http://dx.doi.org/10.1103/PhysRevA.60.1888} {\bibfield  {journal} {\bibinfo
  {journal} {Phys. Rev. A}\ }\textbf {\bibinfo {volume} {60}},\ \bibinfo
  {pages} {1888} (\bibinfo {year} {1999})}\BibitemShut {NoStop}%
\bibitem [{\citenamefont {Jonathan}\ and\ \citenamefont
  {Plenio}(1999)}]{Jonathan1999minimal}%
  \BibitemOpen
  \bibfield  {author} {\bibinfo {author} {\bibfnamefont {D.}~\bibnamefont
  {Jonathan}}\ and\ \bibinfo {author} {\bibfnamefont {M.~B.}\ \bibnamefont
  {Plenio}},\ }\bibfield  {title} {\emph {\bibinfo {title} {Minimal conditions
  for local pure-state entanglement manipulation},}\ }\href
  {http://dx.doi.org/10.1103/PhysRevLett.83.1455} {\bibfield  {journal}
  {\bibinfo  {journal} {Phys. Rev. Lett.}\ }\textbf {\bibinfo {volume} {83}},\
  \bibinfo {pages} {1455} (\bibinfo {year} {1999})}\BibitemShut {NoStop}%
\bibitem [{\citenamefont {Fang}\ \emph {et~al.}(2018)\citenamefont {Fang},
  \citenamefont {Wang}, \citenamefont {Lami}, \citenamefont {Regula},\ and\
  \citenamefont {Adesso}}]{Fang2018probabilistic}%
  \BibitemOpen
  \bibfield  {author} {\bibinfo {author} {\bibfnamefont {K.}~\bibnamefont
  {Fang}}, \bibinfo {author} {\bibfnamefont {X.}~\bibnamefont {Wang}}, \bibinfo
  {author} {\bibfnamefont {L.}~\bibnamefont {Lami}}, \bibinfo {author}
  {\bibfnamefont {B.}~\bibnamefont {Regula}}, \ and\ \bibinfo {author}
  {\bibfnamefont {G.}~\bibnamefont {Adesso}},\ }\bibfield  {title} {\emph
  {\bibinfo {title} {Probabilistic distillation of quantum coherence},}\ }\href
  {http://dx.doi.org/10.1103/PhysRevLett.121.070404} {\bibfield  {journal}
  {\bibinfo  {journal} {Phys. Rev. Lett.}\ }\textbf {\bibinfo {volume} {121}},\
  \bibinfo {pages} {070404} (\bibinfo {year} {2018})}\BibitemShut {NoStop}%
\bibitem [{\citenamefont {Fang}\ and\ \citenamefont
  {Liu}(2020)}]{PhysRevLett.125.060405}%
  \BibitemOpen
  \bibfield  {author} {\bibinfo {author} {\bibfnamefont {K.}~\bibnamefont
  {Fang}}\ and\ \bibinfo {author} {\bibfnamefont {Z.-W.}\ \bibnamefont {Liu}},\
  }\bibfield  {title} {\emph {\bibinfo {title} {No-go theorems for quantum
  resource purification},}\ }\href
  {http://dx.doi.org/10.1103/PhysRevLett.125.060405} {\bibfield  {journal}
  {\bibinfo  {journal} {Phys. Rev. Lett.}\ }\textbf {\bibinfo {volume} {125}},\
  \bibinfo {pages} {060405} (\bibinfo {year} {2020})}\BibitemShut {NoStop}%
\bibitem [{\citenamefont {Regula}(2022{\natexlab{a}})}]{regula_2022}%
  \BibitemOpen
  \bibfield  {author} {\bibinfo {author} {\bibfnamefont {B.}~\bibnamefont
  {Regula}},\ }\bibfield  {title} {\emph {\bibinfo {title} {Probabilistic
  {{Transformations}} of {{Quantum Resources}}},}\ }\href
  {http://dx.doi.org/10.1103/PhysRevLett.128.110505} {\bibfield  {journal}
  {\bibinfo  {journal} {Phys. Rev. Lett.}\ }\textbf {\bibinfo {volume} {128}},\
  \bibinfo {pages} {110505} (\bibinfo {year} {2022}{\natexlab{a}})}\BibitemShut
  {NoStop}%
\bibitem [{\citenamefont {Regula}(2022{\natexlab{b}})}]{regula_2021-4}%
  \BibitemOpen
  \bibfield  {author} {\bibinfo {author} {\bibfnamefont {B.}~\bibnamefont
  {Regula}},\ }\bibfield  {title} {\emph {\bibinfo {title} {Tight constraints
  on probabilistic convertibility of quantum states},}\ }\href
  {http://dx.doi.org/10.22331/q-2022-09-22-817} {\bibfield  {journal} {\bibinfo
   {journal} {{Quantum}}\ }\textbf {\bibinfo {volume} {6}},\ \bibinfo {pages}
  {817} (\bibinfo {year} {2022}{\natexlab{b}})}\BibitemShut {NoStop}%
\bibitem [{\citenamefont {Chiribella}\ \emph
  {et~al.}(2008{\natexlab{a}})\citenamefont {Chiribella}, \citenamefont
  {D'Ariano},\ and\ \citenamefont {Perinotti}}]{Chiribella2008quantum}%
  \BibitemOpen
  \bibfield  {author} {\bibinfo {author} {\bibfnamefont {G.}~\bibnamefont
  {Chiribella}}, \bibinfo {author} {\bibfnamefont {G.~M.}\ \bibnamefont
  {D'Ariano}}, \ and\ \bibinfo {author} {\bibfnamefont {P.}~\bibnamefont
  {Perinotti}},\ }\bibfield  {title} {\emph {\bibinfo {title} {Quantum circuit
  architecture},}\ }\href {http://dx.doi.org/10.1103/PhysRevLett.101.060401}
  {\bibfield  {journal} {\bibinfo  {journal} {Phys. Rev. Lett.}\ }\textbf
  {\bibinfo {volume} {101}},\ \bibinfo {pages} {060401} (\bibinfo {year}
  {2008}{\natexlab{a}})}\BibitemShut {NoStop}%
\bibitem [{\citenamefont {Chiribella}\ \emph
  {et~al.}(2008{\natexlab{b}})\citenamefont {Chiribella}, \citenamefont
  {D'Ariano},\ and\ \citenamefont {Perinotti}}]{Chiribella2008transforming}%
  \BibitemOpen
  \bibfield  {author} {\bibinfo {author} {\bibfnamefont {G.}~\bibnamefont
  {Chiribella}}, \bibinfo {author} {\bibfnamefont {G.~M.}\ \bibnamefont
  {D'Ariano}}, \ and\ \bibinfo {author} {\bibfnamefont {P.}~\bibnamefont
  {Perinotti}},\ }\bibfield  {title} {\emph {\bibinfo {title} {Transforming
  quantum operations: Quantum supermaps},}\ }\href
  {http://dx.doi.org/10.1209/0295-5075/83/30004} {\bibfield  {journal}
  {\bibinfo  {journal} {EPL}\ }\textbf {\bibinfo {volume} {83}},\ \bibinfo
  {pages} {30004} (\bibinfo {year} {2008}{\natexlab{b}})}\BibitemShut {NoStop}%
\bibitem [{\citenamefont {{Berk}}\ \emph {et~al.}(2021)\citenamefont {{Berk}},
  \citenamefont {{Milz}}, \citenamefont {{Pollock}},\ and\ \citenamefont
  {{Modi}}}]{Berk2021extracting}%
  \BibitemOpen
  \bibfield  {author} {\bibinfo {author} {\bibfnamefont {G.~D.}\ \bibnamefont
  {{Berk}}}, \bibinfo {author} {\bibfnamefont {S.}~\bibnamefont {{Milz}}},
  \bibinfo {author} {\bibfnamefont {F.~A.}\ \bibnamefont {{Pollock}}}, \ and\
  \bibinfo {author} {\bibfnamefont {K.}~\bibnamefont {{Modi}}},\ }\bibfield
  {title} {\emph {\bibinfo {title} {{Extracting Quantum Dynamical Resources:
  Consumption of Non-Markovianity for Noise Reduction}},}\ }\href@noop {}
  {\Eprint {http://arxiv.org/abs/2110.02613} {arXiv:2110.02613}  (\bibinfo
  {year} {2021})}\BibitemShut {NoStop}%
\bibitem [{\citenamefont {Horodecki}\ and\ \citenamefont
  {Oppenheim}(2013)}]{horodecki2013quantumness}%
  \BibitemOpen
  \bibfield  {author} {\bibinfo {author} {\bibfnamefont {M.}~\bibnamefont
  {Horodecki}}\ and\ \bibinfo {author} {\bibfnamefont {J.}~\bibnamefont
  {Oppenheim}},\ }\bibfield  {title} {\emph {\bibinfo {title} {(quantumness in
  the context of) resource theories},}\ }\href@noop {} {\bibfield  {journal}
  {\bibinfo  {journal} {Int. J. Mod. Phys. B}\ }\textbf {\bibinfo {volume}
  {27}},\ \bibinfo {pages} {1345019} (\bibinfo {year} {2013})}\BibitemShut
  {NoStop}%
\bibitem [{\citenamefont {Takagi}\ \emph {et~al.}(2019)\citenamefont {Takagi},
  \citenamefont {Regula}, \citenamefont {Bu}, \citenamefont {Liu},\ and\
  \citenamefont {Adesso}}]{takagi2018operational}%
  \BibitemOpen
  \bibfield  {author} {\bibinfo {author} {\bibfnamefont {R.}~\bibnamefont
  {Takagi}}, \bibinfo {author} {\bibfnamefont {B.}~\bibnamefont {Regula}},
  \bibinfo {author} {\bibfnamefont {K.}~\bibnamefont {Bu}}, \bibinfo {author}
  {\bibfnamefont {Z.-W.}\ \bibnamefont {Liu}}, \ and\ \bibinfo {author}
  {\bibfnamefont {G.}~\bibnamefont {Adesso}},\ }\bibfield  {title} {\emph
  {\bibinfo {title} {Operational advantage of quantum resources in subchannel
  discrimination},}\ }\href {http://dx.doi.org/10.1103/PhysRevLett.122.140402}
  {\bibfield  {journal} {\bibinfo  {journal} {Phys. Rev. Lett.}\ }\textbf
  {\bibinfo {volume} {122}},\ \bibinfo {pages} {140402} (\bibinfo {year}
  {2019})}\BibitemShut {NoStop}%
\bibitem [{\citenamefont {Uola}\ \emph {et~al.}(2019)\citenamefont {Uola},
  \citenamefont {Kraft}, \citenamefont {Shang}, \citenamefont {Yu},\ and\
  \citenamefont {G\"uhne}}]{Uola2019quantifying_conic}%
  \BibitemOpen
  \bibfield  {author} {\bibinfo {author} {\bibfnamefont {R.}~\bibnamefont
  {Uola}}, \bibinfo {author} {\bibfnamefont {T.}~\bibnamefont {Kraft}},
  \bibinfo {author} {\bibfnamefont {J.}~\bibnamefont {Shang}}, \bibinfo
  {author} {\bibfnamefont {X.-D.}\ \bibnamefont {Yu}}, \ and\ \bibinfo {author}
  {\bibfnamefont {O.}~\bibnamefont {G\"uhne}},\ }\bibfield  {title} {\emph
  {\bibinfo {title} {Quantifying quantum resources with conic programming},}\
  }\href {http://dx.doi.org/10.1103/PhysRevLett.122.130404} {\bibfield
  {journal} {\bibinfo  {journal} {Phys. Rev. Lett.}\ }\textbf {\bibinfo
  {volume} {122}},\ \bibinfo {pages} {130404} (\bibinfo {year}
  {2019})}\BibitemShut {NoStop}%
\bibitem [{\citenamefont {Oszmaniec}\ and\ \citenamefont
  {Biswas}(2019)}]{Oszmaniec2019operational}%
  \BibitemOpen
  \bibfield  {author} {\bibinfo {author} {\bibfnamefont {M.}~\bibnamefont
  {Oszmaniec}}\ and\ \bibinfo {author} {\bibfnamefont {T.}~\bibnamefont
  {Biswas}},\ }\bibfield  {title} {\emph {\bibinfo {title} {Operational
  relevance of resource theories of quantum measurements},}\ }\href
  {http://dx.doi.org/10.22331/q-2019-04-26-133} {\bibfield  {journal} {\bibinfo
   {journal} {{Quantum}}\ }\textbf {\bibinfo {volume} {3}},\ \bibinfo {pages}
  {133} (\bibinfo {year} {2019})}\BibitemShut {NoStop}%
\bibitem [{\citenamefont {Anshu}\ \emph {et~al.}(2018)\citenamefont {Anshu},
  \citenamefont {Hsieh},\ and\ \citenamefont {Jain}}]{anshu_2017}%
  \BibitemOpen
  \bibfield  {author} {\bibinfo {author} {\bibfnamefont {A.}~\bibnamefont
  {Anshu}}, \bibinfo {author} {\bibfnamefont {M.-H.}\ \bibnamefont {Hsieh}}, \
  and\ \bibinfo {author} {\bibfnamefont {R.}~\bibnamefont {Jain}},\ }\bibfield
  {title} {\emph {\bibinfo {title} {Quantifying resources in general resource
  theory with catalysts},}\ }\href
  {http://dx.doi.org/10.1103/PhysRevLett.121.190504} {\bibfield  {journal}
  {\bibinfo  {journal} {Phys. Rev. Lett.}\ }\textbf {\bibinfo {volume} {121}},\
  \bibinfo {pages} {190504} (\bibinfo {year} {2018})}\BibitemShut {NoStop}%
\bibitem [{\citenamefont {Liu}\ \emph {et~al.}(2019)\citenamefont {Liu},
  \citenamefont {Bu},\ and\ \citenamefont {Takagi}}]{liu2019one}%
  \BibitemOpen
  \bibfield  {author} {\bibinfo {author} {\bibfnamefont {Z.-W.}\ \bibnamefont
  {Liu}}, \bibinfo {author} {\bibfnamefont {K.}~\bibnamefont {Bu}}, \ and\
  \bibinfo {author} {\bibfnamefont {R.}~\bibnamefont {Takagi}},\ }\bibfield
  {title} {\emph {\bibinfo {title} {One-{{Shot Operational Quantum Resource
  Theory}}},}\ }\href {http://dx.doi.org/10.1103/PhysRevLett.123.020401}
  {\bibfield  {journal} {\bibinfo  {journal} {Phys. Rev. Lett.}\ }\textbf
  {\bibinfo {volume} {123}},\ \bibinfo {pages} {020401} (\bibinfo {year}
  {2019})}\BibitemShut {NoStop}%
\bibitem [{\citenamefont {Regula}\ \emph {et~al.}(2020)\citenamefont {Regula},
  \citenamefont {Bu}, \citenamefont {Takagi},\ and\ \citenamefont
  {Liu}}]{Regula2020benchmarking}%
  \BibitemOpen
  \bibfield  {author} {\bibinfo {author} {\bibfnamefont {B.}~\bibnamefont
  {Regula}}, \bibinfo {author} {\bibfnamefont {K.}~\bibnamefont {Bu}}, \bibinfo
  {author} {\bibfnamefont {R.}~\bibnamefont {Takagi}}, \ and\ \bibinfo {author}
  {\bibfnamefont {Z.-W.}\ \bibnamefont {Liu}},\ }\bibfield  {title} {\emph
  {\bibinfo {title} {Benchmarking one-shot distillation in general quantum
  resource theories},}\ }\href {http://dx.doi.org/10.1103/PhysRevA.101.062315}
  {\bibfield  {journal} {\bibinfo  {journal} {Phys. Rev. A}\ }\textbf {\bibinfo
  {volume} {101}},\ \bibinfo {pages} {062315} (\bibinfo {year}
  {2020})}\BibitemShut {NoStop}%
\bibitem [{\citenamefont {Regula}\ and\ \citenamefont
  {Takagi}(2021{\natexlab{b}})}]{Regula2021oneshot}%
  \BibitemOpen
  \bibfield  {author} {\bibinfo {author} {\bibfnamefont {B.}~\bibnamefont
  {Regula}}\ and\ \bibinfo {author} {\bibfnamefont {R.}~\bibnamefont
  {Takagi}},\ }\bibfield  {title} {\emph {\bibinfo {title} {One-shot
  manipulation of dynamical quantum resources},}\ }\href
  {http://dx.doi.org/10.1103/PhysRevLett.127.060402} {\bibfield  {journal}
  {\bibinfo  {journal} {Phys. Rev. Lett.}\ }\textbf {\bibinfo {volume} {127}},\
  \bibinfo {pages} {060402} (\bibinfo {year} {2021}{\natexlab{b}})}\BibitemShut
  {NoStop}%
\bibitem [{\citenamefont {Takagi}\ \emph {et~al.}(2022)\citenamefont {Takagi},
  \citenamefont {Regula},\ and\ \citenamefont {Wilde}}]{Takagi2021oneshot}%
  \BibitemOpen
  \bibfield  {author} {\bibinfo {author} {\bibfnamefont {R.}~\bibnamefont
  {Takagi}}, \bibinfo {author} {\bibfnamefont {B.}~\bibnamefont {Regula}}, \
  and\ \bibinfo {author} {\bibfnamefont {M.~M.}\ \bibnamefont {Wilde}},\
  }\bibfield  {title} {\emph {\bibinfo {title} {One-shot yield-cost relations
  in general quantum resource theories},}\ }\href
  {http://dx.doi.org/10.1103/PRXQuantum.3.010348} {\bibfield  {journal}
  {\bibinfo  {journal} {PRX Quantum}\ }\textbf {\bibinfo {volume} {3}},\
  \bibinfo {pages} {010348} (\bibinfo {year} {2022})}\BibitemShut {NoStop}%
\bibitem [{\citenamefont {Kitaev}(1997)}]{Kitaev1997quantum}%
  \BibitemOpen
  \bibfield  {author} {\bibinfo {author} {\bibfnamefont {A.~Y.}\ \bibnamefont
  {Kitaev}},\ }\bibfield  {title} {\emph {\bibinfo {title} {Quantum
  computations: Algorithms and error correction},}\ }\href
  {http://dx.doi.org/10.1070/RM1997v052n06ABEH002155} {\bibfield  {journal}
  {\bibinfo  {journal} {Russ. Math. Surv.}\ }\textbf {\bibinfo {volume} {52}},\
  \bibinfo {pages} {1191} (\bibinfo {year} {1997})}\BibitemShut {NoStop}%
\bibitem [{\citenamefont {Watrous}(2018)}]{watrous_2018}%
  \BibitemOpen
  \bibfield  {author} {\bibinfo {author} {\bibfnamefont {J.}~\bibnamefont
  {Watrous}},\ }\href@noop {} {\emph {\bibinfo {title} {The {{Theory}} of
  {{Quantum Information}}}}}\ (\bibinfo  {publisher} {{Cambridge University
  Press}},\ \bibinfo {address} {{Cambridge}},\ \bibinfo {year}
  {2018})\BibitemShut {NoStop}%
\bibitem [{\citenamefont {Chiribella}\ \emph
  {et~al.}(2008{\natexlab{c}})\citenamefont {Chiribella}, \citenamefont
  {D'Ariano},\ and\ \citenamefont {Perinotti}}]{Chiribella2008memory}%
  \BibitemOpen
  \bibfield  {author} {\bibinfo {author} {\bibfnamefont {G.}~\bibnamefont
  {Chiribella}}, \bibinfo {author} {\bibfnamefont {G.~M.}\ \bibnamefont
  {D'Ariano}}, \ and\ \bibinfo {author} {\bibfnamefont {P.}~\bibnamefont
  {Perinotti}},\ }\bibfield  {title} {\emph {\bibinfo {title} {Memory effects
  in quantum channel discrimination},}\ }\href
  {http://dx.doi.org/10.1103/PhysRevLett.101.180501} {\bibfield  {journal}
  {\bibinfo  {journal} {Phys. Rev. Lett.}\ }\textbf {\bibinfo {volume} {101}},\
  \bibinfo {pages} {180501} (\bibinfo {year} {2008}{\natexlab{c}})}\BibitemShut
  {NoStop}%
\bibitem [{\citenamefont {Gutoski}(2012)}]{Gutoski2012on_a_measure}%
  \BibitemOpen
  \bibfield  {author} {\bibinfo {author} {\bibfnamefont {G.}~\bibnamefont
  {Gutoski}},\ }\bibfield  {title} {\emph {\bibinfo {title} {On a measure of
  distance for quantum strategies},}\ }\href
  {http://dx.doi.org/10.1063/1.3693621} {\bibfield  {journal} {\bibinfo
  {journal} {J. Math. Phys.}\ }\textbf {\bibinfo {volume} {53}},\ \bibinfo
  {pages} {032202} (\bibinfo {year} {2012})}\BibitemShut {NoStop}%
\bibitem [{\citenamefont {Hoeffding}(1963)}]{Hoeffding1963probability}%
  \BibitemOpen
  \bibfield  {author} {\bibinfo {author} {\bibfnamefont {W.}~\bibnamefont
  {Hoeffding}},\ }\bibfield  {title} {\emph {\bibinfo {title} {Probability
  inequalities for sums of bounded random variables},}\ }\href
  {http://dx.doi.org/10.1080/01621459.1963.10500830} {\bibfield  {journal}
  {\bibinfo  {journal} {J. Am. Stat. Assoc.}\ }\textbf {\bibinfo {volume}
  {58}},\ \bibinfo {pages} {13} (\bibinfo {year} {1963})}\BibitemShut {NoStop}%
\bibitem [{\citenamefont {Regula}\ and\ \citenamefont
  {Lami}(2022)}]{regula_2022-3}%
  \BibitemOpen
  \bibfield  {author} {\bibinfo {author} {\bibfnamefont {B.}~\bibnamefont
  {Regula}}\ and\ \bibinfo {author} {\bibfnamefont {L.}~\bibnamefont {Lami}},\
  }\bibfield  {title} {\emph {\bibinfo {title} {Functional analytic insights
  into irreversibility of quantum resources},}\ }\href@noop {} {\Eprint
  {http://arxiv.org/abs/2211.15678} {arXiv:2211.15678}  (\bibinfo {year}
  {2022})}\BibitemShut {NoStop}%
\bibitem [{\citenamefont {Burniston}\ \emph {et~al.}(2020)\citenamefont
  {Burniston}, \citenamefont {Grabowecky}, \citenamefont {Scandolo},
  \citenamefont {Chiribella},\ and\ \citenamefont
  {Gour}}]{burniston2019necessary}%
  \BibitemOpen
  \bibfield  {author} {\bibinfo {author} {\bibfnamefont {J.}~\bibnamefont
  {Burniston}}, \bibinfo {author} {\bibfnamefont {M.}~\bibnamefont
  {Grabowecky}}, \bibinfo {author} {\bibfnamefont {C.~M.}\ \bibnamefont
  {Scandolo}}, \bibinfo {author} {\bibfnamefont {G.}~\bibnamefont
  {Chiribella}}, \ and\ \bibinfo {author} {\bibfnamefont {G.}~\bibnamefont
  {Gour}},\ }\bibfield  {title} {\emph {\bibinfo {title} {Necessary and
  sufficient conditions on measurements of quantum channels},}\ }\href
  {http://dx.doi.org/10.1098/rspa.2019.0832} {\bibfield  {journal} {\bibinfo
  {journal} {Proc. R. Soc. A}\ }\textbf {\bibinfo {volume} {476}},\ \bibinfo
  {pages} {20190832} (\bibinfo {year} {2020})}\BibitemShut {NoStop}%
\bibitem [{\citenamefont {Vidal}\ and\ \citenamefont
  {Tarrach}(1999)}]{vidal1999robustness}%
  \BibitemOpen
  \bibfield  {author} {\bibinfo {author} {\bibfnamefont {G.}~\bibnamefont
  {Vidal}}\ and\ \bibinfo {author} {\bibfnamefont {R.}~\bibnamefont
  {Tarrach}},\ }\bibfield  {title} {\emph {\bibinfo {title} {Robustness of
  entanglement},}\ }\href {http://dx.doi.org/10.1103/PhysRevA.59.141}
  {\bibfield  {journal} {\bibinfo  {journal} {Phys. Rev. A}\ }\textbf {\bibinfo
  {volume} {59}},\ \bibinfo {pages} {141} (\bibinfo {year} {1999})}\BibitemShut
  {NoStop}%
\bibitem [{\citenamefont {Vidal}\ and\ \citenamefont
  {Werner}(2002)}]{Vidal2002computable}%
  \BibitemOpen
  \bibfield  {author} {\bibinfo {author} {\bibfnamefont {G.}~\bibnamefont
  {Vidal}}\ and\ \bibinfo {author} {\bibfnamefont {R.~F.}\ \bibnamefont
  {Werner}},\ }\bibfield  {title} {\emph {\bibinfo {title} {Computable measure
  of entanglement},}\ }\href {http://dx.doi.org/10.1103/PhysRevA.65.032314}
  {\bibfield  {journal} {\bibinfo  {journal} {Phys. Rev. A}\ }\textbf {\bibinfo
  {volume} {65}},\ \bibinfo {pages} {032314} (\bibinfo {year}
  {2002})}\BibitemShut {NoStop}%
\bibitem [{\citenamefont {Baumgratz}\ \emph {et~al.}(2014)\citenamefont
  {Baumgratz}, \citenamefont {Cramer},\ and\ \citenamefont
  {Plenio}}]{Baumgratz14}%
  \BibitemOpen
  \bibfield  {author} {\bibinfo {author} {\bibfnamefont {T.}~\bibnamefont
  {Baumgratz}}, \bibinfo {author} {\bibfnamefont {M.}~\bibnamefont {Cramer}}, \
  and\ \bibinfo {author} {\bibfnamefont {M.~B.}\ \bibnamefont {Plenio}},\
  }\bibfield  {title} {\emph {\bibinfo {title} {Quantifying coherence},}\
  }\href {http://dx.doi.org/10.1103/PhysRevLett.113.140401} {\bibfield
  {journal} {\bibinfo  {journal} {Phys. Rev. Lett.}\ }\textbf {\bibinfo
  {volume} {113}},\ \bibinfo {pages} {140401} (\bibinfo {year}
  {2014})}\BibitemShut {NoStop}%
\bibitem [{\citenamefont {Steiner}(2003)}]{PhysRevA.67.054305}%
  \BibitemOpen
  \bibfield  {author} {\bibinfo {author} {\bibfnamefont {M.}~\bibnamefont
  {Steiner}},\ }\bibfield  {title} {\emph {\bibinfo {title} {Generalized
  robustness of entanglement},}\ }\href
  {http://dx.doi.org/10.1103/PhysRevA.67.054305} {\bibfield  {journal}
  {\bibinfo  {journal} {Phys. Rev. A}\ }\textbf {\bibinfo {volume} {67}},\
  \bibinfo {pages} {054305} (\bibinfo {year} {2003})}\BibitemShut {NoStop}%
\bibitem [{\citenamefont {Harrow}\ and\ \citenamefont
  {Nielsen}(2003)}]{Harrow2003robustness}%
  \BibitemOpen
  \bibfield  {author} {\bibinfo {author} {\bibfnamefont {A.~W.}\ \bibnamefont
  {Harrow}}\ and\ \bibinfo {author} {\bibfnamefont {M.~A.}\ \bibnamefont
  {Nielsen}},\ }\bibfield  {title} {\emph {\bibinfo {title} {Robustness of
  quantum gates in the presence of noise},}\ }\href
  {http://dx.doi.org/10.1103/PhysRevA.68.012308} {\bibfield  {journal}
  {\bibinfo  {journal} {Phys. Rev. A}\ }\textbf {\bibinfo {volume} {68}},\
  \bibinfo {pages} {012308} (\bibinfo {year} {2003})}\BibitemShut {NoStop}%
\bibitem [{\citenamefont {Johnston}\ \emph {et~al.}(2018)\citenamefont
  {Johnston}, \citenamefont {Li}, \citenamefont {Plosker}, \citenamefont
  {Poon},\ and\ \citenamefont {Regula}}]{Johnston2018evaluating}%
  \BibitemOpen
  \bibfield  {author} {\bibinfo {author} {\bibfnamefont {N.}~\bibnamefont
  {Johnston}}, \bibinfo {author} {\bibfnamefont {C.-K.}\ \bibnamefont {Li}},
  \bibinfo {author} {\bibfnamefont {S.}~\bibnamefont {Plosker}}, \bibinfo
  {author} {\bibfnamefont {Y.-T.}\ \bibnamefont {Poon}}, \ and\ \bibinfo
  {author} {\bibfnamefont {B.}~\bibnamefont {Regula}},\ }\bibfield  {title}
  {\emph {\bibinfo {title} {Evaluating the robustness of $k$-coherence and
  $k$-entanglement},}\ }\href {http://dx.doi.org/10.1103/PhysRevA.98.022328}
  {\bibfield  {journal} {\bibinfo  {journal} {Phys. Rev. A}\ }\textbf {\bibinfo
  {volume} {98}},\ \bibinfo {pages} {022328} (\bibinfo {year}
  {2018})}\BibitemShut {NoStop}%
\bibitem [{\citenamefont {Datta}(2009)}]{Datta2009}%
  \BibitemOpen
  \bibfield  {author} {\bibinfo {author} {\bibfnamefont {N.}~\bibnamefont
  {Datta}},\ }\bibfield  {title} {\emph {\bibinfo {title} {Min- and
  max-relative entropies and a new entanglement monotone},}\ }\href
  {http://dx.doi.org/10.1109/TIT.2009.2018325} {\bibfield  {journal} {\bibinfo
  {journal} {IEEE Trans. Inf. Theor.}\ }\textbf {\bibinfo {volume} {55}},\
  \bibinfo {pages} {2816} (\bibinfo {year} {2009})}\BibitemShut {NoStop}%
\bibitem [{\citenamefont {Regula}\ \emph {et~al.}(2019)\citenamefont {Regula},
  \citenamefont {Fang}, \citenamefont {Wang},\ and\ \citenamefont
  {Gu}}]{regula_2019-2}%
  \BibitemOpen
  \bibfield  {author} {\bibinfo {author} {\bibfnamefont {B.}~\bibnamefont
  {Regula}}, \bibinfo {author} {\bibfnamefont {K.}~\bibnamefont {Fang}},
  \bibinfo {author} {\bibfnamefont {X.}~\bibnamefont {Wang}}, \ and\ \bibinfo
  {author} {\bibfnamefont {M.}~\bibnamefont {Gu}},\ }\bibfield  {title} {\emph
  {\bibinfo {title} {One-shot entanglement distillation beyond local operations
  and classical communication},}\ }\href
  {http://dx.doi.org/10.1088/1367-2630/ab4732} {\bibfield  {journal} {\bibinfo
  {journal} {New J. Phys.}\ }\textbf {\bibinfo {volume} {21}},\ \bibinfo
  {pages} {103017} (\bibinfo {year} {2019})}\BibitemShut {NoStop}%
\bibitem [{\citenamefont {Horodecki}\ and\ \citenamefont
  {Horodecki}(1999)}]{Horodecki1999reduction}%
  \BibitemOpen
  \bibfield  {author} {\bibinfo {author} {\bibfnamefont {M.}~\bibnamefont
  {Horodecki}}\ and\ \bibinfo {author} {\bibfnamefont {P.}~\bibnamefont
  {Horodecki}},\ }\bibfield  {title} {\emph {\bibinfo {title} {Reduction
  criterion of separability and limits for a class of distillation
  protocols},}\ }\href {http://dx.doi.org/10.1103/PhysRevA.59.4206} {\bibfield
  {journal} {\bibinfo  {journal} {Phys. Rev. A}\ }\textbf {\bibinfo {volume}
  {59}},\ \bibinfo {pages} {4206} (\bibinfo {year} {1999})}\BibitemShut
  {NoStop}%
\bibitem [{\citenamefont {Shimony}(1995)}]{shimony_1995}%
  \BibitemOpen
  \bibfield  {author} {\bibinfo {author} {\bibfnamefont {A.}~\bibnamefont
  {Shimony}},\ }\bibfield  {title} {\emph {\bibinfo {title} {Degree of
  {{Entanglement}}},}\ }\href
  {http://dx.doi.org/10.1111/j.1749-6632.1995.tb39008.x} {\bibfield  {journal}
  {\bibinfo  {journal} {Ann. NY Ac.}\ }\textbf {\bibinfo {volume} {755}},\
  \bibinfo {pages} {675} (\bibinfo {year} {1995})}\BibitemShut {NoStop}%
\bibitem [{\citenamefont {Veitch}\ \emph {et~al.}(2014)\citenamefont {Veitch},
  \citenamefont {Mousavian}, \citenamefont {Gottesman},\ and\ \citenamefont
  {Emerson}}]{veitch2014resource}%
  \BibitemOpen
  \bibfield  {author} {\bibinfo {author} {\bibfnamefont {V.}~\bibnamefont
  {Veitch}}, \bibinfo {author} {\bibfnamefont {S.~H.}\ \bibnamefont
  {Mousavian}}, \bibinfo {author} {\bibfnamefont {D.}~\bibnamefont
  {Gottesman}}, \ and\ \bibinfo {author} {\bibfnamefont {J.}~\bibnamefont
  {Emerson}},\ }\bibfield  {title} {\emph {\bibinfo {title} {The resource
  theory of stabilizer quantum computation},}\ }\href
  {http://dx.doi.org/10.1088/1367-2630/16/1/013009} {\bibfield  {journal}
  {\bibinfo  {journal} {New J. Phys.}\ }\textbf {\bibinfo {volume} {16}},\
  \bibinfo {pages} {013009} (\bibinfo {year} {2014})}\BibitemShut {NoStop}%
\bibitem [{\citenamefont {Howard}\ and\ \citenamefont
  {Campbell}(2017)}]{PhysRevLett.118.090501}%
  \BibitemOpen
  \bibfield  {author} {\bibinfo {author} {\bibfnamefont {M.}~\bibnamefont
  {Howard}}\ and\ \bibinfo {author} {\bibfnamefont {E.}~\bibnamefont
  {Campbell}},\ }\bibfield  {title} {\emph {\bibinfo {title} {Application of a
  resource theory for magic states to fault-tolerant quantum computing},}\
  }\href {http://dx.doi.org/10.1103/PhysRevLett.118.090501} {\bibfield
  {journal} {\bibinfo  {journal} {Phys. Rev. Lett.}\ }\textbf {\bibinfo
  {volume} {118}},\ \bibinfo {pages} {090501} (\bibinfo {year}
  {2017})}\BibitemShut {NoStop}%
\bibitem [{\citenamefont {Lostaglio}\ and\ \citenamefont
  {Ciani}(2021)}]{Lostaglio2021error}%
  \BibitemOpen
  \bibfield  {author} {\bibinfo {author} {\bibfnamefont {M.}~\bibnamefont
  {Lostaglio}}\ and\ \bibinfo {author} {\bibfnamefont {A.}~\bibnamefont
  {Ciani}},\ }\bibfield  {title} {\emph {\bibinfo {title} {Error mitigation and
  quantum-assisted simulation in the error corrected regime},}\ }\href
  {http://dx.doi.org/10.1103/PhysRevLett.127.200506} {\bibfield  {journal}
  {\bibinfo  {journal} {Phys. Rev. Lett.}\ }\textbf {\bibinfo {volume} {127}},\
  \bibinfo {pages} {200506} (\bibinfo {year} {2021})}\BibitemShut {NoStop}%
\bibitem [{\citenamefont {Piveteau}\ \emph {et~al.}(2021)\citenamefont
  {Piveteau}, \citenamefont {Sutter}, \citenamefont {Bravyi}, \citenamefont
  {Gambetta},\ and\ \citenamefont {Temme}}]{Piveteau2021error}%
  \BibitemOpen
  \bibfield  {author} {\bibinfo {author} {\bibfnamefont {C.}~\bibnamefont
  {Piveteau}}, \bibinfo {author} {\bibfnamefont {D.}~\bibnamefont {Sutter}},
  \bibinfo {author} {\bibfnamefont {S.}~\bibnamefont {Bravyi}}, \bibinfo
  {author} {\bibfnamefont {J.~M.}\ \bibnamefont {Gambetta}}, \ and\ \bibinfo
  {author} {\bibfnamefont {K.}~\bibnamefont {Temme}},\ }\bibfield  {title}
  {\emph {\bibinfo {title} {Error mitigation for universal gates on encoded
  qubits},}\ }\href {http://dx.doi.org/10.1103/PhysRevLett.127.200505}
  {\bibfield  {journal} {\bibinfo  {journal} {Phys. Rev. Lett.}\ }\textbf
  {\bibinfo {volume} {127}},\ \bibinfo {pages} {200505} (\bibinfo {year}
  {2021})}\BibitemShut {NoStop}%
\bibitem [{\citenamefont {Suzuki}\ \emph {et~al.}(2022)\citenamefont {Suzuki},
  \citenamefont {Endo}, \citenamefont {Fujii},\ and\ \citenamefont
  {Tokunaga}}]{Suzuki2021quantum}%
  \BibitemOpen
  \bibfield  {author} {\bibinfo {author} {\bibfnamefont {Y.}~\bibnamefont
  {Suzuki}}, \bibinfo {author} {\bibfnamefont {S.}~\bibnamefont {Endo}},
  \bibinfo {author} {\bibfnamefont {K.}~\bibnamefont {Fujii}}, \ and\ \bibinfo
  {author} {\bibfnamefont {Y.}~\bibnamefont {Tokunaga}},\ }\bibfield  {title}
  {\emph {\bibinfo {title} {Quantum error mitigation as a universal error
  reduction technique: Applications from the {NISQ} to the fault-tolerant
  quantum computing eras},}\ }\href
  {http://dx.doi.org/10.1103/PRXQuantum.3.010345} {\bibfield  {journal}
  {\bibinfo  {journal} {PRX Quantum}\ }\textbf {\bibinfo {volume} {3}},\
  \bibinfo {pages} {010345} (\bibinfo {year} {2022})}\BibitemShut {NoStop}%
\bibitem [{\citenamefont {Veitch}\ \emph {et~al.}(2012)\citenamefont {Veitch},
  \citenamefont {Ferrie}, \citenamefont {Gross},\ and\ \citenamefont
  {Emerson}}]{veitch2012negative}%
  \BibitemOpen
  \bibfield  {author} {\bibinfo {author} {\bibfnamefont {V.}~\bibnamefont
  {Veitch}}, \bibinfo {author} {\bibfnamefont {C.}~\bibnamefont {Ferrie}},
  \bibinfo {author} {\bibfnamefont {D.}~\bibnamefont {Gross}}, \ and\ \bibinfo
  {author} {\bibfnamefont {J.}~\bibnamefont {Emerson}},\ }\bibfield  {title}
  {\emph {\bibinfo {title} {Negative quasi-probability as a resource for
  quantum computation},}\ }\href
  {http://dx.doi.org/10.1088/1367-2630/14/11/113011} {\bibfield  {journal}
  {\bibinfo  {journal} {New J. Phys.}\ }\textbf {\bibinfo {volume} {14}},\
  \bibinfo {pages} {113011} (\bibinfo {year} {2012})}\BibitemShut {NoStop}%
\bibitem [{\citenamefont {Li}\ and\ \citenamefont {Benjamin}(2017)}]{Li2017}%
  \BibitemOpen
  \bibfield  {author} {\bibinfo {author} {\bibfnamefont {Y.}~\bibnamefont
  {Li}}\ and\ \bibinfo {author} {\bibfnamefont {S.~C.}\ \bibnamefont
  {Benjamin}},\ }\bibfield  {title} {\emph {\bibinfo {title} {Efficient
  variational quantum simulator incorporating active error minimization},}\
  }\href {http://dx.doi.org/10.1103/PhysRevX.7.021050} {\bibfield  {journal}
  {\bibinfo  {journal} {Phys. Rev. X}\ }\textbf {\bibinfo {volume} {7}},\
  \bibinfo {pages} {021050} (\bibinfo {year} {2017})}\BibitemShut {NoStop}%
\bibitem [{\citenamefont {Watrous}(2009)}]{watrous2009semidefinite}%
  \BibitemOpen
  \bibfield  {author} {\bibinfo {author} {\bibfnamefont {J.}~\bibnamefont
  {Watrous}},\ }\bibfield  {title} {\emph {\bibinfo {title} {Semidefinite
  programs for completely bounded norms},}\ }\href
  {https://www.semanticscholar.org/paper/Semidefinite-Programs-for-Completely-Bounded-Norms-Watrous/1922564531530868ce3d5ffe71f5a97ba126c5de}
  {\bibfield  {journal} {\bibinfo  {journal} {Theory Comput.}\ }\textbf
  {\bibinfo {volume} {5}},\ \bibinfo {pages} {217} (\bibinfo {year}
  {2009})}\BibitemShut {NoStop}%
\bibitem [{\citenamefont {Bennett}\ \emph {et~al.}(2002)\citenamefont
  {Bennett}, \citenamefont {Shor}, \citenamefont {Smolin},\ and\ \citenamefont
  {Thapliyal}}]{bennett_2002}%
  \BibitemOpen
  \bibfield  {author} {\bibinfo {author} {\bibfnamefont {C.~H.}\ \bibnamefont
  {Bennett}}, \bibinfo {author} {\bibfnamefont {P.~W.}\ \bibnamefont {Shor}},
  \bibinfo {author} {\bibfnamefont {J.~A.}\ \bibnamefont {Smolin}}, \ and\
  \bibinfo {author} {\bibfnamefont {A.~V.}\ \bibnamefont {Thapliyal}},\
  }\bibfield  {title} {\emph {\bibinfo {title} {Entanglement-assisted capacity
  of a quantum channel and the reverse {{Shannon}} theorem},}\ }\href
  {http://dx.doi.org/10.1109/TIT.2002.802612} {\bibfield  {journal} {\bibinfo
  {journal} {IEEE Trans. Inf. Theory}\ }\textbf {\bibinfo {volume} {48}},\
  \bibinfo {pages} {2637} (\bibinfo {year} {2002})}\BibitemShut {NoStop}%
\bibitem [{\citenamefont {Leung}\ and\ \citenamefont
  {Matthews}(2015)}]{leung_2015}%
  \BibitemOpen
  \bibfield  {author} {\bibinfo {author} {\bibfnamefont {D.}~\bibnamefont
  {Leung}}\ and\ \bibinfo {author} {\bibfnamefont {W.}~\bibnamefont
  {Matthews}},\ }\bibfield  {title} {\emph {\bibinfo {title} {On the {{Power}}
  of {{PPT-Preserving}} and {{Non-Signalling Codes}}},}\ }\href
  {http://dx.doi.org/10.1109/TIT.2015.2439953} {\bibfield  {journal} {\bibinfo
  {journal} {IEEE Trans. Inf. Theory}\ }\textbf {\bibinfo {volume} {61}},\
  \bibinfo {pages} {4486} (\bibinfo {year} {2015})}\BibitemShut {NoStop}%
\bibitem [{\citenamefont {Bennett}\ \emph {et~al.}(1996)\citenamefont
  {Bennett}, \citenamefont {DiVincenzo}, \citenamefont {Smolin},\ and\
  \citenamefont {Wootters}}]{Bennett96}%
  \BibitemOpen
  \bibfield  {author} {\bibinfo {author} {\bibfnamefont {C.~H.}\ \bibnamefont
  {Bennett}}, \bibinfo {author} {\bibfnamefont {D.~P.}\ \bibnamefont
  {DiVincenzo}}, \bibinfo {author} {\bibfnamefont {J.~A.}\ \bibnamefont
  {Smolin}}, \ and\ \bibinfo {author} {\bibfnamefont {W.~K.}\ \bibnamefont
  {Wootters}},\ }\bibfield  {title} {\emph {\bibinfo {title} {Mixed-state
  entanglement and quantum error correction},}\ }\href
  {http://dx.doi.org/10.1103/PhysRevA.54.3824} {\bibfield  {journal} {\bibinfo
  {journal} {Phys. Rev. A}\ }\textbf {\bibinfo {volume} {54}},\ \bibinfo
  {pages} {3824} (\bibinfo {year} {1996})}\BibitemShut {NoStop}%
\bibitem [{\citenamefont {Smith}\ and\ \citenamefont
  {Smolin}(2008)}]{Smith2008additive}%
  \BibitemOpen
  \bibfield  {author} {\bibinfo {author} {\bibfnamefont {G.}~\bibnamefont
  {Smith}}\ and\ \bibinfo {author} {\bibfnamefont {J.~A.}\ \bibnamefont
  {Smolin}},\ }\bibfield  {title} {\emph {\bibinfo {title} {Additive extensions
  of a quantum channel},}\ }in\ \href
  {http://dx.doi.org/10.1109/ITW.2008.4578688} {\emph {\bibinfo {booktitle}
  {2008 IEEE Information Theory Workshop}}}\ (\bibinfo {year} {2008})\ pp.\
  \bibinfo {pages} {368--372}\BibitemShut {NoStop}%
\bibitem [{\citenamefont {Giovannetti}\ and\ \citenamefont
  {Fazio}(2005)}]{Giovannetti2005information}%
  \BibitemOpen
  \bibfield  {author} {\bibinfo {author} {\bibfnamefont {V.}~\bibnamefont
  {Giovannetti}}\ and\ \bibinfo {author} {\bibfnamefont {R.}~\bibnamefont
  {Fazio}},\ }\bibfield  {title} {\emph {\bibinfo {title} {Information-capacity
  description of spin-chain correlations},}\ }\href
  {http://dx.doi.org/10.1103/PhysRevA.71.032314} {\bibfield  {journal}
  {\bibinfo  {journal} {Phys. Rev. A}\ }\textbf {\bibinfo {volume} {71}},\
  \bibinfo {pages} {032314} (\bibinfo {year} {2005})}\BibitemShut {NoStop}%
\bibitem [{\citenamefont {Hakoshima}\ \emph {et~al.}(2021)\citenamefont
  {Hakoshima}, \citenamefont {Matsuzaki},\ and\ \citenamefont
  {Endo}}]{Hokoshima2021relationship}%
  \BibitemOpen
  \bibfield  {author} {\bibinfo {author} {\bibfnamefont {H.}~\bibnamefont
  {Hakoshima}}, \bibinfo {author} {\bibfnamefont {Y.}~\bibnamefont
  {Matsuzaki}}, \ and\ \bibinfo {author} {\bibfnamefont {S.}~\bibnamefont
  {Endo}},\ }\bibfield  {title} {\emph {\bibinfo {title} {Relationship between
  costs for quantum error mitigation and non-markovian measures},}\ }\href
  {http://dx.doi.org/10.1103/PhysRevA.103.012611} {\bibfield  {journal}
  {\bibinfo  {journal} {Phys. Rev. A}\ }\textbf {\bibinfo {volume} {103}},\
  \bibinfo {pages} {012611} (\bibinfo {year} {2021})}\BibitemShut {NoStop}%
\bibitem [{\citenamefont {{Ho}}\ \emph {et~al.}(2021)\citenamefont {{Ho}},
  \citenamefont {{Takagi}},\ and\ \citenamefont {{Gu}}}]{Ho2021enhancing}%
  \BibitemOpen
  \bibfield  {author} {\bibinfo {author} {\bibfnamefont {M.}~\bibnamefont
  {{Ho}}}, \bibinfo {author} {\bibfnamefont {R.}~\bibnamefont {{Takagi}}}, \
  and\ \bibinfo {author} {\bibfnamefont {M.}~\bibnamefont {{Gu}}},\ }\bibfield
  {title} {\emph {\bibinfo {title} {{Enhancing quantum models of stochastic
  processes with error mitigation}},}\ }\href@noop {} {\Eprint
  {http://arxiv.org/abs/2105.06448} {arXiv:2105.06448}  (\bibinfo {year}
  {2021})}\BibitemShut {NoStop}%
\bibitem [{\citenamefont {Boyd}\ and\ \citenamefont
  {Vandenberghe}(2004)}]{Boyd2004convex}%
  \BibitemOpen
  \bibfield  {author} {\bibinfo {author} {\bibfnamefont {S.}~\bibnamefont
  {Boyd}}\ and\ \bibinfo {author} {\bibfnamefont {L.}~\bibnamefont
  {Vandenberghe}},\ }\href@noop {} {\emph {\bibinfo {title} {Convex
  {{Optimization}}}}}\ (\bibinfo  {publisher} {{Cambridge University Press}},\
  \bibinfo {address} {{New York}},\ \bibinfo {year} {2004})\BibitemShut
  {NoStop}%
\end{thebibliography}%

\end{document}